\documentclass[12pt]{article}

\usepackage{amsmath}       
\usepackage{amsthm}        
\usepackage{amssymb}       
\usepackage{amsfonts}
\usepackage[all]{xy}
\usepackage{xspace}
\usepackage{calc}
\usepackage{comment}
\usepackage{pgfplots}
\usepgflibrary{fpu}
\usepackage{mathtools}
\usepackage{tabularx}
\usepackage{ifthen}             
\usepackage{graphicx}
\usepackage{docmute}
\usepackage{array}
\usepackage{subfloat} 
\usepackage{textcomp}
\usepackage{bbm}            
\usepackage{etoolbox}  
\usepackage{verbatim}
\usepackage{stmaryrd}
\usepackage{cancel}
\usepackage{xcolor}
\colorlet{Black}{black}
\usepackage{tensor}
\usepackage{enumerate}  
\usepackage{thm-restate}
\usepackage[numbers,sort]{natbib}
\usepackage{enumitem}
\usepackage{tocloft}

\usepackage[margin=3cm]{geometry}

\newenvironment{tz}[1][]{%
                                \begin{tikzpicture}[baseline={([yshift=-.8ex]current bounding                        box.center)},#1] %
                                }{%
                        \end{tikzpicture} %
                        }

\DeclareRobustCommand{\SkipTocEntry}[5]{}

\usepackage{tikz}
\usetikzlibrary{matrix}
\usetikzlibrary{decorations.pathreplacing}
\usetikzlibrary{decorations.markings}
\usetikzlibrary{arrows}
\usetikzlibrary{calc}
\usetikzlibrary{shapes.misc}
\usetikzlibrary{fit}
\usepgflibrary{decorations.pathmorphing}
\usepgflibrary{shapes.geometric}

\pgfdeclarelayer{edgelayer}
\pgfdeclarelayer{nodelayer}
\pgfdeclarelayer{foreground}
\pgfsetlayers{edgelayer,nodelayer,main,foreground}

\tikzstyle{none}=[inner sep=0pt]

\tikzstyle{rn}=[circle,fill=Red,draw=Black,line width=0.8 pt]
\tikzstyle{gn}=[circle,fill=Lime,draw=Black,line width=0.8 pt]
\tikzstyle{bl}=[circle,fill=Blue,draw=Black,line width=0.8 pt]

\tikzstyle{simple}=[-,draw=Black,thick]
\tikzstyle{arrow}=[-,draw=Black,postaction={decorate},decoration={markings,mark=at position .5 with {\arrow{>}}},thick]
\tikzstyle{tick}=[-,draw=Black,postaction={decorate},decoration={markings,mark=at position .5 with {\draw (0,-0.1) -- (0,0.1);}},line width=2.000]
\makeatother\def\thickness{0.7pt}
\tikzstyle{dot}=[circle, draw=black, fill=black, inner sep=.5ex, line width=\thickness, node on layer=foreground]

\makeatletter
\pgfkeys{%
  /tikz/on layer/.code={
    \pgfonlayer{#1}\begingroup
    \aftergroup\endpgfonlayer
    \aftergroup\endgroup
  },
  /tikz/node on layer/.code={
     \gdef\node@@on@layer{%
      \setbox\tikz@tempbox=\hbox\bgroup\pgfonlayer{#1}\unhbox\tikz@tempbox\endpgfonlayer\egroup}
    \aftergroup\node@on@layer
  },
  /tikz/end node on layer/.code={
    \endpgfonlayer\endgroup\endgroup
  }
}

\def\node@on@layer{\aftergroup\node@@on@layer}

\makeatletter
\def\calign@preamble{%
   &\hfil\strut@
    \setboxz@h{\@lign$\m@th\displaystyle{##}$}%
    \ifmeasuring@\savefieldlength@\fi
    \set@field
    \hfil
    \tabskip\alignsep@
}
\let\cmeasure@\measure@
\patchcmd\cmeasure@{\divide\@tempcntb\tw@}{}{}{}
\patchcmd\cmeasure@{\divide\@tempcntb\tw@}{}{}{}
\patchcmd\cmeasure@{\ifodd\maxfields@
  \global\advance\maxfields@\@ne
  \fi}{}{}{}    
\newenvironment{calign}
{%
  \let\align@preamble\calign@preamble
  \let\measure@\cmeasure@
  \align
}
{%
  \endalign
}  
\makeatother
\tikzset{
    master/.style={
        execute at end picture={
            \coordinate (lower right) at (current bounding box.south east);
            \coordinate (upper left) at (current bounding box.north west);
        }
    },
    slave/.style={
        execute at end picture={
            \pgfresetboundingbox
            \path (upper left) rectangle (lower right);
        }
    }
}

\tikzset{blob/.style={draw, circle, fill=white, inner sep=1pt, minimum width=15pt, font=\scriptsize, line width=0.7pt}}
\tikzset{greenregion/.style={fill=green, fill opacity=0.3, draw=none}}
\tikzset{redregion/.style={fill=red, fill opacity=0.3, draw=none}}
\tikzset{blueregion/.style={fill=blue, fill opacity=0.3, draw=none}}
\tikzset{yellowregion/.style={fill=yellow, fill opacity=0.5, draw=none}}
\tikzset{cyanregion/.style={fill=cyan, fill opacity=0.3, draw=none}}
\tikzset{orangeregion/.style={fill=orange, fill opacity=0.6, draw=none}}
\tikzset{solidgreenregion/.style={fill=green!30, fill opacity=1, draw=none}}
\tikzset{solidredregion/.style={fill=red!30, fill opacity=1, draw=none}}
\tikzset{solidblueregion/.style={fill=blue!30, fill opacity=1, draw=none}}
\tikzset{solidyellowregion/.style={fill=yellow!30, fill opacity=1, draw=none}}
\tikzset{string/.style={line width=0.7pt}}
\tikzset{zig/.style={decoration={zigzag,segment length=3, amplitude=0.5}}}
\tikzset{bnd/.style={draw,string}}   
\tikzset{projector/.style={circle, draw, font=\scriptsize, inner sep=-5pt, minimum width=0.35cm, string, fill=white}}
\tikzset{dimension/.style={font=\scriptsize, inner sep=1pt}}

\tikzset{arrow data/.style 2 args={
      decoration={
         markings,
         mark=at position #1 with \arrow{#2}},
         postaction=decorate}
}

\tikzset{along path/.style={every path/.style={}, sloped, allow upside down}}

\def\zxnormal {
                \def \zxscale{0.55}
                \def\zxnodescale{0.8}
                \def\vertexscale{0.7}
                \def\zxshift{0.075cm}
                \def\hadscale{0.8}
                \def\trianglescale{1}
                \def\boxscale{1}
                }

\def\zxgreen{white}

\def\zxwhite{white}
\def\zxblack{black!50}

\tikzset{front/.style ={node on layer=foreground}}
\tikzset{zx/.style = {string, scale=\zxscale}}
\tikzset{zxnode/.style n args={1}{blob,scale=\zxnodescale,fill=#1,node on layer=foreground}}
\tikzset{box/.style={draw, rectangle, fill=white, inner sep=1pt, minimum width=10pt,minimum height=10pt, font=\scriptsize, line width=0.7pt,scale=\zxnodescale,node on layer=foreground}}
\tikzset{boxvertex/.style={draw, rectangle, fill=white, line width=0.733pt,scale=0.75*\vertexscale}}
\tikzset{bigbox/.style={draw, rectangle, fill=white,  minimum width=\boxscale *18pt,minimum height=\boxscale*8pt, line width=0.7pt,scale=\zxnodescale}}
\newlength{\unitbox}
\tikzset{widebox/.style ={draw,rectangle, fill=white, line width=0.7pt,scale=0.75*\zxnodescale,minimum height=15pt,inner sep=1pt,  minimum width = \unitbox,   anchor=center }}
\tikzset{wideboxm/.style n args={1}{draw,rectangle, fill=white, line width=0.7pt,scale=0.75*\zxnodescale,minimum height=15pt,inner sep=1pt,  minimum width =2\unitbox+#1\unitbox,   anchor=center }}
\tikzset{triangleup/.style n args={1}{draw, shape=isosceles triangle, isosceles triangle stretches, fill=white, line width=0.7pt,scale=0.75*\zxnodescale,minimum height=15pt,inner sep=1pt,  minimum width = #1*\trianglescale cm +0.15*\trianglescale cm,  shape border rotate=90, anchor=south }}
\tikzset{triangledown/.style n args={1}{draw, shape=isosceles triangle, isosceles triangle stretches, fill=white, line width=0.7pt,scale=0.75*\zxnodescale,minimum height=15pt,inner sep=1pt,  minimum width = #1*\trianglescale cm +0.15*\trianglescale cm,  shape border rotate=-90, anchor=north }}
\tikzset{zxvertex/.style n args={1}{draw,fill=#1,circle,line width=0.7pt,scale=0.75*\vertexscale}}
\tikzset{zxdown/.style={yshift=-\zxshift}}
\tikzset{zxup/.style={yshift=\zxshift}}

\newcommand\mult[3]{ 
\draw[string] (#1.center) to [out=up, in=-135] +(0.5*#2,#3) to [out=-45, in=up] +(0.5*#2,-#3);
\node[zxvertex=\zxgreen,zxdown] at ($(#1)+(0.5*#2,#3)$){};
}
\newcommand\comult[3]{ 
\draw[string] (#1.center) to [out=down, in=135] +(0.5*#2,-#3) to [out=45, in=down] +(0.5*#2,#3);
\node[zxvertex=\zxgreen,zxup] at ($(#1) +(0.5*#2,-#3)$){};}

\newcommand\unit[2]{ 
\draw[string] (#1.center) to + (0, -#2);
\node[zxvertex=\zxgreen] at ($(#1) +(0,-#2)$){};
}

\newcommand\emptydiagram{
\begin{tz}[zx]
\draw[dotted] (-0.5,-0.5) rectangle (0.5,0.5);
\end{tz}
}

\newcommand{\Tr}{\mathrm{Tr}}

\newcommand{\minus}{\ensuremath{\text{-}}}
\newcommand\superequals[1]{\stackrel {\makebox[0pt]{\tiny #1}} =}
\newcommand\super[2]{\stackrel{\makebox[0pt]{\tiny #1}} #2}
\newcommand\superequalseq[1]{\stackrel {\makebox[0pt]{\tiny\eqref{#1}}} =}
\newcommand{\ket}[1]{\left|#1\right\rangle}

\newcommand{\quotient}[2]{\left.\raisebox{.1em}{$#1$}\middle/\raisebox{-.1em}{$#2$}\right.}

\renewcommand{\to}[1][]{\ensuremath{\xrightarrow{#1}}}
\newcommand{\To}[1][]{\ensuremath{\xRightarrow{#1}}}

\allowdisplaybreaks[1]

\theoremstyle{plain} 

\newtheorem{theorem}{Theorem}[section]

\newtheorem{corollary}[theorem]{Corollary}          
\newtheorem{proposition}[theorem]{Proposition}              
          
\newtheorem{prop}[theorem]{Proposition}

\newtheorem{res}{Result}
\newtheorem{app}{Application}

\theoremstyle{definition} 
\newtheorem{definition}[theorem]{Definition}
\newtheorem{remark}[theorem]{Remark}
\newtheorem{notation}[theorem]{Notation}

\newtheorem{terminology}[theorem]{Terminology}

\theoremstyle{remark}  
\newtheorem{example}[theorem]{Example}

\newtheoremstyle{special_statement} 
        {\topskip}
        {\topskip}
        {\addtolength{\leftskip}{2.5em} \itshape }
        {}
        {\bfseries}
        {:}
        {.5em}
        {}
\theoremstyle{special_statement}

\DeclareMathOperator{\Hom}{Hom}
\DeclareMathOperator{\End}{End}

\newcommand{\Aut}{\ensuremath{\mathrm{Aut}}}

\newcommand{\Vect}{\mathrm{Vect}}

\newcommand{\Rep}{\mathrm{Rep}}
\newcommand{\Mat}{\mathrm{Mat}}
\newcommand{\act}{\vartriangleright}

\newcommand{\Hilb}{\ensuremath{\mathrm{Hilb}}}
\newcommand{\Set}{\ensuremath{\mathrm{Set}}}

\newcommand\conj[1]{\overline{#1}}
\newcommand\inv[1]{#1^{\minus 1}}

\newcommand\Stab{\mathrm{Stab}}

\newcommand{\QAut}{\ensuremath{\mathrm{QAut}}}

\newcommand{\QGraph}{\ensuremath{\mathrm{QGraph}}}

\newcommand{\QGraphIso}{\ensuremath{\mathrm{QGraphIso}}}

\makeatletter
\DeclareFontFamily{OMX}{MnSymbolE}{}
\DeclareSymbolFont{MnLargeSymbols}{OMX}{MnSymbolE}{m}{n}
\SetSymbolFont{MnLargeSymbols}{bold}{OMX}{MnSymbolE}{b}{n}
\DeclareFontShape{OMX}{MnSymbolE}{m}{n}{
    <-6>  MnSymbolE5
   <6-7>  MnSymbolE6
   <7-8>  MnSymbolE7
   <8-9>  MnSymbolE8
   <9-10> MnSymbolE9
  <10-12> MnSymbolE10
  <12->   MnSymbolE12
}{}
\DeclareFontShape{OMX}{MnSymbolE}{b}{n}{
    <-6>  MnSymbolE-Bold5
   <6-7>  MnSymbolE-Bold6
   <7-8>  MnSymbolE-Bold7
   <8-9>  MnSymbolE-Bold8
   <9-10> MnSymbolE-Bold9
  <10-12> MnSymbolE-Bold10
  <12->   MnSymbolE-Bold12
}{}

\let\llangle\@undefined
\let\rrangle\@undefined
\DeclareMathDelimiter{\llangle}{\mathopen}%
                     {MnLargeSymbols}{'164}{MnLargeSymbols}{'164}
\DeclareMathDelimiter{\rrangle}{\mathclose}%
                     {MnLargeSymbols}{'171}{MnLargeSymbols}{'171}
\makeatother

\usepackage[colorlinks=true, linkcolor=black, citecolor=black, urlcolor=black, hypertexnames=false
        ]{hyperref}

\newcounter{DRcomment}
\newcommand\DR[1]{\ensuremath{{}^{\color{red}\theDRcomment}}\marginpar{\color{red}\tiny\raggedright \theDRcomment: #1}\stepcounter{DRcomment}}
\newcommand\DRcomm[1]{{\color{red}#1}}

\newcounter{DVcomment}
\newcommand\DV[1]{\ensuremath{{}^{\color{green}\theDVcomment}}\marginpar{\color{green}\tiny\raggedright \theDVcomment: #1}\stepcounter{DVcomment}}

\newcounter{BMcomment}
\newcounter{JVcomment}
\newcommand\ignore[1]{}

\tikzstyle{blackdot}=[circle, draw=black, fill=black, inner sep=.5ex, line width=\thickness, node on layer=foreground]
\tikzstyle{whitedot}=[circle, draw=black, fill=white, inner sep=.5ex, line width=\thickness, node on layer=foreground]

\tikzset{proofdiagram/.style={scale=1}}
\newlength\morphismheight
\newlength\minimummorphismwidth
\newlength\stateheight
\title{\vspace{-2cm}The Morita theory of quantum graph isomorphisms\vspace{-.4cm}}
\author{\normalsize \hspace{-1.8cm}\begin{tabular}{c c c c c c}
Benjamin Musto  && David Reutter && Dominic Verdon \\ \texttt{benjamin.musto@cs.ox.ac.uk} && \texttt{david.reutter@cs.ox.ac.uk} && \texttt{dominic.verdon@cs.ox.ac.uk} \end{tabular} \\[15pt] 
\hspace{-2cm}Department of Computer Science, University of Oxford}
\date{}

\begin{document}

\zxnormal
\maketitle

\setcounter{tocdepth}{2}
\def\ang{-17}

\maketitle
\vspace{-.3cm}
\begin{abstract}
We classify instances of quantum pseudo-telepathy in the graph isomorphism game,  exploiting the recently discovered connection between quantum information and the theory of quantum automorphism groups. Specifically, we show that graphs quantum isomorphic to a given graph are in bijective correspondence with Morita equivalence classes of certain Frobenius algebras in the category of finite-dimensional representations of the quantum automorphism algebra of that graph.
We show that such a Frobenius algebra may be constructed from a central type subgroup of the classical automorphism group, whose action on the graph has coisotropic vertex stabilisers. 
In particular, if the original graph has no quantum symmetries, quantum isomorphic graphs are classified by such subgroups.
We show that all quantum isomorphic graph pairs corresponding to a well-known family of binary constraint systems arise from this group-theoretical construction.
We use our classification to show that, of the small order vertex-transitive graphs with no quantum symmetry, none is quantum isomorphic to a non-isomorphic graph. We show that this is in fact asymptotically almost surely true of all graphs.
\end{abstract}
\vspace{.1cm}
\tableofcontents

\section{Introduction}
Quantum pseudo-telepathy~\cite{Brassard2005} is a well studied phenomenon in quantum information theory, where two non-communicating parties can use pre-shared entanglement to perform a task classically impossible without communication. Such tasks are usually formulated as games, where isolated players Alice and Bob are provided with inputs, and must return outputs satisfying some winning condition. 
One such game is the graph isomorphism game~\cite{Atserias2016}, whose instances correspond to pairs of graphs $\Gamma$ and $\Gamma'$, and whose winning classical strategies are exactly graph isomorphisms $\Gamma \to \Gamma'$. Winning quantum strategies are called \emph{quantum isomorphisms}. Quantum pseudo-telepathy is exhibited by graphs that are quantum but not classically isomorphic. \ignore{\DRcomm{However, finding such graph pairs is not easy; the only known construction is from binary constraint system games exhibiting quantum advantage, such as the Mermin Peres magic square. One contribution of this work is a more general group-theoretical construction which captures all known examples and could produce new ones.}}

This work builds on two recent articles, in which Lupini, Man{\v{c}}inska and Roberson~\cite{Lupini2017} and the present authors~\cite{Musto2017a} independently discovered a connection between these quantum isomorphisms and the \emph{quantum automorphism groups} of graphs~\cite{Banica2005,Banica2009,Banica2007_2,Banica2007_3,Bichon2003} studied in the framework of compact quantum groups~\cite{Woronowicz1998}. This connection has already proven to be fruitful, introducing new quantum information-inspired techniques to the study of quantum automorphism groups~\cite{Banica2017,Lupini2017}.


Here, we use this connection in the opposite direction, showing how results from the well developed theory of quantum automorphism groups have implications for the study of pseudo-telepathy. This may seem surprising, since pseudo-telepathy requires quantum isomorphisms between non-isomorphic graphs, not quantum automorphisms.
However, we here show that\ignore{, perhaps somewhat counterintuitively\footnote{In~\cite{Musto2017a} we showed that quantum isomorphisms are in fact dualisable 1-morphisms in a 2-category.\ignore{ and should not be thought of as invertible.} From this perspective, our classification result is far less counterintuitive.},} the graphs quantum isomorphic to a given graph $\Gamma$ can in fact be classified in terms of algebraic structures in the monoidal category $\QAut(\Gamma)$ of finite-dimensional representations of Banica's quantum automorphism Hopf $C^*$-algebra $A(\Gamma)$\footnote{For a definition of this algebra, see~\cite[Definition 2.1]{Banica2007_2}. In Sections~\ref{sec:qgraphsqisos} and~\ref{sec:catqaut}, we give an explicit description of the category $\QAut(\Gamma)$ which does not require knowledge of quantum automorphism groups.}. In other words, the quantum automorphism group of a graph, together with its action on the set of vertices of the graph, fully determines all graphs quantum isomorphic to it.

We further show that much information can be obtained just from the ordinary automorphism group of a graph. For example, if a graph has \emph{no quantum symmetries} (see \cite{Banica2007_3}) it is possible to completely classify quantum isomorphic graphs in terms of certain subgroups of the ordinary automorphism group; as a consequence we show that no vertex-transitive graph of order $\leq 11$ with no quantum symmetry~\cite{Banica2007_2,Schmidt2018} is part of a pseudo-telepathic graph pair. Even if a graph does have quantum symmetries, it is still possible to construct quantum isomorphic graphs using only the structure of the ordinary automorphism group. In particular, we show that all pseudo-telepathic graph pairs arising from Lupini et al.'s version of Arkhipov's construction~\cite{Arkhipov2012,Lupini2017} --- including the graph pairs corresponding to the well-known magic square~\cite{Mermin1990} and magic pentagram constraint systems--- arise from certain $\mathbb{Z}_2^4$ or $\mathbb{Z}_2^6$ symmetries of one of the graphs.
\ignore{ \DRcomm{This shows that these well-known instances of quantum advantage in fact correspond to quite `classical' data: namely a group, a 2-cocycle, and an action of that group on a certain graph.}}

\ignore{
It's somewhat surprising that we can construct instances of pseudo-telepathy by considering only classical data such as the classical automorphism group of the graph.
In fact, we show that all pseudo-telepathic graph pairs arising from a construction of Arkhipov~\ref{} (which are all pseudo-telepathic graph pairs known to the authors) correspond to algebraic structre 
This is a somewhat surprising result 
Surprising because classical data...!!
}


Our classification results are more naturally expressed in terms of (finite) \emph{quantum graphs}, originally introduced by Weaver~\cite{Weaver2015} and generalising the noncommutative graphs of Duan, Severini and Winter~\cite{Duan2013}. These quantum graphs generalise classical graphs, with a possibly non-commutative finite-dimensional $C^*$-algebra taking the role of the set of vertices. The notions of isomorphism and quantum isomorphism can both be generalised to the setting of quantum graphs~\cite{Musto2017a}; in particular, every quantum graph has a group of automorphisms $\Aut(\Gamma)$ and a category of quantum automorphisms $\QAut(\Gamma)$, which can again be understood as the category of finite-dimensional representations of a certain Hopf $C^*$-algebra. 
We are currently not aware of a direct application of quantum isomorphic quantum graphs in quantum information theory.\footnote{Although, see~\cite{Stahlke2016} for a possible interpretation in terms of zero-error quantum communication.} Nevertheless, our classification naturally includes quantum graphs, with the classification of quantum isomorphic classical graphs arising as a special case.

All results are derived in the 2-categorical framework recently introduced by the authors~\cite{Musto2017a}.

\subsection*{The classification}
For a quantum graph $\Gamma$, we classify quantum isomorphic quantum graphs $\Gamma '$ in terms of \emph{simple\footnote{There exists a more general notion of simple Frobenius monoid in a semisimple monoidal category~\cite{Kong2008}; the simple Frobenius monoids appearing here are always simple in this broader sense.} dagger Frobenius monoids} in the representation categories $\QAut(\Gamma)$; these are dagger Frobenius monoids $X$ (see Definition~\ref{def:Frobeniusmonoid}) in $\QAut(\Gamma)$ whose underlying algebra $FX$ is simple, where $F:\QAut(\Gamma) \to \Hilb$ is the forgetful functor. In terms of the Hopf $C^*$-algebra $A(\Gamma)$ such a structure can equivalently be defined as a matrix algebra $\Mat_n(\mathbb{C})$ with normalised trace inner product $\langle A, B\rangle = \frac{1}{n} \Tr(A^\dagger B)$, equipped with a $*$-representation $\vartriangleright: A(\Gamma) \to \End(\Mat_n(\mathbb{C}))$ such that the following holds for all $x\in A(\Gamma)$ and $A,B \in \Mat_n(\mathbb{C})$:
\begin{calign}\label{eq:sweedler}\left(x_{(1)}\act A\right) \left(x_{(2)} \act B\right) = x \act \left(AB\right) 
&
x \act \mathbbm{1}_n = \epsilon(x) \mathbbm{1}_n
\end{calign}
Here, we have used Sweedler's sumless notation for the comultiplication \mbox{$\Delta(x) = x_{(1)} \otimes x_{(2)}$.} We show that two such simple dagger Frobenius monoids produce isomorphic graphs if and only if they are \emph{Morita equivalent}. Morita equivalence plays a central role in modern algebra and mathematical physics, in particular being used to classify module categories~\cite{Ostrik2003_2}, rational conformal field theories~\cite{Runkel2007} and gapped boundaries of two-dimensional gapped phases of matter~\cite{Kitaev2012}. 
\begin{res}[Classification of quantum isomorphic quantum graphs --- Corollary~\ref{cor:bigclassification}]
For a quantum graph $\Gamma$ there is a bijective correspondence between the following structures:
\begin{itemize}[itemsep=3pt]
\item Isomorphism classes of quantum graphs $\Gamma'$ quantum isomorphic to $\Gamma$.
\item Morita equivalence classes of simple dagger Frobenius monoids in $\QAut(\Gamma)$.
\end{itemize}
\end{res}

\noindent 
We remark that this classification depends only on the quantum automorphism group of $\Gamma$, and not on its action on the (quantum) set of vertices.

For applications to pseudo-telepathy, we are of course interested in a classification of quantum isomorphic \emph{classical} graphs. 

\begin{res}[Classification of quantum isomorphic classical graphs --- Corollary~\ref{cor:superclassification}]For a classical graph $\Gamma$ there is a bijective correspondence between the following structures:
\begin{itemize}[itemsep=3pt]
\item Isomorphism classes of classical graphs $\Gamma'$ quantum isomorphic to $\Gamma$.
\item Morita equivalence classes of simple dagger Frobenius monoids in $\QAut(\Gamma)$ fulfilling a certain commutativity condition.
\end{itemize}
\end{res}

\noindent
In contrast to Result 1, the classification of quantum isomorphic classical graphs depends not only on the quantum automorphism group of $\Gamma$, but also on its action on the set of vertices.

Although some of the representation categories $\QAut(\Gamma)$ have been studied before~\cite{Banica2008,Banica2009}, a general classification of Morita classes of simple dagger Frobenius monoids in all such categories seems unfeasible. We therefore focus on the \emph{classical subcategory} of $\QAut(\Gamma)$; this is the full subcategory generated by the classical automorphisms\footnote{Equivalently, the classical subcategories can be understood as the categories of finite-dimensional representations of the commutative algebra of functions on $\Aut(\Gamma)$.} of $\Gamma$, and is equivalent to the category $\Hilb_{\Aut(\Gamma)}$ of $\Aut(\Gamma)$-graded Hilbert spaces. Using the well-known classification of Morita classes of Frobenius monoids in such categories~\cite{Ostrik2003}, we can classify quantum isomorphic graphs in terms of central type subgroups of $\Aut(\Gamma)$. \ignore{In these categories the classification of Morita classes of simple dagger Frobenius monoids is well known~\cite{Ostrik2003}; they correspond to \emph{central type subgroups} of $\Aut(\Gamma)$.} A \emph{group of central type}~\cite[Definition 7.12.21]{Etingof2015}  $(L,\psi)$ is a finite group $L$ with a \emph{non-degenerate} $2$-cocycle $\psi: L \times L \to U(1)$; that is, a $2$-cocycle such that the twisted group algebra $\mathbb{C}L^\psi$ is simple. \ignore{We therefore obtain the following result.}
\begin{res}[Quantum isomorphic quantum graphs from groups --- Corollary~\ref{cor:classification}] Every central type subgroup $(L,\psi)$ of the automorphism group $\Aut(\Gamma)$ of a quantum graph $\Gamma$ gives rise to a quantum graph $\Gamma_{L, \psi}$ and a quantum isomorphism $\Gamma_{L, \psi} \to \Gamma$. Moreover, if $\Gamma$ has no quantum symmetries, this leads to a bijective correspondence between the following structures:
\begin{itemize}[itemsep=3pt]
\item Isomorphism classes of quantum graphs $\Gamma'$ quantum isomorphic to $\Gamma$.
\item Central type subgroups $(L,\psi)$ of $\Aut(\Gamma)$ up to the following equivalence relation:
\begin{equation} \label{eq:equivalenceintro} \begin{split}
(L, \psi) \sim (L', \psi') \hspace{0.1cm}\Leftrightarrow\hspace{0.1cm} &L' = gLg^{-1}\text{ and }\psi'\text{ is cohomologous to }\\
&\psi^g(x,y):=\psi(gxg^{-1}, gyg^{-1})\text{ for some }g\in \Aut(\Gamma)
\end{split}
\end{equation}
\end{itemize}
\end{res}
\noindent
Classicality of the generated graph can also be expressed in group-theoretical terms. A nondegenerate $2$-cocycle $\psi$ of a group of central type $L$ gives rise to a symplectic form\footnote{See~\cite{BenDavid:2014} for an introduction to symplectic forms on groups.} \mbox{$\rho_{\psi}: L\times L \to U(1)$}, where $\rho_{\psi}(a,b):= \psi(a,b) \overline{\psi}(aba^{-1}, a)$. In particular, a subset $S\subseteq L$ is said to be \emph{coisotropic} if it contains its orthogonal complement $S^\bot$, defined as follows, where $Z_g = \{ h \in L~|~ hg =gh\}$ denotes the centraliser of $g\in L$:
\[ S^\bot:= \left\{ g\in L ~|~ \rho_{\psi}(g,a) = 1~\forall a \in Z_g \cap S\right\} \ignore{ ~\subseteq ~S}
\]
For a subgroup $L \subseteq \Aut(\Gamma)$ and a vertex $v$ of $\Gamma$ we denote the corresponding stabiliser subgroup by $\Stab_L(v) := \{ l \in L~|~ l(v) = v\}$. We say that a central type subgroup $(L,\psi)$ of $\Aut(\Gamma)$ \emph{has coisotropic stabilisers} if the stabiliser subgroups $\Stab_L(v)$ are coisotropic for every vertex $v$ of $\Gamma$.
\begin{res}[Quantum isomorphic classical graphs from groups --- Corollary~\ref{cor:classgroup}] Every central type subgroup $(L,\psi)$ of the automorphism group $\Aut(\Gamma)$ of a classical graph $\Gamma$ with coisotropic stabilisers gives rise to a classical graph $\Gamma_{L, \psi}$ and a quantum isomorphisms $\Gamma_{L, \psi}\to \Gamma$. Moreover, if $\Gamma$ has no quantum symmetries this leads to a bijective correspondence between the following structures:
\begin{itemize}[itemsep=3pt]
\item Isomorphism classes of classical graphs $\Gamma'$ quantum isomorphic to $\Gamma$.
\item Central type subgroups $(L,\psi)$ of $\Aut(\Gamma)$ with coisotropic stabilisers up to the equivalence relation~\eqref{eq:equivalenceintro}.
\end{itemize}
\end{res}

\subsection*{Applications to pseudo-telepathy} We exhibit some first simple applications of this classification.

\begin{app}[Corollary~\ref{cor:goestozero}] The proportion of $n$-vertex graphs which admit a quantum isomorphism to a non-isomorphic graph goes to zero as $n$ goes to infinity.
\end{app}
\noindent
In~\cite{Banica2007_2,Schmidt2018} all vertex transitive graphs of order $\leq 11$ without quantum symmetries are classified. The following is then a simple application of Result 4. 
\begin{app}[Theorem~\ref{thm:vertextransitive}] None of the vertex transitive graphs of order $\leq 11$ with no quantum symmetry admits a quantum isomorphism to a non-isomorphic graph. 
\end{app}

\noindent
Conversely, we use Result 4 to construct graphs quantum isomorphic to a given graph. We will give an example of such a construction in the next paragraph. In fact, we show that all pseudo-telepathic graph pairs arising from Lupini et al.'s variant of Arkhipov's construction~\cite{Lupini2017,Arkhipov2012} are obtained by the central type subgroup construction of Result 4.

\begin{app}[Theorem~\ref{thm:arkhipov}] All pseudo-telepathic graph pairs obtained from Arkhipov's construction~\cite[Definition 4.4 and Theorem 4.5]{Lupini2017} arise from a central type subgroup of the automorphism group of one of the graphs, with coisotropic stabilisers. In particular, the central type subgroup can always taken to be isomorphic to either $\mathbb{Z}_2^4$ or $\mathbb{Z}_2^6$.
\end{app}
\ignore{This result is somewhat surprising 
In particular, surprising since completely classical data!!!\\
Next paragraph is an example.}

\subsection*{Quantum isomorphisms from groups of central type}
We now demonstrate how Result 4 --- the construction of quantum isomorphisms between classical graphs from group-theoretical data --- may be used in practise to produce pairs of graphs exhibiting pseudo-telepathy. Recall that the following input data are required:
\begin{enumerate}
\item A graph $\Gamma$;
\item A subgroup $H$ of the automorphism group of $\Gamma$;
\item A non-degenerate 2-cocycle on $H$, such that the stabiliser subgroup $\Stab_H(v) \subset H$ is coisotropic for each vertex $v$ of $\Gamma$.
\end{enumerate}
We now describe a choice of such data which produces a pseudo-telepathic graph pair.

\begin{enumerate}[wide, labelwidth=!, labelindent=0pt]
\item \emph{The graph $\Gamma$.}
The graph $\Gamma$ \ignore{the graph corresponding to the homogenisation of the magic square binary constraint system} is the \emph{homogeneous BCS graph} introduced by Atserias et al.~\cite[Figure 2]{Atserias2016} for the \textit{binary magic square} (BMS). Explicitly, this graph is defined as follows.
A binary magic square is a $3{\times}3$ matrix with entries drawn from $\{0,1\}$, such that each row and each column sum up to an even number. The following are examples:\begin{calign}\begin{pmatrix}0&0&0\\ 0&0 &0 \\0 & 0 & 0\end{pmatrix}&\begin{pmatrix}0&0&0\\ 0&1 &1 \\0 & 1 & 1\end{pmatrix}& \begin{pmatrix}1&1&0\\ 0&1 &1 \\1 & 0 & 1\end{pmatrix}\end{calign}
The definition of $\Gamma$ is as follows.
\begin{itemize}
\item Vertices of $\Gamma$ correspond to partial BMS; that is, binary magic squares in which only one row or column is filled. The following are examples:
\begin{calign}\label{eq:pBMS}\begin{pmatrix}0&0&0\\ \cdot&\cdot &\cdot \\\cdot & \cdot & \cdot\end{pmatrix} & \begin{pmatrix}\cdot&\cdot&\cdot\\ \cdot&\cdot &\cdot \\1 & 1 & 0\end{pmatrix} &\begin{pmatrix}1&\cdot&\cdot\\ 0&\cdot &\cdot \\1 & \cdot & \cdot\end{pmatrix} \end{calign}
In total there are $24$ distinct partial BMS, so the graph $\Gamma$ has 24 vertices.
\item We draw an edge between two vertices if the corresponding partial BMS are incompatible. For example, there is an edge between the vertices corresponding to the first and the last partial BMS of~\eqref{eq:pBMS}, but not between any other pair. 
\end{itemize}
\item \emph{The symmetry $\left(\mathbb{Z}_2\right)^4$.}
Given a binary magic square, we can flip bits to obtain another binary magic square, so long as we preserve the parity of each row and each column. 
We denote such symmetries as follows:
\begin{equation}\begin{pmatrix} a_{11} & a_{12} & a_{13}\\ a_{21} & a_{22} & a_{23} \\ a_{31} & a_{32} & a_{33}
\end{pmatrix}
\to[\begin{pmatrix} \times & \cdot & \times\\ \cdot & \cdot & \cdot \\ \times & \cdot & \times
\end{pmatrix}]
\begin{pmatrix} \neg a_{11} & a_{12} & \neg a_{13}\\ a_{21} & a_{22} & a_{23} \\ \neg a_{31} & a_{32} & \neg a_{33}
\end{pmatrix}
\end{equation}
These symmetries of binary magic squares induce symmetries of the graph $\Gamma$. \ignore{For example:
\begin{calign}\begin{pmatrix}0&1&1\\ \cdot&\cdot &\cdot \\\cdot & \cdot & \cdot\end{pmatrix} \to[{\begin{pmatrix}\times & \times & \cdot \\ \times & \cdot & \times\\\cdot&\times&\times\end{pmatrix}}] \begin{pmatrix}1&0&1\\ \cdot&\cdot &\cdot \\\cdot & \cdot & \cdot\end{pmatrix} &
 \begin{pmatrix}\cdot&1&\cdot\\ \cdot&0 &\cdot \\\cdot & 1 & \cdot\end{pmatrix} \to[{\begin{pmatrix}\times & \cdot & \times \\ \times & \cdot & \times\\\cdot&\cdot&\cdot\end{pmatrix}}] \begin{pmatrix}\cdot&1&\cdot\\ \cdot&0&\cdot \\\cdot & 1 & \cdot\end{pmatrix} \end{calign}}
They form a subgroup of $\Aut(\Gamma)$ isomorphic to $\left(\mathbb{Z}_2\right)^4$, and generated by the following transformations:
\begin{calign}\label{eq:symmetrytransformations}
\begin{pmatrix}
\cdot & \times & \times \\ \cdot & \times & \times \\ \cdot & \cdot & \cdot
\end{pmatrix}
&
\begin{pmatrix} \cdot & \cdot & \cdot \\\times & \times & \cdot \\ \times & \times & \cdot 
\end{pmatrix}
&
\begin{pmatrix} \times & \times & \cdot \\ \times & \times & \cdot \\ \cdot & \cdot & \cdot
\end{pmatrix}
&
\begin{pmatrix}
\cdot & \cdot &\cdot \\ 
\cdot & \times & \times \\ \cdot & \times & \times
\end{pmatrix} \\[2pt]\nonumber
(1,0,0,0) & (0,1,0,0) & (0,0,1,0) & (0,0,0,1)
\end{calign}
\item \emph{A non-degenerate 2-cocycle on $\left(\mathbb{Z}_2\right)^4$.}
It is well known that abelian groups of symmetric type --- that is, groups of the form $A{\times} A$ for some abelian group $A$ --- admit non-degenerate $2$-cocycles~\cite[Theorem 5]{Bahturin2001}.
The Pauli matrices, which form a faithful projective representation\footnote{In quantum information theory, such faithful projective representations are known as \emph{nice unitary error bases}~\cite{Klappenecker2003}. See Definition~\ref{def:niceueb}.} of $\mathbb{Z}_2^2$, give rise to such a 2-cocycle $\psi_{\mathrm{P}}$ on $\mathbb{Z}_2^2$:
\begin{calign}\nonumber \hspace{-0cm}P_{0,0}=\begin{pmatrix} 1& 0 \\ 0 & 1
\end{pmatrix}
& 
P_{1,0} = \sigma_X = \begin{pmatrix} 0& 1 \\ 1 & 0
\end{pmatrix}
&
P_{0,1} = \sigma_Z = \begin{pmatrix} 1& 0 \\ 0 & \minus 1
\end{pmatrix}
& 
P_{1,1} = -i \sigma_Y = \begin{pmatrix} 0&\minus 1 \\  1 & 0
\end{pmatrix}
\end{calign}
\begin{equation}\label{eq:paulicocycle}
P_{a_1,b_1} P_{a_2,b_2} = \psi_P ((a_1,b_1),(a_2,b_2)) P_{(a_1+a_2),(b_1+b_2)} \quad \forall \; a_1,a_2,b_1,b_2 \in \mathbb{Z}_2
\end{equation}
This induces a non-degenerate $2$-cocycle $\psi_{\mathrm{P}^2}$ on $\left(\mathbb{Z}_2\right)^4$, corresponding to the projective representation consisting of pairwise tensor products of Pauli matrices:
\begin{equation}\label{eq:Pauliproduct} U_{a,b,c,d} = P_{a,b} \otimes P_{c,d}\hspace{0.5cm} \forall~a,b,c,d \in \mathbb{Z}_2
\end{equation}
\end{enumerate}
\noindent
We now verify that the stabiliser subgroups of the action of  $\left(\mathbb{Z}_2\right)^4$ on $\Gamma$ are coisotropic for the 2-cocycle $\psi_{\mathrm{P}^2}$ and its induced symplectic form $\rho_{\mathrm{P}^2}$:
\begin{equation}\label{eq:sympformz24}
\rho_{\text{P}^2} (a,b) = \psi_{\mathrm{P}^2}(a,b) \overline{\psi_{\mathrm{P}^2}(b,a)} \hspace{0.5cm} \forall~ a,b \in \mathbb{Z}_2^4
\end{equation}
The group $\mathbb{Z}_2^4$ can be understood as a four-dimensional vector space over the finite field $\mathbb{Z}_2$. From this perspective, order $2^k$ subgroups of $\mathbb{Z}_2^4$ correspond to $k$-dimensional subspaces and  the symplectic form $\rho_{\mathrm{P}^2}$ is a symplectic form in the linear algebraic sense. Since all stabiliser subgroups are two-dimensional, they are coisotropic if and only if they are isotropic (and hence Lagrangian). A subgroup is isotropic if the restriction of the symplectic form~\eqref{eq:sympformz24} to this subgroup is trivial.  By~\eqref{eq:paulicocycle}, the form $\rho_{\mathrm{P}^2}$ is trivial on two group elements of $\mathbb{Z}_2^4$ if the corresponding tensor products of Pauli matrices~\eqref{eq:Pauliproduct} commute. For example, let $v$ be a vertex corresponding to a partial BMS in which only the first row is filled. Its stabiliser subgroup is generated by the group elements $(0,1,0,0)$ and $(0,0,0,1)$ (see~\eqref{eq:symmetrytransformations}) with corresponding Pauli matrices $\sigma_Z\otimes \mathbbm{1}_2$ and $\mathbbm{1}_2\otimes \sigma_Z$, which clearly commute. Similarly, the stabiliser subgroup of a middle column vertex is generated by the group elements $(1,0,1,0)$ and $(0,1,0,1)$ with corresponding commuting matrices $\sigma_X \otimes \sigma_X$ and $\sigma_Z \otimes \sigma_Z$.
\ignore{
All the stabiliser subgroups are two-dimensional; therefore, if they are coisotropic then they are Lagrangian and hence isotropic, and vice versa. Isotropy is triviality of the form~\eqref{eq:sympformz24} on the subgroup, and is equivalent to symmetricity of the cocycle; we need only prove this for the generators of the subgroup. Consider a vertex corresponding to a filling of the middle row of the BMS. The cocycle is indeed symmetric on the generators of its stabiliser subgroup: \begin{align*}
\psi_{\mathrm{P}^2}((1,0,0,1),(0,1,1,0)) &= \psi_P((1,0),(0,1))\psi_P((0,1),(1,0)) \\ &=\psi_P((0,1),(1,0))\psi_P((1,0),(0,1)) \\&=\psi_{\mathrm{P}^2}((0,1,1,0),(1,0,0,1))
\end{align*}
}
A similar argument holds for all rows and columns, showing that all stabiliser subgroups are coisotropic.\footnote{We note that the simultaneous assignment of $\mathbb{Z}_2^4$ group elements to symmetry transformations~\eqref{eq:symmetrytransformations} and Pauli matrices~\eqref{eq:Pauliproduct} plays an important role in this argument. Other such assignments correspond to other, possibly non-cohomologous, non-degenerate $2$-cocycles which might not have coisotropic stabilisers.}

Our construction therefore produces a graph $\Gamma_{\mathbb{Z}_2^4, \psi_{\mathrm{P}^2}}$ that is quantum isomorphic to~$\Gamma$. 
We show in Section~\ref{sec:bcsarkhipov} that this graph is isomorphic to the \emph{inhomogenous BCS graph} for the binary magic square~\cite[Figure 1]{Atserias2016}, which is known to be non-isomorphic to $\Gamma$. The two graphs therefore form a pseudo-telepathic graph pair.

\subsection*{Notations and conventions}
We assume basic familiarity with monoidal category theory~\cite{Selinger2010} and $2$-category theory~\cite[Chapter 7]{Borceux1994}. Dagger categories are defined in~\cite{Selinger2010}; a \emph{unitary} morphism in a dagger category is one whose $\dagger$-adjoint is its inverse. Strict dagger $2$-categories are defined in~\cite{Heunen2016}.\footnote{Weak dagger 2-categories are the obvious generalisation, with unitary associators and unitors.}

We use the diagrammatic calculus for monoidal categories~\cite{Selinger2010,Coecke2010} throughout; with the exception of Section~\ref{app:Frobenius}, these diagrams will always represent morphisms in $\Hilb$, the monoidal dagger category of finite-dimensional Hilbert spaces and linear maps. In Appendix~\ref{app:2categorypictures}, we additionally use the diagrammatic calculus for 2-categories~\cite{Selinger2010,Marsden2014}.

`Frobenius algebra' and `Frobenius monoid' are usually taken to be synonymous, but in this work we reserve the term `Frobenius algebra' for Frobenius monoids in Hilb and use the term `Frobenius monoid' to refer to Frobenius monoids in general monoidal categories, to aid the reader in distinguishing between the two cases.

All our definitions are adapted to the dagger (or $*$- or unitary) setting. In particular, when we say that two dagger Frobenius monoids in a dagger monoidal category are Morita equivalent we require that the corresponding invertible bimodules are compatible with the dagger structure (see Definition~\ref{def:daggerbimodule}).
\ignore{A projector on a Hilbert space $H$ is an endomorphism $P: H \to H$ which is idempotent and self-adjoint $P = P^{\dagger} = P^2$. All sets appearing in this work are finite, and all vector spaces and all algebras are finite-dimensional.}

\ignore{We will need to distinguish between quantum isomorphisms and ordinary isomorphisms, and quantum graphs and ordinary graphs. }Whenever we say `graph' or `isomorphism' without the modifier `quantum' we always refer to the ordinary, or classical notion (isomorphisms between quantum graphs are defined in Definition~\ref{def:quantumgraphiso}). Occasionally, to clearly distinguish between the two cases, we explicitly use the modifier 	`classical' or `ordinary'.

\subsection*{Acknowledgments}
We are grateful to Jamie Vicary for many useful discussions and to David Roberson for sending us an early draft of~\cite{Lupini2017}. \ignore{We also thank Samson Abramsky and Rui Soares Barbosa for many useful comments on an early version of this work.\DV{This sounds like Samson and Rui read an early draft of the paper. Maybe `during the early stages of this work'?}}

\section{Background}

In this section, we recall various definitions and results; most of these are treated in greater detail in~\cite{Musto2017a}.

\subsection{String diagrams, Frobenius monoids and Gelfand duality}

\ignore{
In this work we consider two monoidal categories. The first is the compact closed~\cite{Kelly1972,Kelly1980,Abramsky2004} category $\Hilb$ of finite dimensional Hilbert spaces and linear maps, and the second is the category $\QAut(\Gamma)$ of automorphisms of a classical graph $\Gamma$. We treat both  these categories using the graphical calculus for monoidal categories~\cite{Joyal1991a,Joyal1991b,Selinger2010,Coecke2010}, which we briefly summarise now.
}
Most results in this work are derived using the graphical calculus of monoidal dagger categories~\cite{Selinger2010,Coecke2010}. Except for Section~\ref{app:Frobenius} and Appendix~\ref{app:2categorypictures}, we only use the graphical calculus of the compact closed~\cite{Kelly1980,Abramsky2004} dagger category $\Hilb$ of finite-dimensional Hilbert spaces and linear maps.

In the graphical calculus, morphisms are displayed as \emph{string diagrams}, which we read from bottom to top. In these diagrams of strings and nodes, strings are labelled with objects, and nodes are labelled with morphisms. The string for the monoidal unit $I$ is not drawn.
Composition and tensor product are depicted as follows:
\begin{calign}
\begin{tz}[zx]
\draw (0,0) to (0,3.5);
\node[zxnode=\zxwhite] at (0,1.05) {$f$};
\node[zxnode=\zxwhite] at (0,2.45) {$g$};
\node[dimension, right] at (0,0){$A$};
\node[dimension, right] at (0, 1.75) {$B$};
\node[dimension, right] at (0, 3.5) {$C$};
\end{tz}
&
\begin{tz}[zx]
\draw (0,0) to (0,3.5);
\draw (1.5,0) to (1.5,3.5);
\node[zxnode=\zxwhite] at (0,1.75) {$f$};
\node[zxnode=\zxwhite] at (1.5,1.75) {$g$};
\node[dimension, right] at (0,0){$A$};
\node[dimension, right] at (1.5, 0) {$B$};
\node[dimension, right] at (0, 3.5) {$C$};
\node[dimension, right] at (1.5, 3.5) {$D$};
\end{tz}\\\nonumber
gf:A\to C
& 
f\otimes g: A\otimes B\to C\otimes D
\end{calign}
In a monoidal dagger category, given a morphism $f:A\to B$, we express its $\dagger$-adjoint $f^\dagger:B\to A$ as a reflection of the corresponding diagram across a horizontal axis. 

Restricting attention to the category $\Hilb$, we note that all finite-dimensional Hilbert spaces $V$ have dual spaces $V^*=\Hom(V,\mathbb{C})$, represented in the graphical calculus as an oriented wire with the opposite orientation as $V$. Duality is characterized by the following linear maps, here called \textit{cups and caps}:
\def\pv{\vphantom{V^*}}
\begin{calign}\label{eq:cupscapsHilb}
\begin{tz}[zx]
\draw[arrow data ={0.15}{<}, arrow data={0.89}{<}] (0,0) to [out=up, in=up, looseness=2.5] (2,0) ;
\node[dimension, right] at (2.05,0) {$\pv V$};
\node[dimension, left] at (0,0) {$V^*$};
\end{tz}
&
\begin{tz}[zx,scale=-1]
\draw[arrow data ={0.15}{>}, arrow data={0.89}{>}] (0,0) to [out=up, in=up, looseness=2.5] (2,0) ;
\node[dimension, left] at (2.0,0) {$\pv V$};
\node[dimension, right] at (0,0) {$V^*$};
\end{tz}
&
\begin{tz}[zx,xscale=-1]
\draw[arrow data ={0.15}{<}, arrow data={0.89}{<}] (0,0) to [out=up, in=up, looseness=2.5] (2,0) ;
\node[dimension, left] at (2.0,0) {$\pv V$};
\node[dimension, right] at (0,0) {$V^*$};
\end{tz}
&
\begin{tz}[zx,yscale=-1]
\draw[arrow data ={0.15}{>}, arrow data={0.89}{>}] (0,0) to [out=up, in=up, looseness=2.5] (2,0) ;
\node[dimension, right] at (2.05,0) {$\pv V$};
\node[dimension, left] at (0,0) {$V^*$};
\end{tz}\\\nonumber
f\otimes v \mapsto f(v) 
& 
~1\mapsto \mathbbm{1}_V
&
v\otimes f\mapsto f(v)
&
~~1\mapsto \mathbbm{1}_V
\end{calign}
To define the second and fourth map, we have identified $V\otimes V^* \cong V^*\otimes V \cong\End(V)$. It may be verified that these maps fulfill the following \textit{snake equations}:
\begin{calign}\label{eq:snake}
\begin{tz}[zx]
\draw[arrow data={0.15}{>}, arrow data={0.5}{>}, arrow data={0.9}{>}] (0,0) to (0,1) to [out=up, in=up, looseness=2] (1,1) to [out=down, in=down, looseness=2] (2,1) to (2,2);
\end{tz}
~~~=~~~
\begin{tz}[zx]
\draw[arrow data={0.5}{>}] (0,0) to (0,2);
\end{tz}
~~~= ~~~
\begin{tz}[zx,xscale=-1]
\draw[arrow data={0.15}{>}, arrow data={0.5}{>}, arrow data={0.9}{>}] (0,0) to (0,1) to [out=up, in=up, looseness=2] (1,1) to [out=down, in=down, looseness=2] (2,1) to (2,2);
\end{tz}
&
\begin{tz}[zx]
\draw[arrow data={0.15}{<}, arrow data={0.5}{<}, arrow data={0.9}{<}] (0,0) to (0,1) to [out=up, in=up, looseness=2] (1,1) to [out=down, in=down, looseness=2] (2,1) to (2,2);
\end{tz}
~~~ =~~~
\begin{tz}[zx]
\draw[arrow data={0.5}{<}] (0,0) to (0,2);
\end{tz}
~~~= ~~~
\begin{tz}[zx,xscale=-1]
\draw[arrow data={0.15}{<}, arrow data={0.5}{<}, arrow data={0.9}{<}] (0,0) to (0,1) to [out=up, in=up, looseness=2] (1,1) to [out=down, in=down, looseness=2] (2,1) to (2,2);
\end{tz}
\end{calign}
Together with the swap map $\sigma_{V,W}:\ignore{V\otimes W\to W\otimes V$,~$}v\otimes w\mapsto w\otimes v$, depicted as a crossing of wires, this leads to a very flexible topological calculus, allowing us to untangle arbitrary diagrams and straighten out any twists:
\begin{calign}\label{eq:untangleunwindtwists}
\begin{tz}[zx,scale=1]
	\begin{pgfonlayer}{nodelayer}
		\node [style=none] (0) at (0, -0) {};
		\node [style=none] (1) at (0, 3) {};
		\node [style=none] (2) at (1, -0) {};
		\node [style=none] (3) at (1, 3) {};
		\node [style=none] (4) at (2, -0) {};
		\node [style=none] (5) at (2, 3) {};
		\node [style=none] (6) at (3, -0) {};
		\node [style=none] (7) at (3, 3) {};
		\node [style=none] (8) at (2.25, 0.75) {};
		\node [style=none] (9) at (1.25, 1.25) {};
		\node [style=none] (10) at (3, 2) {};
		\node [style=none] (11) at (1, 2) {};
		\node [style=none] (12) at (1.75, 1.25) {};
		\node [style=none] (13) at (0.75, 0.5) {};
		\node [style=none] (14) at (0, 2) {};
		\node [style=none] (15) at (2, -0) {};
		\node [style=none] (16) at (2.25, 2) {};
		\node [style=none] (17) at (4, 1.5) {$=$};
		\node [style=none] (18) at (6, -0) {};
		\node [style=none] (19) at (7, -0) {};
		\node [style=none] (20) at (8, -0) {};
		\node [style=none] (21) at (5, -0) {};
		\node [style=none] (22) at (7, -0) {};
		\node [style=none] (23) at (7, 3) {};
		\node [style=none] (24) at (8, 3) {};
		\node [style=none] (25) at (6, 3) {};
		\node [style=none] (26) at (5, 3) {};
	\end{pgfonlayer}
	\begin{pgfonlayer}{edgelayer}
		\draw [thick, in=-90, out=27, looseness=0.75] (0.center) to (8.center);
		\draw [thick, in=-45, out=105, looseness=0.75] (8.center) to (9.center);
		\draw [thick, in=-90, out=105, looseness=0.75] (9.center) to (10.center);
		\draw [thick, in=-90, out=105, looseness=0.50] (10.center) to (1.center);
		\draw [thick, bend right=15, looseness=1.00] (3.center) to (11.center);
		\draw [thick, in=90, out=-90, looseness=0.75] (11.center) to (12.center);
		\draw [thick, in=-90, out=90, looseness=0.50] (13.center) to (12.center);
		\draw [thick, bend right, looseness=1.00] (13.center) to (2.center);
		\draw [thick, in=90, out=-153, looseness=0.75] (5.center) to (14.center);
		\draw [thick, in=135, out=-90, looseness=1.50] (14.center) to (4.center);
		\draw [thick, in=-127, out=90, looseness=1.50] (16.center) to (7.center);
		\draw [thick, in=111, out=-90, looseness=1.25] (16.center) to (6.center);
		\draw [style=simple] (21.center) to (26.center);
		\draw [style=simple] (18.center) to (25.center);
		\draw [style=simple] (19.center) to (23.center);
		\draw [style=simple] (20.center) to (24.center);
	\end{pgfonlayer}
\end{tz}
&
\begin{tz}[zx]
\draw[arrow data ={0.075}{>}, arrow data ={0.5}{>}, arrow data={0.925}{>}] (0,0) to [out=up, in=right] (-1.25, 2.25) to [out=left, in=left, looseness=1] (-1.25,0.75) to [out=right, in=down] (0,3);
\end{tz}
~~~=~~~\begin{tz}[zx]
\draw[arrow data={0.5}{>}] (0,0) to (0,3);
\end{tz}
~~~=~~~
\begin{tz}[zx,xscale=-1]
\draw[arrow data ={0.075}{>}, arrow data ={0.5}{>}, arrow data={0.925}{>}] (0,0) to [out=up, in=right] (-1.25, 2.25) to [out=left, in=left, looseness=1] (-1.25,0.75) to [out=right, in=down] (0,3);
\end{tz}%
\ignore{
\input{unwindtwists.tikz}}
\end{calign}
\noindent
A closed circle evaluates to the dimension of the corresponding Hilbert space:
\begin{calign}\label{eq:closedcircle}
\begin{tz}[zx]
\draw[arrow data = {0}{>}] (0,0) circle (0.75) ; 
\end{tz}
~~=~~
\begin{tz}[zx]
\draw[arrow data = {0}{<}] (0,0) circle (0.75) ; 
\end{tz}
~~=~\dim(H)
\end{calign}

\subsubsection{Frobenius monoids}

We now recall the notion of a dagger Frobenius monoid in a monoidal dagger category. \ignore{Although our monoidal categories have duals, we refrain from drawing an orientation on the wire corresponding to the object carrying the algebra structure, for reasons which will soon become apparent.}

\begin{definition}\label{def:algebra}A \textit{monoid} in a monoidal category is an object $M$ with multiplication and unit morphisms, depicted as follows:
\begin{calign}
\begin{tz}[zx,master]
\coordinate (A) at (0,0);
\draw (0.75,1) to (0.75,2);
\mult{A}{1.5}{1}
\end{tz}
&
\begin{tz}[zx,slave]
\coordinate (A) at (0.75,2);
\unit{A}{1}
\end{tz}
\\[0pt]\nonumber
m:M\otimes M \to M& u: I \to M 
\end{calign}\hspace{-0.2cm}
These morphisms satisfy the following associativity and unitality equations:
\begin{calign}\label{eq:assocandunitality}
\begin{tz}[zx]
\coordinate(A) at (0.25,0);
\draw (1,1) to [out=up, in=-135] (1.75,2);
\draw (1.75,2) to [out=-45, in=up] (3.25,0);
\draw (1.75,2) to (1.75,3);
\mult{A}{1.5}{1}
\node[zxvertex=\zxwhite,zxdown] at (1.75,2){};
\end{tz}
\quad = \quad
\begin{tz}[zx,xscale=-1]
\coordinate(A) at (0.25,0);
\draw (1,1) to [out=up, in=-135] (1.75,2);
\draw (1.75,2) to [out=-45, in=up] (3.25,0);
\draw (1.75,2) to (1.75,3);
\mult{A}{1.5}{1}
\node[zxvertex=\zxwhite,zxdown] at (1.75,2){};
\end{tz}
&
\begin{tz}[zx]
\coordinate (A) at (0,0);
\draw (0,-0.25) to (0,0);
\draw (0.75,1) to (0.75,2);
\mult{A}{1.5}{1}
\node[zxvertex=\zxwhite,zxdown] at (1.5,0){};
\end{tz}
\quad =\quad
\begin{tz}[zx]
\draw (0,0) to (0,2);
\end{tz}
\quad= \quad
\begin{tz}[zx,xscale=-1]
\coordinate (A) at (0,0);
\draw (0,-0.25) to (0,0);
\draw (0.75,1) to (0.75,2);
\mult{A}{1.5}{1}
\node[zxvertex=\zxwhite,zxdown] at (1.5,0){};
\end{tz}
\end{calign}
Analogously, a \textit{comonoid} is an object $C$ with a coassociative comultiplication $\delta: C \to C\otimes C$ and counit $\epsilon:C\to I$. The $\dagger$-adjoint of a monoid in a monoidal dagger category is a comonoid.
\end{definition}
\noindent
Note that for the multiplication and unit morphisms of an monoid we simply draw white nodes rather than labelled boxes, for concision. Likewise, we draw the comultiplication and counit morphisms of the $\dagger$-adjoint comonoid as white nodes. Despite having the same label in the diagram, they can be easily distinguished by their type. 

\begin{definition} \label{def:Frobeniusmonoid}A \textit{dagger Frobenius monoid} in a monoidal dagger category is a monoid where the monoid and $\dagger$-adjoint comonoid structures are related by the following Frobenius equations:
\begin{equation}\label{eq:Frobenius}
\begin{tz}[zx]
\draw (0,0) to [out=up, in=-135] (0.75,2) to (0.75,3);
\draw (0.75,2) to [out=-45, in=135] (2.25,1);
\draw (2.25,0) to (2.25,1) to [out=45, in=down] (3,3);
\node[zxvertex=\zxwhite,zxup] at (2.25,1){};
\node[zxvertex=\zxwhite,zxdown] at (0.75,2){};
\end{tz}
\quad = \quad
\begin{tz}[zx]
\coordinate (A) at (0,0);
\coordinate (B) at (0,3);
\draw (0.75,1) to (0.75,2);
\mult{A}{1.5}{1}
\comult{B}{1.5}{1}
\end{tz}
\quad = \quad 
\begin{tz}[zx]
\draw (0,0) to [out=up, in=-135] (0.75,2) to (0.75,3);
\draw (0.75,2) to [out=-45, in=135] (2.25,1);
\draw (2.25,0) to (2.25,1) to [out=45, in=down] (3,3);
\node[zxvertex=\zxwhite,zxup] at (2.25,1){};
\node[zxvertex=\zxwhite,zxdown] at (0.75,2){};
\end{tz}
\end{equation}
\ignore{
\begin{terminology}
In this work we only consider dagger Frobenius monoids, and therefore  refer to these simply as \emph{Frobenius monoids}, omitting the extra adjective.
\end{terminology}
}%
A dagger Frobenius monoid is \textit{special} if equation (\ref{eq:special}a) holds. A dagger Frobenius algebra in $\Hilb$ is moreover \textit{symmetric} or \textit{commutative} if one of (\ref{eq:special}b) or (\ref{eq:special}c) holds.
\begin{calign}\label{eq:special}
\begin{tz}[zx,every to/.style={out=up, in=down}]\draw (0,0) to (0,1) to [out=135] (-0.75,2) to [in=-135] (0,3) to (0,4);
\draw (0,1) to [out=45] (0.75,2) to [in=-45] (0,3);
\node[zxvertex=\zxwhite, zxup] at (0,1){};
\node[zxvertex=\zxwhite,zxdown] at (0,3){};\end{tz}
\quad = \quad
\begin{tz}[zx]
\draw (0,0) to +(0,4);
\end{tz}
&
\begin{tz}[zx,every to/.style={out=up, in=down}]
\draw (0.25,-1.5) to (1.75,0) to [in=-45] (1,1) to (1,1.5);
\draw (1.75,-1.5) to (0.25,0) to [in=-135] (1,1);
\node[zxvertex=\zxwhite, zxdown] at (1,1){};
\node[zxvertex=\zxwhite] at (1,1.5){};
\end{tz}
\quad = \quad
\begin{tz}[zx,every to/.style={out=up, in=down}]
\draw (1.75,-1.5) to (1.75,0) to [in=-45] (1,1) to (1,1.5);
\draw (0.25,-1.5) to (0.25,0) to [in=-135] (1,1);
\node[zxvertex=\zxwhite, zxdown] at (1,1){};
\node[zxvertex=\zxwhite] at (1,1.5){};
\end{tz}
&
\begin{tz}[zx,every to/.style={out=up, in=down}]
\draw (0.25,-1.5) to (1.75,0) to [in=-45] (1,1) to (1,2);
\draw (1.75,-1.5) to (0.25,0) to [in=-135] (1,1);
\node[zxvertex=\zxwhite, zxdown] at (1,1){};
\end{tz}
\quad = \quad
\begin{tz}[zx,every to/.style={out=up, in=down}]
\draw (1.75,-1.5) to (1.75,0) to [in=-45] (1,1) to (1,2);
\draw (0.25,-1.5) to (0.25,0) to [in=-135] (1,1);
\node[zxvertex=\zxwhite, zxdown] at (1,1){};
\end{tz}\\\nonumber
\text{a) special}&\text{b) symmetric}&\text{c) commutative}
\end{calign} 
\end{definition}
\noindent
Dagger Frobenius monoids are closely related to dualities. In particular, it is a direct consequence of~\eqref{eq:assocandunitality} and~\eqref{eq:Frobenius} that the following cups and caps fulfil the snake equations~\eqref{eq:snake}:
\begin{calign}\label{eq:cupcapfrob}
\begin{tz}[zx]
\clip (-0.1,0) rectangle (2.1,2.);
\draw (0,0) to [out=up, in=up, looseness=2] node[zxvertex=\zxwhite, pos=0.5]{} (2,0);
\end{tz}
~:=~
\begin{tz}[zx]
\clip (-0.1,0) rectangle (2.1,2.2);
\draw (0,0) to [out=up, in=up, looseness=2] node[front,zxvertex=\zxwhite, pos=0.5](A){} (2,0);
\draw[string] (A.center) to (1,1.8);
\node[zxvertex=\zxwhite] at (1,1.8) {};
\end{tz}
&
\begin{tz}[zx,yscale=-1]
\clip (-0.1,0) rectangle (2.1,2.);
\draw (0,0) to [out=up, in=up, looseness=2] node[zxvertex=\zxwhite, pos=0.5]{} (2,0);
\end{tz}
~:=~
\begin{tz}[zx,yscale=-1]
\clip (-0.1,0) rectangle (2.1,2.2);
\draw (0,0) to [out=up, in=up, looseness=2] node[front,zxvertex=\zxwhite, pos=0.5](A){} (2,0);
\draw[string] (A.center) to (1,1.8);
\node[zxvertex=\zxwhite] at (1,1.8) {};
\end{tz}
\end{calign}
\ignore{It follows that every Frobenius monoid is canonically self-dual, $A^*\cong A$; this is why we do not need to draw an orientation on the corresponding wire.
}
\noindent
Finally, we define a notion of homomorphism between dagger Frobenius monoids.
\begin{definition}A \textit{$*$-homomorphism} $f:A\to B$ between dagger Frobenius monoids $A$ and $B$ is a morphism $f:A\to B$ satisfying the following equations:
\begin{calign}\label{eq:homo}
\begin{tz}[zx, master, every to/.style={out=up, in=down},yscale=-1]
\draw (0,0) to (0,2) to [out=135] (-0.75,3);
\draw (0,2) to [out=45] (0.75, 3);
\node[zxnode=\zxwhite] at (0,1) {$f$};
\node[zxvertex=\zxwhite, zxdown] at (0,2) {};
\end{tz}
=
\begin{tz}[zx, every to/.style={out=up, in=down},yscale=-1]
\draw (0,0) to (0,0.75) to [out=135] (-0.75,1.75) to (-0.75,3);
\draw (0,0.75) to [out=45] (0.75, 1.75) to +(0,1.25);
\node[zxnode=\zxwhite] at (-0.75,2) {$f$};
\node[zxnode=\zxwhite] at (0.75,2) {$f$};
\node[zxvertex=\zxwhite, zxdown] at (0,0.75) {};
\end{tz}
&
\begin{tz}[zx,slave, every to/.style={out=up, in=down},yscale=-1]
\draw (0,0) to (0,2) ;
\node[zxnode=\zxwhite] at (0,1) {$f$};
\node[zxvertex=\zxwhite, zxup] at (0,2) {};
\end{tz}
=
\begin{tz}[zx,slave, every to/.style={out=up, in=down},yscale=-1]
\draw (0,0) to (0,0.75) ;
\node[zxvertex=\zxwhite, zxup] at (0,0.75) {};
\end{tz}
&
\begin{tz}[zx,slave, every to/.style={out=up, in=down},scale=-1]
\draw (0,0) to (0,3);
\node[zxnode=\zxwhite] at (0,1.5) {$f^\dagger$};
\end{tz}
=~~
\begin{tz}[zx,slave,every to/.style={out=up, in=down},scale=-1]
\draw (0,1.5) to (0,2) to [in=left] node[pos=1] (r){} (0.5,2.5) to [out=right, in=up] (1,2)  to [out=down, in=up] (1,0);
\draw (-1,3) to [out=down,in=up] (-1,1) to [out=down, in=left] node[pos=1] (l){} (-0.5,0.5) to [out=right, in=down] (0,1) to (0,1.5);
\node[zxnode=\zxwhite] at (0,1.5) {$f$};
\node[zxvertex=\zxwhite] at (l.center){};
\node[zxvertex=\zxwhite] at (r.center){};
\end{tz}
\end{calign}
A \textit{$*$-cohomomorphism} $f:A\to B$ is a morphism $f:A\to B$ satisfying the following equations:
\begin{calign}\label{eq:cohomo}
\begin{tz}[zx, master, every to/.style={out=up, in=down}]
\draw (0,0) to (0,2) to [out=135] (-0.75,3);
\draw (0,2) to [out=45] (0.75, 3);
\node[zxnode=\zxwhite] at (0,1) {$f$};
\node[zxvertex=\zxwhite, zxup] at (0,2) {};
\end{tz}
=
\begin{tz}[zx, every to/.style={out=up, in=down}]
\draw (0,0) to (0,0.75) to [out=135] (-0.75,1.75) to (-0.75,3);
\draw (0,0.75) to [out=45] (0.75, 1.75) to +(0,1.25);
\node[zxnode=\zxwhite] at (-0.75,2) {$f$};
\node[zxnode=\zxwhite] at (0.75,2) {$f$};
\node[zxvertex=\zxwhite, zxup] at (0,0.75) {};
\end{tz}
&
\begin{tz}[zx,slave, every to/.style={out=up, in=down}]
\draw (0,0) to (0,2) ;
\node[zxnode=\zxwhite] at (0,1) {$f$};
\node[zxvertex=\zxwhite, zxup] at (0,2) {};
\end{tz}
=
\begin{tz}[zx,slave, every to/.style={out=up, in=down}]
\draw (0,0) to (0,0.75) ;
\node[zxvertex=\zxwhite, zxup] at (0,0.75) {};
\end{tz}
&
\begin{tz}[zx,slave, every to/.style={out=up, in=down}]
\draw (0,0) to (0,3);
\node[zxnode=\zxwhite] at (0,1.5) {$f^\dagger$};
\end{tz}
=~~
\begin{tz}[zx,slave,every to/.style={out=up, in=down}]
\draw (0,1.5) to (0,2) to [in=left] node[pos=1] (r){} (0.5,2.5) to [out=right, in=up] (1,2)  to [out=down, in=up] (1,0);
\draw (-1,3) to [out=down,in=up] (-1,1) to [out=down, in=left] node[pos=1] (l){} (-0.5,0.5) to [out=right, in=down] (0,1) to (0,1.5);
\node[zxnode=\zxwhite] at (0,1.5) {$f$};
\node[zxvertex=\zxwhite] at (l.center){};
\node[zxvertex=\zxwhite] at (r.center){};
\end{tz}
\end{calign}
A \emph{$*$-isomorphism} is a morphism which is both a $*$-homomorphism and a $*$-cohomomorphism. 
\end{definition}
\noindent
We observe that the dagger of a $*$-homomorphism is a $*$-cohomomorphism, that every {$*$-isomorphism} is unitary, and that every unitary $*$-homomorphism between dagger Frobenius monoids is a \mbox{$*$-isomorphism.}

\ignore{
Define a Frobenius monoid in a general monoidal cat., define 
special. \\}
 
 Recall that we refer to Frobenius monoids in $\Hilb$ as Frobenius algebras. A major reason for defining these structures is the fact that special symmetric dagger Frobenius algebras coincide with finite-dimensional $C^*$-algebras.
\begin{theorem}[{\cite[Theorem 4.6 and 4.7]{Vicary2010}}]\label{thm:jamiec*frob} Every finite-dimensional $C^*$-algebra admits a unique inner product making it into a special symmetric dagger Frobenius algebra. Conversely, every special symmetric dagger Frobenius algebra $A$ admits a unique norm such that the canonical involution, defined by its action on vectors $\ket{a}\in A$ as the following antihomomorphism, endows it with the structure of a $C^*$-algebra:%
\begin{calign}
\begin{tz}[zx,master]
\draw (0.25,3) to (0.25,1);
\node[zxnode=\zxwhite] at (0.25,1) {$a$};
\end{tz}
\quad \mapsto \quad
\begin{tz}[zx,slave]
\draw (0.25,3) to  (0.25,2) to [out=down, in=135] (1,1) to [out=45, in=down] (1.75,2);
\draw (1,1) to (1,0.5);
\node[zxnode=\zxwhite] at (1.75,2) {$a^\dagger$};
\node[zxvertex=\zxwhite,zxup] at (1,1){};
\node[zxvertex=\zxwhite] at (1,0.5){};
\end{tz}\\[-5pt]\nonumber\end{calign} 
Moreover, the notions of $*$-homomorphism and $*$-isomorphism between special symmetric dagger Frobenius algebras coincide with the corresponding notions for finite-dimensional \mbox{$C^*$-algebras.}
\end{theorem}
\noindent
One advantage of explicitly using special symmetric dagger Frobenius algebras instead of $C^*$-algebras is that Frobenius algebras already contain `up-front' all emergent structures of finite-dimensional $C^*$-algebras, such as the comultiplication $\Delta = m^\dagger : H\to H \otimes H$; they are therefore more amenable to the purely compositional reasoning of the graphical calculus.

One important example of a special symmetric dagger Frobenius algebra is the endomorphism algebra of a Hilbert space.
\begin{definition} The \emph{endomorphism algebra} of a Hilbert space $H$ is defined to  be the following special symmetric dagger Frobenius algebra on $H\otimes H^*$ (where $n=\dim(H)$):
\begin{calign}\label{eq:endomorphismalgebra}
\frac{1}{\sqrt{n}}~
\begin{tz}[zx,master,every to/.style={out=up, in=down}]
\draw[arrow data={0.5}{>}] (0,0) to [looseness=1.6] (1,3.5);
\draw[arrow data={0.5}{<}] (3,0) to [looseness=1.6] (2, 3.5);
\draw[arrow data={0.5}{<}] (1,0) to [out=up, in=up, looseness=3] (2,0);
\end{tz}
&
\sqrt{n}\!\!\!\!\!
\begin{tz}[zx,slave,every to/.style={out=up, in=down}]
\draw[arrow data={0.52}{<}] (1,3.5) to [out=down, in=down, looseness=3] (2,3.5);
\end{tz}
&
\frac{1}{\sqrt{n}}~
\begin{tz}[zx,master,every to/.style={out=up, in=down},yscale=-1]
\draw[arrow data={0.5}{<}] (0,0) to [looseness=1.6] (1,3.5);
\draw[arrow data={0.5}{>}] (3,0) to [looseness=1.6] (2, 3.5);
\draw[arrow data={0.5}{>}] (1,0) to [out=up, in=up, looseness=3] (2,0);
\end{tz}
&
\sqrt{n}\!\!\!\!\!
\begin{tz}[zx,slave,every to/.style={out=up, in=down},yscale=-1]
\draw[arrow data={0.52}{>}] (1,3.5) to [out=down, in=down, looseness=3] (2,3.5);
\end{tz}
\end{calign}
\end{definition}
\begin{remark}\label{rem:normalisation}The normalisation factors were chosen to make the endomorphism algebra special. This is not essential but simplifies some of our arguments.
The algebra~\eqref{eq:endomorphismalgebra} is $*$-isomorphic to the unique special symmetric dagger Frobenius algebra corresponding to the usual $C^*$-algebra structure on $\End(H)$ which is usually given with unnormalised multiplication and unit but normalised inner product $\langle A,B\rangle := \frac{1}{n}\Tr(A^\dagger B)$ to retain specialness.  We prefer the normalisation~\eqref{eq:endomorphismalgebra}, since the normalised inner product does not arise as the canonical induced inner product on the tensor product Hilbert space $H\otimes H^*$.  %
\ignore{\ignore{ In particular, that algebra has a different normalisation; the multiplication and unit are not normalised, but the comultiplication and counit are rescaled by a factor of $\frac{1}{n}$ and $n$ instead of $\frac{1}{\sqrt{n}}$ and $\sqrt{n}$, respectively.}

Choosing a normalised inner product $\langle A , B \rangle := \frac{1}{n} \Tr(A^\dagger B)$ on $\End(H)$ would have let to another $*$-isomorphic special dagger Frobenius algebra which would have avoided the normalisation factors for multiplication and unit. We prefer the normalisation~\eqref{eq:quantumbijectionFrobenius}, since the normalised inner product does not arise as the canonical induced inner product on the tensor product Hilbert space $H\otimes H^*$.

\DRcomm{
We note that this normalised special dagger Frobenius algebra is $*$-isomorphic to the one with more conventional normalisation where multiplication and unit are not normalised but comultiplication and counit are rescaled by a factor of $\frac{1}{n}$ and $n$ instead of $\frac{1}{\sqrt{n}}$ and $\sqrt{n}$, respectively. In particular, this Frobenius algebra is the unique special symmetric dagger Frobenius algebra corresponding to the usual $C^*$-algebra structure on $\End(H)$. We prefer the normalisation~\eqref{eq:quantumbijectionFrobenius}, since the more conventional normalisation just discussed only becomes a dagger Frobenius algebra with the normalised inner product $\langle A, B \rangle :=\frac{1}{n} \Tr(A^\dagger B)$ which does not arise as the canonical induced inner product on the tensor product Hilbert space $H\otimes H^*$.}
}%
\end{remark}

\subsubsection{Gelfand duality and Frobenius algebras}

We now recall the graphical version of finite- dimensional Gelfand duality in the framework established by Coecke, Pavlovi{\'c} and Vicary~\cite{Coecke2009}. We first observe that every orthonormal basis of a Hilbert space $H$ defines a special commutative dagger Frobenius algebra on $H$.

\begin{example} \label{exm:Frob}Let $\left\{\ket{i}\right\}_{1\leq i\leq n}$ be an orthonormal basis of a Hilbert space $H$. Then the following multiplication and unit maps, together with their adjoints, form a special commutative dagger Frobenius algebra on $H$:
\begin{calign}\label{eq:classicalcopy}\begin{tz}[zx,master, scale=1.3]
\coordinate (A) at (0,0);
\draw (0.75,1) to (0.75,2);
\mult{A}{1.5}{1}
\end{tz} := ~\sum_{i=1}^{n} ~~~
\begin{tz}[zx,slave,scale=1.3]
\draw (0,0) to (0,0.5);
\draw (1.5,0) to (1.5,0.5);
\draw (0.75,1.5) to (0.75,2);
\node[zxnode=\zxwhite] at (0,0.5) {$i^\dagger$};
\node[zxnode=\zxwhite] at (1.5,0.5) {$i^\dagger$};
\node[zxnode=\zxwhite] at (0.75,1.5) {$i$};
\end{tz}
&
\begin{tz}[zx,slave,scale=1.3]
\draw (0.75,1) to (0.75,2);
\node[zxvertex=\zxwhite] at (0.75,1){};
\end{tz} := ~\sum_{i=1}^{n} 
\begin{tz}[zx,slave,scale=1.3]
\draw (0.75,1.25) to (0.75,2);
\node[zxnode=\zxwhite] at (0.75,1.25) {$i$};
\end{tz}\\[5pt]\nonumber
m: \ket{i} \otimes \ket{j} \mapsto \delta_{i,j} \ket{i} 
&
u: 1 \mapsto \sum_{i=1}^{n} \ket{i}
\end{calign}
\end{example}
\noindent
Conversely, every special commutative dagger Frobenius algebra $A$ gives rise to an orthonormal basis of $A$; the basis vectors are given by the copyable elements of $A$, defined as follows.

\begin{definition}\label{def:copyablestates} A \textit{copyable element} of a special commutative dagger Frobenius algebra $A$ is a $*$-cohomomorphism $\psi: \mathbb{C} \to A$; that is, a vector $\ket{\psi} \in A$, such that the following hold:
\begin{calign}\label{eq:ordinaryelement}\begin{tz}[zx, master, every to/.style={out=up, in=down}]
\draw (0,1) to (0,2) to [out=135] (-0.75,3);
\draw (0,2) to [out=45] (0.75, 3);
\node[zxnode=\zxwhite] at (0,1) {$\psi$};
\node[zxvertex=\zxwhite, zxup] at (0,2) {};
\end{tz}
=~~~
\begin{tz}[zx,slave, every to/.style={out=up, in=down}]
\draw (-0.75,2) to (-0.75,3);
\draw (0.75, 2) to +(0,1.);
\node[zxnode=\zxwhite] at (-0.75,2) {$\psi$};
\node[zxnode=\zxwhite] at (0.75,2) {$\psi$};
\end{tz}
&
\begin{tz}[zx, every to/.style={out=up, in=down}]
\draw (0,1) to (0,2) ;
\node[zxnode=\zxwhite] at (0,1) {$\psi$};
\node[zxvertex=\zxwhite, zxup] at (0,2) {};
\end{tz}
~~=~~
\emptydiagram
&
\begin{tz}[zx,slave, every to/.style={out=up, in=down}]
\draw (0,0) to (0,1.5);
\node[zxnode=\zxwhite] at (0,1.5) {$\psi^\dagger$};
\end{tz}
=
\begin{tz}[zx,slave,every to/.style={out=up, in=down}]
\draw (0,1.5) to (0,2) to [in=left] node[pos=1] (r){} (0.5,2.5) to [out=right, in=up] (1,2)  to [out=down, in=up] (1,0);
\node[zxnode=\zxwhite] at (0,1.5) {$\psi$};
\node[zxvertex=\zxwhite,zxdown] at (r.center){};
\end{tz}\\\nonumber
\end{calign}
\end{definition}
\begin{theorem}[{\cite[Theorem 5.1.]{Coecke2009}}] \label{thm:classificationONB}The copyable elements of a special commutative dagger Frobenius algebra $A$ form an orthonormal basis of $A$ for which the monoid is of the form given in Example~\ref{exm:Frob}.
\end{theorem}
\noindent
In other words, every special commutative dagger Frobenius algebra in \Hilb\ is of the form~\eqref{eq:classicalcopy} for some orthonormal basis on a Hilbert space.

Given a special commutative dagger Frobenius algebra $A$, we denote its set of copyable elements by $\widehat{A}$. For such algebras $A$ and $B$, it can easily be verified that every function $\widehat{A}\to \widehat{B}$ gives rise to a $*$-cohomomorphism between $A$ and $B$ and that conversely every $*$-cohomomorphism $A\to B$ comes from such a function $\widehat{A} \to \widehat{B}$. Therefore, Theorem~\ref{thm:classificationONB} gives rise to the following Frobenius-algebraic version of finite-dimensional Gelfand duality.

\begin{corollary}[{\cite[Corollary 7.2.]{Coecke2009}}]The category of special commutative dagger Frobenius algebras and $*$-cohomomorphisms in \Hilb\ is equivalent to the category of finite sets and functions.
\end{corollary}
\noindent
Explicitly, this equivalence maps a special commutative dagger Frobenius algebra $A$ to its set of copyable elements $\widehat{A}$ and a set $X$ to the algebra associated to the orthonormal basis $\{ \ket{x}~|~x\in X\}$ of the Hilbert space $\mathbb{C}^{|X|}$.
%
%
%
Under this correspondence, we may therefore consider the category of finite sets as `contained within $\Hilb$' using the following identification.
\begin{center}
\begin{tabular}{l | l}
$\Set$ & $\mathrm{Hilb}$\\
\hline
sets of cardinality $n$ & special commutative dagger Frobenius algebras of dimension $n$ \\
elements of the set & copyable states of the Frobenius algebra \\
functions & $*$-cohomomorphisms\\
bijections & $*$-isomorphisms\\
the one element set $\{*\}$ & the one-dimensional Frobenius algebra $\mathbb{C}$
\end{tabular}
\end{center}

\begin{terminology}\label{terminology:commutative}Throughout this paper, we will take pairs of words in this table to be synonymous. In particular, we will denote a set and its corresponding commutative algebra by the same symbol. It will always be clear from context whether we refer to the set $X$ or the algebra $X$.
\end{terminology}

\subsection{Quantum graphs and quantum graph isomorphisms}
\label{sec:qgraphsqisos}
The fundamental idea of noncommutative topology is to generalise the correspondence between spaces and commutative algebras by considering noncommutative algebras in light of Gelfand duality.
\begin{terminology}By analogy with Gelfand duality, we think of a special symmetric dagger Frobenius algebra as being associated to an imagined finite quantum set, just as a special commutative dagger Frobenius algebra is associated to a finite set. We follow Terminology~\ref{terminology:commutative} in denoting both the algebra and its associated imagined quantum set by the same symbol.
\end{terminology}
\noindent
We can endow a quantum set with graph structure in the following way.
\begin{definition}\label{def:quantumgraphsbyadjmats}
A \emph{quantum graph} is a pair $(V_{\Gamma},\Gamma)$ of a special symmetric dagger Frobenius algebra $V_{\Gamma}$ (the \emph{quantum set of vertices}) and a self-adjoint linear map $\Gamma:V_{\Gamma} \to V_{\Gamma}$ (the \emph{quantum adjacency matrix}) satisfying the following equations:%
\begin{calign}\label{eq:propadjacency}
\begin{tz}[zx]
\draw (0,0) to (0,0.5) to [out=135, in=-135,looseness=1.5] node[zxnode=\zxwhite, pos=0.5] {$\Gamma$} (0,2.5) to (0,3);
\draw[string] (0,0.5) to [out=45, in=-45,looseness=1.5] node[zxnode=\zxwhite, pos=0.5] {$\Gamma$} (0,2.5);
\node[zxvertex=\zxwhite,zxup] at (0,0.5){};
\node[zxvertex=\zxwhite,zxdown] at (0,2.5){};
\end{tz}\quad  = \quad 
\begin{tz}[zx]
\draw (0,0) to node[zxnode=\zxwhite, pos=0.5] {$\Gamma$} (0,3);
\end{tz}
&
\begin{tz}[zx]
\draw (1,0) to (1,1.5) to [out=up, in=up, looseness=2.5]node[zxvertex=\zxwhite, pos=0.5] {} (0,1.5) to [out=down, in=down, looseness=2.5]node[zxvertex=\zxwhite, pos=0.5] {} (-1,1.5) to (-1,3);
\node[zxnode=\zxwhite] at (0,1.5) {$\Gamma$};
\end{tz}
\quad = \quad 
\begin{tz}[zx]
\draw (0,0) to (0,3);
\node[zxnode=\zxwhite] at (0,1.5) {$\Gamma$};
\end{tz}
&
\begin{tz}[zx]
\draw (0,0) to (0,0.5) to [out=135, in=-135,looseness=1.5] node[zxnode=\zxwhite, pos=0.5] {$\Gamma$} (0,2.5) to (0,3);
\draw[string] (0,0.5) to [out=45, in=-45,looseness=1.5](0,2.5);
\node[zxvertex=\zxwhite,zxup] at (0,0.5){};
\node[zxvertex=\zxwhite,zxdown] at (0,2.5){};
\end{tz}\quad  = \quad 
\begin{tz}[zx]
\draw (0,0) to (0,3);
\end{tz}
\end{calign}
\end{definition}
\noindent
We will often omit the underlying algebra from the notation and denote quantum graphs $(V_{\Gamma},\Gamma)$ simply by $\Gamma$.

For a classical set $V_{\Gamma}$ (that is, for a special commutative dagger Frobenius algebra), Definition~\ref{def:quantumgraphsbyadjmats} reduces to the definition of an adjacency matrix $\{\Gamma_{v,w}\}_{v,w\in V_{\Gamma}}$; from left to right, the conditions state that $\Gamma_{v,w}^2 = \Gamma_{v,w}$, that $\Gamma_{v,w} = \Gamma_{w,v}$, and that $\Gamma_{v,v}=1$. Therefore, a quantum graph defined on a commutative algebra is precisely a graph in the usual sense.

\begin{remark}\label{rem:literatureqgraph}Notions of quantum graph have been defined elsewhere. In~\cite[Section 7]{Musto2017a}, we prove that:
\begin{itemize}
\item Our quantum graphs coincide with Weaver's finite-dimensional quantum graphs~\cite{Weaver2015}, defined in terms of symmetric and reflexive quantum relations~\cite{Kuperberg2012,Weaver2010}.
\item Our quantum graphs $(\mathrm{Mat}_n, \Gamma)$ on matrix algebras coincide with Duan, Severini and Winter's noncommutative graphs~\cite{Duan2013}.
\end{itemize}
\end{remark}

\begin{definition} \label{def:quantumgraphiso}An \emph{isomorphism} of quantum graphs $\Gamma$ and $\Gamma'$ is a $*$-isomorphism of the underlying Frobenius algebras $f:V_{\Gamma} \to V_{\Gamma'}$ intertwining the corresponding quantum adjacency matrices, i.e. such that $f\Gamma = \Gamma' f$. 
We denote the group of automorphisms of a quantum graph $\Gamma$ by $\Aut(\Gamma)$.
\end{definition}
\noindent
For classical graphs, Definition~\ref{def:quantumgraphiso} coincides with the usual notion of graph isomorphism. In particular, for a classical graph $\Gamma$, the group $\Aut(\Gamma)$ is the usual automorphism group. 

\subsubsection{Quantum isomorphisms}

We now come to the central definition of this work.
\begin{definition}\label{def:quantumfunction} A \textit{quantum isomorphism} between quantum graphs $\Gamma$ and $\Gamma'$ is a pair ($H,P$), where $H$ is a Hilbert space and $P$ is a linear map $P: H \otimes V_{\Gamma}\! \to\! V_{\Gamma'} \otimes H$ satisfying the following equations, where the algebras $V_\Gamma$ and $V_{\Gamma'}$ are depicted as white and grey nodes respectively:
\begin{calign}\label{eq:quantumfunction} \begin{tz}[zx,xscale=-1,every to/.style={out=up, in=down}]
\draw (0,0) to (0,2) to [out=45] (0.75,3);
\draw (0,2) to [out=135] (-0.75,3);
\draw[arrow data={0.2}{>}, arrow data={0.8}{>}] (1.75,0) to [looseness=0.9] node[zxnode=\zxwhite, pos=0.5] {$P$} (-1.75,2.5) to (-1.75,3);
\node[zxvertex=\zxblack, zxup] at (0,2){};
\end{tz}
=
\begin{tz}[zx,xscale=-1,every to/.style={out=up, in=down}]
\draw (0,0) to (0,0.75) to [out=45] (0.75,1.75) to (0.75,3);
\draw (0,0.75) to [out=135] (-0.75,1.75) to (-0.75,3);
\draw[arrow data={0.2}{>}, arrow data={0.9}{>}] (1.75,0) to (1.75,0.75) to  [looseness=1.1, in looseness=0.9] node[zxnode=\zxwhite, pos=0.36] {$P$} node[zxnode=\zxwhite, pos=0.64] {$P$}(-1.75,3);
\node[zxvertex=\zxwhite, zxup] at (0,0.75){};
\end{tz}
&
\begin{tz}[zx,xscale=-1,every to/.style={out=up, in=down}]
\draw (0,0) to (0,2.25);
\draw[arrow data={0.2}{>}, arrow data={0.8}{>}] (1,0) to [looseness=0.9] node[zxnode=\zxwhite, pos=0.5] {$P$} (-1,2.5) to (-1,3);
\node[zxvertex=\zxblack] at (0,2.25){};
\end{tz}
=
\begin{tz}[zx,xscale=-1,every to/.style={out=up, in=down}]
\draw (0,0) to (0,0.75);
\draw[arrow data={0.2}{>}, arrow data={0.9}{>}] (1.,0) to (1.,0.75) to   (-1,3);
\node[zxvertex=\zxwhite, zxup] at (0,0.75){};
\end{tz}
&
\begin{tz}[zx,xscale=-1,every to/.style={out=up, in=down},xscale=0.8]
\draw [arrow data={0.2}{>},arrow data={0.8}{>}]  (0,0) to (2.25,3);
\draw (2.25,0) to node[zxnode=\zxwhite, pos=0.5] {$P^\dagger$} (0,3);
\end{tz}
=~~
\begin{tz}[zx,xscale=-1,xscale=0.6,yscale=-1]
\draw[arrow data={0.5}{<}] (0.25,-0.5) to (0.25,0) to [out=up, in=-135] (1,1);
\draw (1,1) to [out=135, in=right] node[zxvertex=\zxwhite, pos=1]{} (-0.3, 1.7) to [out=left, in=up] (-1.25,1) to (-1.25,-0.5);
\draw[arrow data={0.5}{>}] (1.75,2.5) to (1.75,2) to [out=down, in=45] (1,1);
\draw (1,1) to [out= -45, in= left] node[zxvertex=\zxblack, pos=1] {} (2.3,0.3) to [out=right, in=down] (3.25,1) to (3.25,2.5);
\node [zxnode=\zxwhite] at (1,1) {$P$};
\end{tz} 
\end{calign}
\begin{calign}\label{eq:quantumfunction2}
\begin{tz}[zx,xscale=-1,every to/.style={out=up, in=down},scale=-1]
\draw (0,0) to (0,2) to [out=45] (0.75,3);
\draw (0,2) to [out=135] (-0.75,3);
\draw[arrow data={0.2}{<}, arrow data={0.8}{<}] (1.75,0) to [looseness=0.9] node[zxnode=\zxwhite, pos=0.5] {$P$} (-1.75,2.5) to (-1.75,3);
\node[zxvertex=\zxwhite, zxdown] at (0,2){};
\end{tz}
=
\begin{tz}[zx,xscale=-1,every to/.style={out=up, in=down},scale=-1]
\draw (0,0) to (0,0.75) to [out=45] (0.75,1.75) to (0.75,3);
\draw (0,0.75) to [out=135] (-0.75,1.75) to (-0.75,3);
\draw[arrow data={0.2}{<}, arrow data={0.9}{<}] (1.75,0) to (1.75,0.75) to  [looseness=1.1, in looseness=0.9] node[zxnode=\zxwhite, pos=0.36] {$P$} node[zxnode=\zxwhite, pos=0.64] {$P$}(-1.75,3);
\node[zxvertex=\zxblack, zxdown] at (0,0.75){};
\end{tz}
&
\begin{tz}[zx,xscale=-1,every to/.style={out=up, in=down},scale=-1]
\draw (0,0) to (0,2.25);
\draw[arrow data={0.2}{<}, arrow data={0.8}{<}] (1,0) to [looseness=0.9] node[zxnode=\zxwhite, pos=0.5] {$P$} (-1,2.5) to (-1,3);
\node[zxvertex=\zxwhite] at (0,2.25){};
\end{tz}
=
\begin{tz}[zx,xscale=-1,every to/.style={out=up, in=down},scale=-1]
\draw (0,0) to (0,0.75);
\draw[arrow data={0.2}{<}, arrow data={0.9}{<}] (1.,0) to (1.,0.75) to   (-1,3);
\node[zxvertex=\zxblack] at (0,0.75){};
\end{tz}
&
\begin{tz}[zx,xscale=-1,every to/.style={out=up, in=down},scale=1]
\draw (0,0) to (0,3);
\draw[arrow data={0.2}{>}, arrow data={0.9}{>}] (1.75,0) to (1.75,0.75) to  [looseness=1.1, in looseness=0.9]  node[zxnode=\zxwhite, pos=0.5] {$P$}(-1.75,3);
\node[zxnode=\zxwhite] at (0,0.9){$\Gamma$};
\end{tz}
=
\begin{tz}[zx,xscale=-1,every to/.style={out=up, in=down},scale=1]
\draw (0,0) to (0,3);
\draw[arrow data={0.2}{>}, arrow data={0.8}{>}] (1.75,0) to [looseness=0.9] node[zxnode=\zxwhite, pos=0.5] {$P$} (-1.75,2.5) to (-1.75,3);
\node[zxnode=\zxwhite] at (0,2.35) {$\Gamma'$};
\end{tz}
\end{calign}
The \emph{dimension} of a quantum isomorphism is defined as the dimension of the underlying Hilbert space $H$. 
\end{definition}
\begin{notation} To clearly distinguish between the wires corresponding to the Hilbert space $H$ and the wires corresponding to the algebras $V_{\Gamma}$ and $V_{\Gamma'}$, we will always draw the Hilbert space wire with an orientation and leave the algebra wires unoriented.
\end{notation}
\begin{remark} There are classical and quantum isomorphisms between classical graphs, and classical (see Definition~\ref{def:quantumgraphiso}) and quantum isomorphisms between quantum graphs. 
\end{remark}
\begin{remark}
A one-dimensional quantum isomorphism between quantum graphs is an ordinary isomorphism (see Definition~\ref{def:quantumgraphiso}). In particular, a one-dimensional quantum isomorphism between classical graphs is a graph isomorphism. 
\end{remark}
\noindent 
A quantum isomorphism $(H,P): \Gamma\to \Gamma'$ between classical graphs with adjacency matrices $\{\Gamma_{v,v'}\}_{v,v'\in V_{\Gamma}}$ and $\{\Gamma'_{w,w'}\}_{w,w'\in V_{\Gamma'}}$ can equivalently be expressed as a family of projectors $\{P_{v,w}\}_{v\in V_{\Gamma}, w \in V_{\Gamma'}}$ on $H$ such that the following holds for all vertices $v,v_1,v_2\in V_{\Gamma}$ and $w,w_1,w_2\in V_{\Gamma'}$:
\begin{calign}\label{eq:PPM1}P_{v,w_1} P_{v,w_2} = \delta_{w_1,w_2} P_{v,w_1} 
& 
\sum_{w\in V_{\Gamma'}} P_{v,w} = \mathbbm{1}_H
\\[-2pt]
P_{v_1,w} P_{v_2,w} = \delta_{v_1,v_2} P_{v_1,w} 
& 
\sum_{v\in V_{\Gamma}} P_{v,w} = \mathbbm{1}_H
\end{calign}\vspace{-15pt}
\begin{calign} \sum_{v' \in V_{\Gamma} }\Gamma_{v,v'}P_{v',w} ~=~ \sum_{w'\in V_{\Gamma'}} P_{v,w'} \Gamma'_{w',w}
\end{calign}
We will refer to such families of projectors as \emph{projective permutation matrices}~\cite{Atserias2016}.
Given a quantum isomorphism $(H,P):\Gamma \to \Gamma'$ between classical graphs, the corresponding projective permutation matrix can be obtained as follows. A classical set $X$ corresponds to a special commutative dagger Frobenius algebra (Example~\ref{exm:Frob}); the elements of $X$ form a basis of copyable elements of this algebra. Using this basis, the projectors $P_{x,y}$ can be obtained as follows:
\begin{equation}\label{eq:componentPPM}
\begin{tz}[zx]
\draw[arrow data={0.2}{>}, arrow data={0.8}{>}] (0,0) to (0,3);
\node[zxnode=\zxwhite] at (0,1.5) {$P_{x,y}$};
\end{tz}
~~:=
\begin{tz}[zx,xscale=-1,every to/.style={out=up, in=down},xscale=-0.8]
\draw [arrow data={0.2}{>},arrow data={0.8}{>}]  (0,0) to (2.25,3);
\draw (1.75,0.5) to node[zxnode=\zxwhite, pos=0.5] {$P$} (0.5,2.5);
\node[zxnode=\zxwhite] at (1.75,0.5){$x$};
\node[zxnode=\zxwhite] at (0.5,2.5){$y$};
\end{tz}
\end{equation}

\noindent
Like ordinary isomorphisms, quantum isomorphisms $(H,P):\Gamma\to \Gamma'$ can only exist between quantum graphs with quantum vertex sets of equal dimension.

\begin{proposition}[{\cite[Proposition 4.17]{Musto2017a}}]\label{prop:dimensionpreserved}If there is a quantum isomorphism $(H,P):\Gamma \to \Gamma'$, then $\dim(V_{\Gamma}) = \dim(V_{\Gamma'})$. In particular, quantum isomorphisms can only exist between classical graphs with an equal number of vertices. 
\end{proposition}

\subsubsection{The $2$-category $\QGraphIso$}
Quantum graphs and quantum isomorphisms can be organised into a $2$-category. The $2$-morphisms of this $2$-category are defined as follows.
\begin{definition}\label{def:intertwiner} An \textit{intertwiner} of quantum isomorphisms $(H,P) \to (H', P')$ is a linear map $f:H\to H'$ such that the following holds:%
\begin{calign}
\begin{tz}[zx,xscale=-1,every to/.style={out=up, in=down}]
\draw (0,0) to (0,1) to (2,3);
\draw[arrow data={0.35}{>}] (2,0) to node[zxnode=\zxwhite, pos=0.9] {$f$} (2,1);
\draw[string,arrow data={0.9}{>},arrow data={0.26}{>}](2,1) to node[zxnode=\zxwhite, pos=0.5]{$P'$}  (0,3);
\end{tz}
\quad = \quad
\begin{tz}[zx,xscale=-1,every to/.style={out=up, in=down},scale=-1]
\draw (0,0) to (0,1) to (2,3);
\draw[arrow data={0.35}{<}] (2,0) to node[zxnode=\zxwhite, pos=0.9] {$f$} (2,1);
\draw[string,arrow data={0.9}{<},arrow data={0.26}{<}](2,1) to node[zxnode=\zxwhite, pos=0.5]{$P$}  (0,3);
\end{tz}
\end{calign}
\end{definition}

\begin{definition}[{\cite[Definition 3.18 and Theorem 3.20]{Musto2017a}}]\label{def:2catqgraph} The dagger 2-category $\QGraphIso$ is defined as follows:
\begin{itemize}
\item \textbf{objects} are quantum graphs $\Gamma, \Gamma'$, ...;
\item \textbf{1-morphisms} $\Gamma\to \Gamma'$ are quantum isomorphisms $(H,P):\Gamma \to \Gamma'$;
\item \textbf{2-morphisms} $(H,P)\to  (\!H', P')$ are intertwiners of quantum isomorphisms.
\end{itemize}
The composition of two quantum isomorphisms $(H,P):A\to  B$ and $(H', Q): B \to C$ is a quantum isomorphism $(H'\otimes H, Q\circ P)$ defined as follows:
\begin{calign}
\begin{tz}[zx,xscale=-1, every to/.style={out=up, in=down}] \label{eq:1composition}
\draw[arrow data={0.15}{>}, arrow data={0.8}{>}] (1.575,0) to (0.325,3.5);
\draw[string,arrow data={0.18}{>}, arrow data={0.85}{>}] (2.175,0) to (0.925,3.5);
\draw (0,0) to node[zxvertex=\zxwhite, pos=0.5] {$Q\circ P$} (2.5,3.5);
\node[dimension,left] at (2.175,0) {$H'\otimes H$};
\end{tz}
\! := \quad
\begin{tz}[zx,xscale=-1, every to/.style={out=up, in=down}]
\draw (0,0) to (2.5,3.5);
\draw[arrow data={0.15}{>}, arrow data={0.8}{>}] (1.25,0) to node[zxnode=\zxwhite, pos=0.4] {$P$} (0,3.5);
\draw[string,arrow data={0.2}{>}, arrow data={0.9}{>}] (2.5,0) to node[zxnode=\zxwhite, pos=0.58] {$Q$} (1.25,3.5);
\node[dimension,left] at (1.25,0) {$H$};
\node[dimension, left] at (2.5,0) {$H'$};
\end{tz}
\end{calign}
Vertical and horizontal composition of 2-morphisms is defined as the ordinary composition and tensor product of linear maps, respectively. The $\dagger$-adjoint of a 2-morphism is defined as the Hilbert space adjoint of the underlying linear map.
\end{definition}
\noindent
In~\cite{Musto2017a}, we define a $2$-category $\QGraph$ of quantum graphs and quantum \emph{homomorphisms}. For the purpose of this work, it suffices to focus on quantum isomorphisms.

This $2$-category $\QGraphIso$ has the advantage that every $1$-morphism is dualisable.
\begin{theorem}[{\cite[Theorem 4.8]{Musto2017a}}]\label{thm:dualisable}
Every quantum isomorphism $(H,P): \Gamma \to \Gamma'$ is dualisable in $\QGraphIso$. In particular, this means that there is a quantum isomorphism $(H^*, \overline{P}): \Gamma' \to \Gamma$, whose underlying linear map $\overline{P}:H^* \otimes V_{\Gamma'}  \to   V_{\Gamma} \otimes H^*$ is defined by equation~\eqref{eq:daggerdual} and fulfils equations~\eqref{eq:rightdual} and~\eqref{eq:leftdual}.
\begin{equation}\label{eq:daggerdual}
\begin{tz}[zx, xscale=-1,every to/.style={out=up, in=down},xscale=-0.8]
\draw [arrow data={0.2}{<},arrow data={0.8}{<}]  (0,0) to (2.25,3);
\draw (2.25,0) to node[zxnode=\zxwhite, pos=0.5] {$\overline{P}$} (0,3);
\end{tz}
~~~:=~~~
\begin{tz}[zx, xscale=-1,xscale=0.6,xscale=-1]
\draw(0.25,-0.5) to (0.25,0) to [out=up, in=-135] (1,1);
\draw [arrow data={0.34}{>}]  (1,1) to [out=135, in=right]  (-0.3, 1.7) to [out=left, in=up] (-1.25,1) to (-1.25,-0.5);
\draw(1.75,2.5) to (1.75,2) to [out=down, in=45] (1,1);
\draw[arrow data={0.32}{<}] (1,1) to [out= -45, in= left]   (2.3,0.3) to [out=right, in=down] (3.25,1) to (3.25,2.5);
\node [zxnode=\zxwhite] at (1,1) {$P^\dagger$};
\end{tz} 
~~~=~~~
\begin{tz}[zx, xscale=-1,xscale=0.6,xscale=-1]
\draw[arrow data={0.68}{>}]  (4.75,2.5) to (4.75, 0.5) to [out=down, in=right] (2.65,-0.6) to [out=left, in=down] (0.55,0.5);
\draw (0.55,0.5) to [out=up, in=-135] (1,1);
\draw  (1,1) to [out=135, in=right]  (-0.3, 1.7) to [out=left, in=up]  node[zxvertex=\zxblack, pos=0] {}(-1.25,1) to (-1.25,-0.5);
\draw[string, arrow data={0.68}{<}] (-2.75, -0.5) to (-2.75,1.5) to [out=up, in=left] (-0.65, 2.6) to [out=right, in=up] (1.45,1.5);
\draw  (1.45,1.5) to [out=down, in=45] (1,1);
\draw (1,1) to [out= -45, in= left]  node[zxvertex=\zxwhite, pos=1] {}  (2.3,0.3) to [out=right, in=down] (3.25,1) to (3.25,2.5);
\node [zxnode=\zxwhite] at (1,1) {$P$};
\end{tz} 
\end{equation}
\def\scl{0.75}
\begin{calign}\label{eq:rightdual}
\begin{tz}[zx, xscale=-1,scale=\scl]
\clip (-2.05,-0.05) rectangle (4.05,4.05);
\draw[arrow data={0.1}{>},arrow data={0.499}{>}, arrow data={0.9}{>}] (0,0) to [out=up, in=up, looseness=6.] (2,0);
\draw (-2,0) to [out=up, in=down]  node[zxnode=\zxwhite, pos=0.4] {$P$} node[zxnode=\zxwhite, pos=0.6] {$\overline{P}$}(4,4);
\end{tz}
=
\begin{tz}[zx, xscale=-1,scale=\scl]
\clip (-2.05,-0.05) rectangle (4.05,4.05);
\draw[arrow data={0.2}{>}, arrow data={0.8}{>}] (0,0) to [out=up, in=up, looseness=2.] (2,0);
\draw (-2,0) to [out=up, in=down]  (4,4);
\end{tz}
&
\begin{tz}[zx, xscale=-1,scale=-1,scale=\scl]
\clip (-2.05,-0.05) rectangle (4.05,4.05);
\draw[arrow data={0.1}{<},arrow data={0.499}{<}, arrow data={0.9}{<}] (0,0) to [out=up, in=up, looseness=6.] (2,0);
\draw (-2,0) to [out=up, in=down]  node[zxnode=\zxwhite, pos=0.4] {$P$} node[zxnode=\zxwhite, pos=0.6] {$\overline{P}$}(4,4);
\end{tz}
=
\begin{tz}[zx, xscale=-1,scale=-1,scale=\scl]
\clip (-2.05,-0.05) rectangle (4.05,4.05);
\draw[arrow data={0.2}{<}, arrow data={0.8}{<}] (0,0) to [out=up, in=up, looseness=2.] (2,0);
\draw (-2,0) to [out=up, in=down]  (4,4);
\end{tz}
\end{calign}
\begin{calign}\label{eq:leftdual}
\begin{tz}[zx, xscale=-1,master,scale=-1,scale=\scl]
\clip (-2.05,-0.05) rectangle (4.05,4.05);
\draw[arrow data={0.1}{>},arrow data={0.499}{>}, arrow data={0.9}{>}] (0,0) to [out=up, in=up, looseness=6.] (2,0);
\draw (-2,0) to [out=up, in=down]  node[zxnode=\zxwhite, pos=0.4] {$\overline{P}$} node[zxnode=\zxwhite, pos=0.6] {$P$}(4,4);
\end{tz}
=
\begin{tz}[zx, xscale=-1,scale=-1,scale=\scl]
\clip (-2.05,-0.05) rectangle (4.05,4.05);
\draw[arrow data={0.2}{>}, arrow data={0.8}{>}] (0,0) to [out=up, in=up, looseness=2.] (2,0);
\draw (-2,0) to [out=up, in=down]  (4,4);
\end{tz}
&
\begin{tz}[zx, xscale=-1,scale=1,scale=\scl]
\clip (-2.05,-0.05) rectangle (4.05,4.05);
\draw[arrow data={0.1}{<},arrow data={0.499}{<}, arrow data={0.9}{<}] (0,0) to [out=up, in=up, looseness=6.] (2,0);
\draw (-2,0) to [out=up, in=down]  node[zxnode=\zxwhite, pos=0.4] {$\overline{P}$} node[zxnode=\zxwhite, pos=0.6] {$P$}(4,4);
\end{tz}
=
\begin{tz}[zx, xscale=-1,scale=1,scale=\scl]
\clip (-2.05,-0.05) rectangle (4.05,4.05);
\draw[arrow data={0.2}{<}, arrow data={0.8}{<}] (0,0) to [out=up, in=up, looseness=2.] (2,0);
\draw (-2,0) to [out=up, in=down]  (4,4);
\end{tz}
\end{calign}
In particular, the linear map $P:H \otimes V_{\Gamma} \to V_{\Gamma'} \otimes H$ is unitary.
\end{theorem}
\begin{proposition}[{\cite[Proposition 4.2]{Musto2017a}}]\label{prop:equivalenceQGraph}
Equivalences in $\QGraphIso$ are ordinary isomorphisms as in Definition~\ref{def:quantumgraphiso}.
\end{proposition}
\ignore{
\noindent
We also remark that endomorphism categories of classical graphs in $\QGraphIso$ have appeared in the theory of compact quantum groups.

}



\subsection{The monoidal dagger category $\QAut(\Gamma)$}
\label{sec:catqaut}
For a quantum graph $\Gamma$, we write $\QAut(\Gamma)$ for the monoidal dagger category $\QGraphIso(\Gamma,\Gamma)$ of quantum automorphisms of $\Gamma$. 
For classical graphs $\Gamma$, the category $\QAut(\Gamma)$ (or rather the Hopf $C^*$-algebra for which it is the category of finite-dimensional representations) has been studied in the context of compact quantum groups~\cite{Banica2005,Banica2009,Banica2007_2,Banica2007_3,Bichon2003,Banica2008}.
\begin{proposition}[{\cite[Proposition 5.19]{Musto2017a}}]\label{prop:Banica} Let $\Gamma$ be a classical graph. The category $\QAut(\Gamma)$ is the category of finite-dimensional representations of Banica's quantum automorphism algebra $A(\Gamma)$ of the graph $\Gamma$ (see e.g.~\cite[Definition 2.1]{Banica2007_2}).
\end{proposition}

\noindent
In particular, $\QAut(\Gamma)$ is \emph{semisimple} (see~\cite[Corollary 6.21]{Musto2017a}). The \textit{direct sum} of two quantum automorphisms $(H,P), (H',Q): \Gamma \to \Gamma$ is defined as the direct sum of the underlying linear maps:
\begin{calign}\label{eq:directsum}
\begin{tz}[zx, xscale=-1,every to/.style={out=up, in=down}]
\draw  (0,1) to (2,3);
\draw[string,arrow data={0.9}{>},arrow data={0.2}{>}](2,1) to node[zxnode=\zxwhite, pos=0.5]{$P\oplus Q$}  (0,3);
\node[dimension, right] at (0,1) {$V_{\Gamma}$};
\node[dimension, right] at (0,3) {$H\oplus H'$};
\node[dimension,left] at (2,3) {$V_{\Gamma}$};
\node[dimension, left] at (2,1) {$H\oplus H'$};
\end{tz}
=
\begin{tz}[zx, xscale=-1,every to/.style={out=up, in=down}]
\draw  (0,1) to (2,3);
\draw[string,arrow data={0.9}{>},arrow data={0.26}{>}](2,1) to node[zxnode=\zxwhite, pos=0.5]{$P$}  (0,3);
\node[dimension, right] at (0,1) {$V_{\Gamma}$};
\node[dimension, right] at (0,3) {$H$};
\node[dimension, left] at (2,3) {$V_{\Gamma}$};
\node[dimension, left] at (2,1) {$H$};
\end{tz}
\oplus 
\begin{tz}[zx, xscale=-1,every to/.style={out=up, in=down}]
\draw  (0,1) to (2,3);
\draw[string,arrow data={0.9}{>},arrow data={0.26}{>}](2,1) to node[zxnode=\zxwhite, pos=0.5]{$Q$}  (0,3);
\node[dimension,right] at (0,1) {$V_{\Gamma}$};
\node[dimension, right] at (0,3) {$H'$};
\node[dimension, left] at (2,3) {$V_{\Gamma}$};
\node[dimension, left] at (2,1) {$H'$};
\end{tz}
\end{calign}
\noindent
Conversely, a quantum isomorphism $(H,P)$ is \emph{simple} if it cannot be further decomposed into a non-trivial direct sum or equivalently, if it has no non-trivial interchangers, i.e. if $\QGraphIso((H,P), (H,P)) \cong \mathbb{C}$. Semisimplicity implies that every quantum isomorphism $\Gamma \to \Gamma$ is isomorphic to a direct sum of simple quantum isomorphisms.  The decomposition is unique up to permutation of the summands.

\begin{remark}
By dimensional considerations, every ordinary isomorphism is a simple quantum isomorphism. However, in general not all simple quantum isomorphisms are ordinary isomorphisms.
\end{remark}
\noindent
Under composition of quantum isomorphisms, $\QAut(\Gamma)$ becomes a \emph{monoidal} semisimple dagger category. In particular, since all quantum isomorphisms are dualisable, we obtain a monoidal semisimple dagger category with dualisable objects. For a finite number of simple objects such a structure is known as a \textit{unitary fusion category}\footnote{For fusion categories, it is additionally required that the monoidal unit is simple, which is straightforward to verify in our setting.}~\cite{Etingof2015}. In general, however, the number of simple objects of $\QAut(\Gamma)$ is not finite.

\begin{definition}\label{def:classical} The \emph{classical subcategory} of $\QAut(\Gamma)$ is the full semisimple monoidal subcategory of quantum automorphisms which are decomposable into a direct sum of classical automorphisms.
\end{definition}
\noindent
In other words, a quantum automorphism $(H,P)$ in the classical subcategory is of the following form, where $\{\ket{i}\}$ is an orthonormal basis corresponding to the decomposition of the Hilbert space $H$ into one-dimensional subspaces $H\cong \bigoplus_{i} \mathbb{C} \ket{i}$ and $f_i:\Gamma \to \Gamma$ are classical automorphisms:
\begin{equation}\label{eq:classicalquantumfunction}\begin{tz}[zx, xscale=-1,every to/.style={out=up, in=down}]
\draw (0,0) to (2,3);
\draw[arrow data ={0.2}{>}, arrow data={0.8}{>}] (2,0) to (0,3);
\node[zxnode=\zxwhite] at (1,1.5) {$P$};
\node[dimension,right] at (0,0) {$V_{\Gamma}$};
\node[dimension,left] at (2,3) {$V_{\Gamma}$};
\node[dimension,left] at (2,0) {$H$};
\node[dimension,right] at (0,3){$H$};
\end{tz}
~=~\sum_i~\begin{tz}[zx, xscale=-1,every to/.style={out=up, in=down}]
\draw (0,0) to (2,3);
\draw[arrow data ={0.8}{>}] (0,2.) to (0,3);
\draw[arrow data ={0.4}{>}] (2,0) to (2,1);
\node[zxnode=\zxwhite] at (1,1.5) {$f_i$};
\node[zxnode=\zxwhite] at (0,2.) {$i$};
\node[zxnode=\zxwhite] at (2,1) {$i^\dagger$};
\end{tz}
\end{equation}
We note that a quantum isomorphism between classical graphs is in the classical subcategory if and only if all projectors in its projective permutation matrix commute with each other~ \cite[Proposition 6.9]{Musto2017a}.

\begin{definition}
For a finite group $G$, we denote the unitary fusion category of \emph{$G$-graded Hilbert spaces} by $\Hilb_G$; its simple objects are the group elements of $G$ and the monoidal product is induced by group multiplication.\footnote{The category $\Hilb_G$ is the dagger analogue of the fusion category $\Vect_G$ of $G$-graded vector spaces~\cite{Etingof2015} in which every vector space is equipped with an inner product compatible with the grading.}
\end{definition}
\noindent 
It is not hard to see that the classical subcategory of $\QAut(\Gamma)$ is equivalent to the unitary fusion category $\Hilb_{\Aut(\Gamma)}$. We therefore have a full inclusion $\Hilb_{\Aut(\Gamma)} \subseteq \QAut(\Gamma)$. In general the inclusion is strict; there will be simple quantum automorphisms which are \emph{not} one-dimensional. However, for some graphs this is not the case.   \begin{definition}[\cite{Banica2007_3}]\label{def:noquantumsymmetries} A quantum graph $\Gamma$ is said to have \textit{no quantum symmetries} if every quantum automorphism is in the classical subcategory, i.e. if $\QAut(\Gamma) \cong \Hilb_{\Aut(\Gamma)}$; or equivalently, if all simple quantum automorphisms are 1-dimensional.
\end{definition}

\subsection{A rapid introduction to Morita theory}\label{app:Frobenius}
We now recall the theory of Morita equivalence in monoidal dagger categories. Similar expositions can be found in a variety of contexts~\cite{Carqueville2016,Kong2008,Heunen2014_2}. In the following, we focus on special dagger Frobenius monoids in monoidal dagger categories; however, most definitions and statements below have analogues for more general Frobenius monoids in monoidal categories.

\begin{definition}\label{def:daggeridempotentsandsplitting}A \emph{dagger idempotent} in a dagger category is an endomorphism $p:A\to A$ such that $p^2 = p$ and $p^\dagger = p$. We say that a dagger idempotent \textit{splits}, if there is an object $V$ and a morphisms $i:V\to A$ such that $p = i i^\dagger $ and $i^\dagger i = 1_V$.
\end{definition}
\begin{example}\label{ex:daggeridempotentsinhilb}
A dagger idempotent in the dagger category $\Hilb$ of finite-dimensional Hilbert spaces and linear maps is an orthogonal projection. Dagger splitting expresses the projector as a map onto the image composed with its adjoint.
\end{example}
\noindent The splitting of an indempotent is unique up to a unitary isomorphism: Indeed, if $(i,V)$ and $(i',V')$ split the same idempotent, then $U= i^\dagger i' :V' \to V$ is unitary with $i' = i U$. 

\def\d{0.5}
\def\h{2.25}
\def\inang{-45}
\begin{definition} \label{def:daggerbimodule}Let $A$ and $B$ be special dagger Frobenius monoids in a monoidal dagger category. An $A{-}B$-\textit{dagger bimodule} is an object $M$ together with an morphism $\rho:A\otimes M\otimes B \to M$ fulfilling the following equations:
\begin{calign}\label{eq:bimodule}
\begin{tz}[zx,every to/.style={out=up, in=down}]
\draw (0,0) to (0,3);
\draw (-\d-1.5,0) to [in=-135] (-\d-0.75,1) to[in=180-\inang] (0,\h);
\draw (-\d,0) to [in=-45] (-\d-0.75,1) ;
\draw (\d,0) to [in=-135] (\d+0.75,1) to [in=\inang] (0,\h);
\draw (\d+1.5,0) to [in=-45] (\d+0.75,1) ;
\node[zxvertex=\zxwhite, zxdown] at (-\d-0.75,1){};
\node[zxvertex=\zxwhite, zxdown] at (\d+0.75,1){};
\node[box,zxdown] at (0,\h) {$\rho$};
\end{tz}
~~=~~
\def\htop{2.25}
\def\hbot{1.25}
\begin{tz}[zx,every to/.style={out=up, in=down}]
\draw (0,0) to (0,3);
\draw (-\d-1.5,0) to [in=-135] (0,\htop);
\draw (-\d,0) to [in=-135] (0,\hbot);
\draw (\d,0) to [in=-45] (0,\hbot);
\draw (\d+1.5,0) to [in=-45] (0,\htop);
\node[box,zxdown] at (0,\hbot) {$\rho$};
\node[box,zxdown] at (0,\htop) {$\rho$};
\end{tz}
&
\begin{tz}[zx, every to/.style={out=up, in=down}]
\draw (0,0) to (0,3);
\draw (-\d,1.2) to [in=-135] (0,\h);
\draw (\d,1.2) to [in=-45] (0,\h);
\node[box,zxdown] at (0,\h) {$\rho$};
\node[zxvertex=\zxwhite] at (-\d,1.2){};
\node[zxvertex=\zxwhite] at (\d,1.2){};
\end{tz}
~~=~~~
\begin{tz}[zx, every to/.style={out=up, in=down}]
\draw (0,0) to (0,3);
\end{tz}
&
\def\x{0.2}
\begin{tz}[zx, every to/.style={out=up, in=down},xscale=0.8]
\draw (0,0) to (0,3);
\draw (-\x,1.5) to [out=up, in=right] (-0.75-\x, 2.25) to [out=left, in=up] (-1.5-\x, 1.5) to (-1.5-\x,0);
\draw (\x,1.5) to [out=up, in=left] (0.75+\x, 2.25) to [out=right, in=up] (1.5+\x, 1.5) to (1.5+\x,0);
\node[zxvertex=\zxwhite] at (-0.75-\x, 2.25){};
\node[zxvertex=\zxwhite] at (0.75+\x, 2.25){};
\node[box] at (0,1.5) {$\rho^\dagger$};
\end{tz}
~~=~~
\begin{tz}[zx, every to/.style={out=up, in=down},xscale=0.8]
\draw (0,0) to (0,3);
\draw (-1.25, 0) to [in=-135] (0,1.95) ;
\draw (1.25,0) to [in=-45] (0,1.95);
\node[box] at (0,1.95) {$\rho$};
\end{tz}
\end{calign}
\end{definition}
\noindent
We usually denote an $A{-}B$-dagger bimodule $M$ by $_AM_B$. For a dagger bimodule ${}_AM_B$, we introduce the following shorthand notation:
\begin{calign}\label{eq:commute}
\begin{tz}[zx, every to/.style={out=up, in=down}]
\draw (0,0) to (0,3);
\draw (1,0) to [in=-45] (0,\h);
\node[boxvertex,zxdown] at (0,\h) {};
\end{tz}
~:=~\begin{tz}[zx, every to/.style={out=up, in=down}]
\draw (0,0) to (0,3);
\draw (-\d,1.2) to [in=-135] (0,\h);
\draw (1,0) to [in=-45] (0,\h);
\node[box,zxdown] at (0,\h) {$\rho$};
\node[zxvertex=\zxwhite] at (-\d,1.2){};
\end{tz}
&
\begin{tz}[zx, every to/.style={out=up, in=down},xscale=-1]
\draw (0,0) to (0,3);
\draw (1,0) to [in=-45] (0,\h);
\node[boxvertex,zxdown] at (0,\h) {};
\end{tz}
~:=~\begin{tz}[zx, every to/.style={out=up, in=down},xscale=-1]
\draw (0,0) to (0,3);
\draw (-\d,1.2) to [in=-135] (0,\h);
\draw (1,0) to [in=-45] (0,\h);
\node[box,zxdown] at (0,\h) {$\rho$};
\node[zxvertex=\zxwhite] at (-\d,1.2){};
\end{tz}
&
\begin{tz}[zx, every to/.style={out=up, in=down}]
\draw (0,0) to (0,3);
\draw (1,0) to [in=-45] (0,\h);
\draw (-1,0) to [in=-135] (0,\h);
\node[boxvertex,zxdown] at (0,\h) {};
\end{tz}
~:=~\begin{tz}[zx, every to/.style={out=up, in=down}]
\draw (0,0) to (0,3);
\draw (-1,0) to [in=-135] (0,\h);
\draw (1,0) to [in=-45] (0,\h);
\node[box,zxdown] at (0,\h) {$\rho$};
\end{tz}
~=~
\begin{tz}[zx, every to/.style={out=up, in=down}]
\draw (0,0) to (0,3);
\draw (1,0) to [in=-45] (0,\h);
\draw (-1,0) to [in=-135] (0, 1.5);
\node[boxvertex,zxdown] at (0,1.5){};
\node[boxvertex,zxdown] at (0,\h) {};
\end{tz}
~=~
\begin{tz}[zx, every to/.style={out=up, in=down},xscale=-1]
\draw (0,0) to (0,3);
\draw (1,0) to [in=-45] (0,\h);
\draw (-1,0) to [in=-135] (0, 1.5);
\node[boxvertex,zxdown] at (0,1.5){};
\node[boxvertex,zxdown] at (0,\h) {};
\end{tz}
\end{calign}
\ignore{It is easy to see that by the `only-left' and `only-right' actions the $A-B$-bimodule structure induces left $A$-module and right $B$-module structures on $M$.}
\noindent
Every special dagger Frobenius monoid $A$ gives rise to a trivial dagger bimodule ${}_AA_A$:
\begin{calign}
\begin{tz}[zx, every to/.style={out=up, in=down}]
\draw (0,0) to (0,3);
\draw (-1,0) to [in=-135] (0,2.);
\draw (1,0) to [in=-45] (0,2.);
\node[boxvertex,zxdown] at (0,2.){};
\end{tz}
~~:= ~~
\begin{tz}[zx]
\coordinate(A) at (0.25,0);
\draw (1,1) to [out=up, in=-135] (1.75,2);
\draw (1.75,2) to [out=-45, in=up] (3.25,0);
\draw (1.75,2) to (1.75,3);
\mult{A}{1.5}{1}
\node[zxvertex=\zxwhite,zxdown] at (1.75,2){};
\end{tz}
~~= ~~
\begin{tz}[zx,xscale=-1]
\coordinate(A) at (0.25,0);
\draw (1,1) to [out=up, in=-135] (1.75,2);
\draw (1.75,2) to [out=-45, in=up] (3.25,0);
\draw (1.75,2) to (1.75,3);
\mult{A}{1.5}{1}
\node[zxvertex=\zxwhite,zxdown] at (1.75,2){};
\end{tz}\end{calign}%
\begin{definition} A \textit{morphism of dagger bimodules} $_AM_B\to {}_AN_B$ is a morphism $f:M\to N$ that commutes with the action of the Frobenius monoids:
\begin{calign}
\begin{tz}[zx]
\draw (0,0) to (0,3);
\draw (-1,0) to [out=up, in=-135] (0,2.15);
\draw (1,0) to [out=up, in=-45] (0,2.15) ;
\node[zxnode=\zxwhite] at (0,0.85) {$f$};
\node[boxvertex,zxdown] at (0,2.15){};
\end{tz}
=
\begin{tz}[zx]
\draw (0,0) to (0,3);
\draw (-1,0) to [out=up, in=-135] (0,0.85);
\draw (1,0) to [out=up, in=-45] (0,0.85) ;
\node[zxnode=\zxwhite] at (0,2.15) {$f$};
\node[boxvertex,zxdown] at (0,0.85){};
\end{tz}
\end{calign}
Two dagger bimodules are \textit{isomorphic}, here written ${}_AM_B\cong{}_AN_B$, if there is a unitary morphism of dagger bimodules ${}_AM_B\to{}_AN_B$.
\end{definition}
\noindent
In a monoidal dagger category in which dagger idempotents split, we can compose dagger bimodules ${}_AM_B$ and ${}_BN_C$ to obtain an $A{-}C$-dagger bimodule ${}_AM{\otimes_B}N_C$ as follows. First note that the following endomorphism is a dagger idempotent:
\begin{calign}\label{eq:idempotentforrelprod}
\begin{tz}[zx,every to/.style={out=up, in=down}]
\draw (0,0) to (0,3);
\draw (2,0) to (2,3);
\draw (0,2.25) to [out=-45, in=left] (1, 1.2) to[out=right, in=-135] (2,2.25);
\node[zxvertex=\zxwhite] at (1,1.2){};
\node[boxvertex,zxdown] at (0,2.25){};
\node[boxvertex,zxdown] at (2,2.25){};
\node[dimension, left] at (0,0) {$M$};
\node[dimension, right] at (2,0) {$N$};
\end{tz}
\end{calign}
The \emph{relative tensor product} ${}_AM{\otimes_B}N_C$ is defined as the image of the splitting of this idempotent. We depict the morphism $i: M\otimes_B N\to M\otimes N$ as a downwards pointing triangle:
\begin{calign}\label{eq:moritaidempotentsplit}
\begin{tz}[zx,every to/.style={out=up, in=down}]
\draw (0,0) to (0,3);
\draw (2,0) to (2,3);
\draw (0,2.25) to [out=-45, in=left] (1, 1.2) to[out=right, in=-135] (2,2.25);
\node[zxvertex=\zxwhite] at (1,1.2){};
\node[boxvertex,zxdown] at (0,2.25){};
\node[boxvertex,zxdown] at (2,2.25){};
\end{tz}
~~=~~
\begin{tz}[zx,every to/.style={out=up, in=down}]
\draw (0,0) to (0,0.5);
\draw (0,2.5) to (0,3);
\draw (2,0) to (2,0.5);
\draw (2,2.5) to (2,3);
\draw (1,0.5) to (1,2.5);
\node[dimension, right] at (1,1.5) {$M{\otimes_B}N$};
\node[triangleup=2] at (1,0.5){};
\node[triangledown=2] at (1,2.5){};
\end{tz}
&
\begin{tz}[zx]
\clip (-1.2, -0.3) rectangle (1.9,3.3);
\draw (0,0) to (0,1);
\draw (-1,1) to (-1,2);
\draw (1,1) to (1,2);
\draw (0, 2) to (0,3);
\node[triangleup=2] at (0,2){};
\node[triangledown=2] at (0,1){};
\node[dimension, right] at (0,0) {$M{\otimes_B}N$};
\end{tz}
=~~
\begin{tz}[zx]
\clip (-0.2, -0.3) rectangle (1.9,3.3);
\draw (0,0) to (0,3);
\node[dimension, right] at (0,0) {$M{\otimes_B}N$};
\end{tz}
\end{calign}
\ignore{One can convince onself that $f:M\otimes N\to M\otimes_B N$ is a coequalizer for the }
For dagger bimodules ${}_AM_B$ and ${}_BN_C$, the relative tensor product $M\otimes_B N$ is itself an $A{-}C$-dagger bimodule with the following action $A\otimes(M{\otimes_B}N) \otimes C\to M{\otimes_B}N$:
\begin{equation}
\begin{tz}[zx,every to/.style={out=up, in=down}]
\draw (0,-0.) to (0,1) ;
\draw (0,3) to (0,3.5);
\draw (-1,1) to (-1,2.5);
\draw (1,1) to (1,2.5);
\draw (-2,-0.) to [in=-135] (-1,1.75);
\draw (2,-0.) to [in=-45] (1,1.75);
\node[triangledown=2] at (0,1){};
\node[triangleup=2] at (0,2.5){};
\node[boxvertex] at (-1,1.75){};
\node[boxvertex] at (1,1.75){};
\end{tz}
\end{equation}
\noindent
\begin{definition} \label{def:daggermoritaequiv}Two special dagger Frobenius monoids $A$ and $B$ are \textit{Morita equivalent} if there are dagger bimodules $_AM_B$ and $_BN_A$ such that $_AM{\otimes_B}N_A\cong {}_AA_A$ and ${_BN{\otimes_A}M_B \cong {}_BB_B}$.
\end{definition}
\noindent
In other words, special dagger Frobenius monoids $A$ (depicted as a white node) and $B$ (depicted as a grey node) are Morita equivalent if there are dagger bimodules ${}_AM_B$ and ${}_BN_A$ and morphisms $i: A\to M\otimes N$ (depicted as a downwards-pointing white triangle) and $i':B\to N\otimes M$ (depicted as a downwards-pointing grey triangle) such that the following equations hold:
\def\xscl{1}
\def\scl{0.8}
\def\trianglescale{0.8}
\begin{calign}\label{eq:alldataMorita}\hspace{-0.2cm}
\begin{tz}[zx,every to/.style={out=up, in=down},scale=\scl,xscale=\xscl]
\draw (0,0) to (0,1.5) ;
\draw (0,2.5) to (0,4);
\draw (-1,1.5) to (-1,2.5);
\draw (1,1.5) to (1,2.5);
\draw (-2,0) to [in=-135] (-1,2);
\draw (2,0) to [in=-45] (1,2);
\node[triangledown=2] at (0,1.5){};
\node[triangleup=2] at (0,2.5){};
\node[boxvertex] at (-1,2){};
\node[boxvertex] at (1,2){};
\end{tz}
=
\begin{tz}[zx,every to/.style={out=up, in=down},scale=\scl,xscale=\xscl]
\draw (0,0 ) to (0,4);
\draw (-1, 0) to [in=-135] (0,2.);
\draw (1,0) to [in=-45] (0,2.);
\node[zxvertex=\zxwhite,zxdown] at (0,2.){};
\end{tz}
&
\begin{tz}[zx,scale=\scl,xscale=\xscl]
\draw (0,0) to (0,1);
\draw (2,0) to (2,1);
\draw (1,1) to (1,3);
\draw (0,3) to (0,4);
\draw (2,3) to (2,4);
\node[triangleup=2] at (1,1){};
\node[triangledown=2] at (1,3){};
\end{tz}
\!=
\begin{tz}[zx,every to/.style={out=up, in=down},scale=\scl,xscale=\xscl]
\draw (0,0.) to (0,4);
\draw (2,0.) to (2,4);
\draw (0,2.25) to [out=-45, in=left] (1, 1.2) to[out=right, in=-135] (2,2.25);
\node[zxvertex=\zxblack] at (1,1.2){};
\node[boxvertex,zxdown] at (0,2.25){};
\node[boxvertex,zxdown] at (2,2.25){};
\end{tz}
&
\begin{tz}[zx,every to/.style={out=up, in=down},scale=\scl,xscale=\xscl]
\draw (0,0) to (0,1.5) ;
\draw (0,2.5) to (0,4);
\draw (-1,1.5) to (-1,2.5);
\draw (1,1.5) to (1,2.5);
\draw (-2,0) to [in=-135] (-1,2);
\draw (2,0) to [in=-45] (1,2);
\node[triangledown=2,fill=\zxblack] at (0,1.5){};
\node[triangleup=2,fill=\zxblack] at (0,2.5){};
\node[boxvertex] at (-1,2){};
\node[boxvertex] at (1,2){};
\end{tz}
=
\begin{tz}[zx,every to/.style={out=up, in=down},scale=\scl,xscale=\xscl]
\draw (0,0 ) to (0,4);
\draw (-1, 0) to [in=-135] (0,2.);
\draw (1,0) to [in=-45] (0,2.);
\node[zxvertex=\zxblack,zxdown] at (0,2.){};
\end{tz}
&
\begin{tz}[zx,scale=\scl,xscale=\xscl]
\draw (0,0) to (0,1);
\draw (2,0) to (2,1);
\draw (1,1) to (1,3);
\draw (0,3) to (0,4);
\draw (2,3) to (2,4);
\node[triangleup=2,fill=\zxblack] at (1,1){};
\node[triangledown=2,fill=\zxblack] at (1,3){};
\end{tz}
\!=
\begin{tz}[zx,every to/.style={out=up, in=down},scale=\scl,xscale=\xscl]
\draw (0,0.) to (0,4);
\draw (2,0.) to (2,4);
\draw (0,2.25) to [out=-45, in=left] (1, 1.2) to[out=right, in=-135] (2,2.25);
\node[zxvertex=\zxwhite] at (1,1.2){};
\node[boxvertex,zxdown] at (0,2.25){};
\node[boxvertex,zxdown] at (2,2.25){};
\end{tz}
\end{calign}
It can easily be verified that $*$-isomorphic special dagger Frobenius monoids are Morita equivalent.

\ignore{
\begin{example} Isomorphic Frobenius monoids are Morita equivalent, as can be easily seen. (For isomorphic Frobenius monoids $A$ and $B$, $A$ can be given the structure of a right $B$-module, and $B$ the structure of a left $A$-module,  by using the isomorphism followed by the multiplication in the algebra. The other morphisms in the definition are easily identified.)
\end{example}
\begin{definition}
A Frobenius monoid is Morita trivial if it is Morita equivalent to the monoidal unit considered as a Frobenius monoid using the structural morphisms of the monoidal category.
\end{definition}
\noindent
In fact, it is well known that Morita trivial Frobenius monoids can be characterised up to $*$-isomorphism.
\begin{definition}
For any dualisable object $\tau$ in a monoidal dagger category $\mathcal{C}$, there is a \emph{pair of pants} Frobenius structure on the tensor product $\conj{\tau} \otimes \tau$ defined using the cups and caps of the duality. This structure is shown in~\eqref{eq:quantumbijectionfrobenius}.
\end{definition}
\begin{proposition} A Frobenius monoid $M$ in a monoidal dagger category $\mathcal{C}$ is \emph{Morita trivial} if there is a dualisable object $\tau$ in $\mathcal{C}$ and a $*$-isomorphism $M \cong \conj{\tau} \otimes \tau$.
\end{proposition}
}

\section{A classification of quantum isomorphic graphs}
\label{sec:classification}
In this section we present our classification of quantum graphs $\Gamma'$ quantum isomorphic to a given graph $\Gamma$ in terms of algebraic structures in the monoidal category $\QAut(\Gamma)$.

These results are based on the observation that dualisable $1$-morphisms $\Gamma' \to \Gamma$ in a dagger $2$-category give rise to dagger Frobenius monoids in the endomorphism category $\Hom(\Gamma,\Gamma)$. These Frobenius monoids can therefore be used to classify dualisable morphisms into $\Gamma$.
This is a prominent technique employed, for example, in the theory of module categories~\cite{Ostrik2003,Ostrik2003_2,Etingof2015} and the classification of subfactors~\cite{Bischoff2015}.

\ignore{
We make use of the theory of fusion categories and the classification of Frobenius algebras within them. Good references are~\cite{Etingof2015,Mueger2003}, but we provide a review of the essentials of Frobenius algebras and Morita equivalence in Appendix \ref{app:Frobenius}. For reasons of brevity, some constructions and proofs in this section are only sketched; the full proofs are  provided in Appendix \ref{app:proof0}.
}

In Section~\ref{sec:quantumisomorphicquantum}, we classify quantum graphs quantum isomorphic to a given quantum graph (Corollary~\ref{cor:bigclassification}). In Section~\ref{sec:quantumisomorphicclassical}, we restrict attention to classical graphs, and classify classical graphs quantum isomorphic to a given classical graph (Corollary~\ref{cor:superclassification}).
\ignore{
\DRcomm{Maybe some more bla: Overall, that shows that we can classify quantum graphs that are quantum isomorphic to a quantum graph $\Gamma$ in terms of structures in $\QAut(\Gamma)$ and therefore ultimately.... (or maybe put this into the appendix.\\}
\DRcomm{Theorem~\ref{} will be proven in Section~\ref{}.}
}

\subsection{Classifying quantum isomorphic quantum graphs}\label{sec:quantumisomorphicquantum}
\ignore{
In this section, we classify quantum isomorphisms out of a quantum graph $\Gamma$ in terms of Morita equivalence classes of dagger Frobenius monoids\footnote{We refer to Frobenius monoids...Recall Terminology~}\DR{footnote} in the monoidal dagger categories $\QAut(\Gamma)$. \DRcomm{Summary maybe in intro to sect 3}
}

We first establish that dualisable $1$-morphisms in $\QGraphIso$ into a quantum graph $\Gamma$ give rise to Frobenius monoids in $\QAut(\Gamma)$. 
\begin{proposition}\label{prop:quantumbijectionFrobenius} A quantum isomorphism $(H,P)$ between quantum graphs $\Gamma'$ and $\Gamma$ gives rise to a special dagger Frobenius monoid in $\QAut(\Gamma)$, whose underlying object is the composition $(H\otimes H^*, P \circ \overline{P})$, and whose underlying algebra is the endomorphism algebra~\eqref{eq:endomorphismalgebra}:
\begin{calign}\label{eq:quantumbijectionFrobenius}
\begin{tz}[zx, xscale=-1, every to/.style={out=up, in=down}]
\draw (0,0) to (2.5,3.5);
\draw[arrow data={0.15}{<}, arrow data={0.8}{<}] (1.25,0) to node[zxnode=\zxwhite, pos=0.4] {$\conj{P}$} (0,3.5);
\draw[string,arrow data={0.2}{>}, arrow data={0.9}{>}] (2.5,0) to node[zxnode=\zxwhite, pos=0.58] {$P$} (1.25,3.5);
\node[dimension,left] at (1.25,0) {$H^*$};
\node[dimension, left] at (2.5,0) {$H$};
\end{tz}
\end{calign}
\end{proposition}
\begin{proof} It is well known that the composition $P\circ \overline{P}$ of a 1-morphism with its dual in a dagger $2$-category gives rise to a Frobenius monoid. In our case, all we need to show is that the structural morphisms of the endomorphism algebra~\eqref{eq:endomorphismalgebra} are intertwiners for $P \circ \overline{P}$. This follows immediately from equations~\eqref{eq:rightdual} and~\eqref{eq:leftdual}.
\end{proof}
\noindent
The Frobenius monoids arising from dualisable $1$-morphisms in Proposition~\ref{prop:quantumbijectionFrobenius} have an underlying endomorphism algebra. We abstract this property.
\begin{definition} \label{def:simpleFrob}Let $\mathcal{C}$ be a monoidal dagger category with a faithful monoidal dagger functor $F:\mathcal{C} \to \Hilb$. A \emph{$F$-simple dagger Frobenius monoid} in $\mathcal{C}$ is a dagger Frobenius monoid $A$ in $\mathcal{C}$ such that the underlying dagger Frobenius algebra $FA$ in $\Hilb$ is $*$-isomorphic to an endomorphism algebra~\eqref{eq:endomorphismalgebra}.
\end{definition}
\noindent
Every $F$-simple dagger Frobenius monoid $A$ is special, since $FA$ is special.

In the following, we will be concerned with $F$-simple dagger Frobenius monoids in $\QAut(\Gamma)$ where $F: \QAut(\Gamma) \to \Hilb$ is the evident forgetful functor.\footnote{The forgetful functor $\QAut(\Gamma)\to \Hilb$ takes a quantum isomorphism $(H,P)$ to the Hilbert space $H$ and an intertwiner to the underlying linear map; equivalently it is the forgetful functor of the finite-dimensional representation category $\QAut(\Gamma) = \Rep_{\text{fd}}(A(\Gamma))$. See \cite[Section 3.3]{Musto2017a}} From now on, we omit the functor $F$ from the notation and refer to \emph{simple dagger Frobenius monoids} in $\QAut(\Gamma)$.

\begin{remark} Since $\QAut(\Gamma)$ is the category of finite-dimensional $*$-representations of the Hopf $C^*$-algebra $A(\Gamma)$, unpacking Definition~\ref{def:simpleFrob} gives the definition~\eqref{eq:sweedler} made in the introduction.
\end{remark}

\noindent
The main result of this section is that the converse of Proposition~\ref{prop:quantumbijectionFrobenius} is also true: simple dagger Frobenius monoids in $\QAut(\Gamma)$ give rise to quantum isomorphisms into $\Gamma$. 

\begin{theorem}[restate=splitfrobenius,name={}] \label{thm:Frobeniussplit}Let $\Gamma$ be a quantum graph and let $X$ be a simple dagger Frobenius monoid in $\QAut(\Gamma)$. Then there exist a quantum graph $\Gamma_X$ and a quantum isomorphism ${(H,P):\Gamma_X\to\Gamma}$ such that $X$ is $*$-isomorphic to $(H\otimes H^*,P \circ \conj{P})$.
\end{theorem}
\begin{proof} We will prove this in Section~\ref{sec:proof}.
\end{proof}
\begin{remark}From the perspective of category theory, the quantum graph $\Gamma_X$ is both an Eilenberg-Moore and a Kleisli object~\cite{Lack:2002} for the Frobenius monad $X$ in the $2$-category $\QGraphIso$.
\end{remark}
\noindent
Although we postpone the details of the proof, we quickly sketch the reconstruction of the quantum graph $\Gamma_X$ and the quantum isomorphism $\Gamma_X \to \Gamma$ from a simple dagger Frobenius monoid $X$ in $\QAut(\Gamma)$. Note firstly that the Frobenius monoid $X$ is a quantum isomorphism $(H\otimes H^*, X): \Gamma \to \Gamma$ for which the endomorphism algebra~\eqref{eq:endomorphismalgebra} is an intertwiner. It is then easy to check that the following linear map  $x\in\End(H^*\otimes V_{\Gamma} \otimes H)$ is a dagger idempotent, i.e. is self-adjoint and fulfils $x^2=x$:
\begin{equation}
\frac{1}{n}~~~
\begin{tz}[zx,xscale=-1,xscale=0.6,yscale=1]
\draw(0.25,-0.5) to (0.25,0) to [out=up, in=-135] (1,1);
\draw [arrow data={0.23}{<},arrow data={0.88}{<}]  (1,1) to [out=135, in=right]  (-0.3, 1.7) to [out=left, in=up] (-1.25,1) to (-1.25,-0.5);
\draw(1.75,2.5) to (1.75,2) to [out=down, in=45] (1,1);
\draw[arrow data={0.23}{<},arrow data={0.9}{<}] (1,1) to [out= -45, in= left]   (2.3,0.3) to [out=right, in=down] (3.25,1) to (3.25,2.5);
\draw[arrow data={0.83}{>}] (1,1) to [out=up, in=down] (0.25,2.5);
\draw[arrow data={0.83}{>}] (1,1) to [out=down, in=up] (1.75,-0.5);
\node [zxnode=\zxwhite] at (1,1) {$X$};
\end{tz} 
 \end{equation}
Splitting this idempotent results in a new Hilbert space $A$ and an isometry $i:A\to H^*\otimes V_{\Gamma} \otimes H$ which gives rise to (by bending wires) a linear map $P: H \otimes A \to V_{\Gamma}\otimes H $, so that $X$ is of the form~\eqref{eq:quantumbijectionFrobenius}. 

We now define the structure of a quantum graph on $A$. For this, we use the following shorthand notation:
\begin{calign}\label{eq:shorthand}
\hspace{-0.75cm}
\begin{tz}[zx,xscale=-1,every to/.style={out=up, in=down},xscale=-0.8]
\draw [arrow data={0.2}{>},arrow data={0.8}{>}]  (0,0) to (2.25,3);
\draw (2.25,0) to node[zxvertex=\zxwhite, pos=0.5] {} (0,3);
\end{tz}
~=~
\begin{tz}[zx,xscale=-1,every to/.style={out=up, in=down},xscale=-0.8]
\draw [arrow data={0.2}{>},arrow data={0.8}{>}]  (0,0) to (2.25,3);
\draw (2.25,0) to node[zxnode=\zxwhite, pos=0.5] {$P$} (0,3);
\end{tz}
&
\begin{tz}[zx,every to/.style={out=up, in=down},xscale=-0.8]
\draw [arrow data={0.2}{>},arrow data={0.8}{>}]  (0,0) to (2.25,3);
\draw (2.25,0) to node[zxvertex=\zxwhite, pos=0.5] {} (0,3);
\end{tz}
~=~
\begin{tz}[zx,every to/.style={out=up, in=down},xscale=-0.8]
\draw [arrow data={0.2}{>},arrow data={0.8}{>}]  (0,0) to (2.25,3);
\draw (2.25,0) to node[zxnode=\zxwhite, pos=0.5] {$P^\dagger$} (0,3);
\end{tz}
&
\begin{tz}[zx,xscale=-1,every to/.style={out=up, in=down},xscale=-0.8,yscale=-1]
\draw [arrow data={0.2}{>},arrow data={0.8}{>}]  (0,0) to (2.25,3);
\draw (2.25,0) to node[zxvertex=\zxwhite, pos=0.5] {} (0,3);
\end{tz}~=~
\begin{tz}[zx,xscale=-1,xscale=0.6,yscale=1,scale=1]
\draw(0.25,-0.5) to (0.25,0) to [out=up, in=-135] (1,1);
\draw [arrow data={0.34}{>}]  (1,1) to [out=135, in=right]  (-0.3, 1.7) to [out=left, in=up] (-1.25,1) to (-1.25,-0.5);
\draw(1.75,2.5) to (1.75,2) to [out=down, in=45] (1,1);
\draw[arrow data={0.32}{<}] (1,1) to [out= -45, in= left]   (2.3,0.3) to [out=right, in=down] (3.25,1) to (3.25,2.5);
\node [zxnode=\zxwhite] at (1,1) {$P$};
\end{tz}
&
\begin{tz}[zx,xscale=-1,every to/.style={out=up, in=down},xscale=0.8,yscale=-1]
\draw [arrow data={0.2}{>},arrow data={0.8}{>}]  (0,0) to (2.25,3);
\draw (2.25,0) to node[zxvertex=\zxwhite, pos=0.5] {} (0,3);
\end{tz}~=~
\begin{tz}[zx,xscale=-1,xscale=-0.6,yscale=1,scale=1]
\draw(0.25,-0.5) to (0.25,0) to [out=up, in=-135] (1,1);
\draw [arrow data={0.34}{>}]  (1,1) to [out=135, in=right]  (-0.3, 1.7) to [out=left, in=up] (-1.25,1) to (-1.25,-0.5);
\draw(1.75,2.5) to (1.75,2) to [out=down, in=45] (1,1);
\draw[arrow data={0.32}{<}] (1,1) to [out= -45, in= left]   (2.3,0.3) to [out=right, in=down] (3.25,1) to (3.25,2.5);
\node [zxnode=\zxwhite] at (1,1) {$P^\dagger$};
\end{tz}
\end{calign}
Using the algebra structure on $V_{\Gamma}$ (depicted as a white node), we define an algebra structure on $A$ (depicted as a grey node) as follows:
\def\yoff{-0.15}
\def\loff{-0.15}
\def\scl{0.7}
\begin{calign}
\label{eq:secondalgebra}
\hspace{-0.8cm}
\begin{tz}[zx,scale=\scl]
\draw (0,0) to (0,1) to [out=up, in=-135] (0.75,2) to (0.75,3.75);
\draw (1.5,0) to (1.5,1) to  [out=up, in=-45] (0.75,2);
\node[zxvertex=\zxblack, zxdown] at  (0.75,2){};
\node[dimension, right] at (1.5,0) {$A$};
\node[dimension, right] at (0,0) {$A$};
\node[dimension, right] at (0.75,3.75) {$A$};
\end{tz}
\!\!\!:=~\frac{1}{n}~
\begin{tz}[zx,scale=\scl]
\draw[string] (0,0) to node[front,zxvertex=\zxwhite, pos=0.69]{}(0,1) to [out=up, in=-135] (0.75,2) to node[front, zxvertex=\zxwhite, pos=0.5]{} (0.75,3.75);
\draw[string] (1.5,0) to (1.5,1) to [out=up, in=-45] node[front,zxvertex=\zxwhite, pos=0.0]{} (0.75, 2);
\draw[string,arrow data={0.01}{>},arrow data = {0.33}{>},arrow data={0.665}{>}] (-0.25+\loff, 0.75+\yoff) to [out=20, in=down] (2.25,2+\yoff) to [out=up, in=-20] (-0.25+\loff, 3.25+\yoff)  to [out=160 , in=200] (-0.25+\loff, 0.75+\yoff) ;
\node[zxvertex=\zxwhite,zxdown] at (0.75,2){};
\end{tz}
&
\begin{tz}[zx,scale=\scl]
\clip (0.45, -0.25) rectangle (1.55, 4);
\draw  (0.75,2) to (0.75,3.75);
\node[zxvertex=\zxblack] at  (0.75,2){};
\node[dimension, right] at (0.75,3.75) {$A$};
\end{tz}
\!:=~\frac{1}{n}~
\begin{tz}[zx,scale=\scl]
\clip (-1+\loff, 0) rectangle (2.45,3.75); 
\draw[string] (0.75,2) to node[front, zxvertex=\zxwhite, pos=0.5]{} (0.75,3.75);
\draw[string,arrow data={0.01}{>},arrow data = {0.33}{>},arrow data={0.665}{>}] (-0.25+\loff, 0.75+\yoff) to [out=20, in=down] (2.25,2+\yoff) to [out=up, in=-20] (-0.25+\loff, 3.25+\yoff)  to [out=160 , in=200] (-0.25+\loff, 0.75+\yoff) ;
\node[zxvertex=\zxwhite] at (0.75,2){};
\end{tz}
&
\begin{tz}[zx,scale=\scl, yscale=-1]
\draw (0,0) to (0,1) to [out=up, in=-135] (0.75,2) to (0.75,3.75);
\draw (1.5,0) to (1.5,1) to  [out=up, in=-45] (0.75,2);
\node[zxvertex=\zxblack, zxup] at  (0.75,2){};
\node[dimension, right] at (1.5,0) {$A$};
\node[dimension, right] at (0,0) {$A$};
\node[dimension, right] at (0.75,3.75) {$A$};
\end{tz}
\!\!\!:=~\frac{1}{n}~
\begin{tz}[zx,scale=\scl,yscale=-1]
\draw[string] (0,0) to node[front,zxvertex=\zxwhite, pos=0.69]{}(0,1) to [out=up, in=-135] (0.75,2) to node[front, zxvertex=\zxwhite, pos=0.5]{} (0.75,3.75);
\draw[string] (1.5,0) to (1.5,1) to [out=up, in=-45] node[front,zxvertex=\zxwhite, pos=0.0]{} (0.75, 2);
\draw[string,arrow data={0.0}{<},arrow data = {0.33}{<},arrow data={0.67}{<}] (-0.25+\loff, 0.75+\yoff) to [out=20, in=down] (2.25,2+\yoff) to [out=up, in=-20] (-0.25+\loff, 3.25+\yoff)  to [out=160 , in=200] (-0.25+\loff, 0.75+\yoff) ;
\node[zxvertex=\zxwhite,zxup] at (0.75,2){};
\end{tz}
&
\begin{tz}[zx,scale=\scl, yscale=-1]
\clip (0.45, -0.25) rectangle (1.55, 4);
\draw  (0.75,2) to (0.75,3.75);
\node[zxvertex=\zxblack] at  (0.75,2){};
\node[dimension, right] at (0.75,3.75) {$A$};
\end{tz}
\!:=~\frac{1}{n}~
\begin{tz}[zx,scale=\scl, yscale=-1]
\clip (-1+\loff, 0) rectangle (2.45,3.75); 
\draw[string] (0.75,2) to node[front, zxvertex=\zxwhite, pos=0.5]{} (0.75,3.75);
\draw[string,arrow data={0.01}{<},arrow data = {0.33}{<},arrow data={0.665}{<}] (-0.25+\loff, 0.75+\yoff) to [out=20, in=down] (2.25,2+\yoff) to [out=up, in=-20] (-0.25+\loff, 3.25+\yoff)  to [out=160 , in=200] (-0.25+\loff, 0.75+\yoff) ;
\node[zxvertex=\zxwhite] at (0.75,2){};
\end{tz}
\end{calign}
We will show in Section~\ref{sec:proof} that this makes $A$ into a special symmetric dagger Frobenius algebra.
The quantum graph $\Gamma_X$ has vertex quantum set $V_{\Gamma_X}:=A$ and quantum adjacency matrix $\Gamma_X:A \to A$, defined as follows:
\begin{equation}\begin{tz}[zx,scale=\scl, yscale=1]
\draw[string] (0.75,0) to (0.75,3.75);
\node[zxnode=\zxwhite] at (0.75, 1.875) {$\Gamma_X$};
\node[dimension, right] at (0.75,0) {$A$};
\node[dimension, right] at (0.75, 3.75) {$A$};
\end{tz}
~=~\frac{1}{n}~
\begin{tz}[zx,scale=\scl, yscale=1]
\draw[string] (0.75,0) to node[front, zxvertex=\zxwhite, pos=0.225]{}node[front,zxvertex=\zxwhite, pos=0.765]{} (0.75,3.75);
\draw[string,arrow data={0.01}{>},arrow data = {0.33}{>},arrow data={0.665}{>}] (-0.25+\loff, 0.75+\yoff) to [out=20, in=down] (2.25,2+\yoff) to [out=up, in=-20] (-0.25+\loff, 3.25+\yoff)  to [out=160 , in=200] (-0.25+\loff, 0.75+\yoff) ;
\node[zxnode=\zxwhite] at (0.75, 1.875) {$\Gamma$};
\end{tz}
\end{equation}
We will prove in Section~\ref{sec:proof} that $\Gamma_X$ is a quantum graph, and that $P$ is a quantum isomorphism from $\Gamma_X$ to $\Gamma$.

\begin{remark}\label{rem:quantumgraphsemerge} The algebra $V_{\Gamma_X}$ is in general noncommutative, even if $V_{\Gamma}$ is commutative. Quantum graphs therefore naturally emerge in Theorem~\ref{thm:Frobeniussplit}, even if we restrict our attention to classical graphs $\Gamma$. For pseudo-telepathy, we are  interested in classical graphs $\Gamma_X$, and therefore want $V_{\Gamma_X}$ to be commutative; in Section~\ref{sec:quantumisomorphicclassical}, we give a necessary and sufficient condition on the Frobenius monoid $X$ for this to be the case.
\end{remark}
\noindent
In summary, for every quantum isomorphism $\Gamma' \to \Gamma$ we get a simple dagger Frobenius monoid in $\QAut(\Gamma)$ (Proposition~\ref{prop:quantumbijectionFrobenius}), and for every simple dagger Frobenius monoid in $\QAut(\Gamma)$ we get a quantum isomorphism $\Gamma' \to \Gamma$ (Theorem~\ref{thm:Frobeniussplit}). With the right notion of equivalence (Definition~\ref{def:daggermoritaequiv}) of simple dagger Frobenius monoids, this in fact gives us a classification of  quantum graphs quantum isomorphic to $\Gamma$. 
\begin{corollary}[restate=onetoonecorrespondence, name={}]\label{cor:bigclassification} Let $\Gamma$ be a quantum graph. The constructions of Proposition~\ref{prop:quantumbijectionFrobenius} and Theorem~\ref{thm:Frobeniussplit}  induce a bijective correspondence between:
\vspace{-0.1cm}
\begin{itemize} \item Isomorphism classes of quantum graphs $\Gamma'$ such that there exists a quantum isomorphism $\Gamma'\to \Gamma$.
\item Morita equivalence classes of simple dagger Frobenius monoids in $\QAut(\Gamma)$.
\end{itemize}
\end{corollary}
\begin{proof}
This is a straightforward application of a general theorem (Theorem~\ref{thm:maintechnical}), which holds in any dagger 2-category in which dagger idempotents split, and which is proved in Appendix~\ref{app:2categorypictures}. 

To apply this theorem, we note that dagger idempotents split in $\QGraphIso$, as shown in~\cite[Proof of Theorem 6.4]{Musto2017a}. The conditions of the theorem are therefore satisfied. The result follows immediately, since every 1-morphism in $\QGraphIso$ can be normalized to a special $1$-morphism (see Appendix~\ref{app:2categorypictures}) by multiplication with a scalar factor, dagger equivalences in $\QGraphIso$ are precisely ordinary isomorphisms (Proposition~\ref{prop:equivalenceQGraph}), and dagger Frobenius monoids in $\QAut(\Gamma)$ are split if and only if they are simple (Proposition~\ref{prop:quantumbijectionFrobenius} and Theorem~\ref{thm:Frobeniussplit}).
\end{proof}
\ignore{By Corollary~\ref{cor:bigclassification}, we can understand pseudo-telepathic graphs in terms of Morita equivalence classes of Frobenius monoids of central type in $\QAut(G)$.}%

\ignore{
\begin{remark} Since every equivalence in $\QGraph$ is a dagger equivalence, we could have equivalently classified quantum isomorphic quantum graphs in terms of Frobenius monoids of central type up to \emph{dagger Morita equivalence} as defined in Section~\ref{app:Frobenius}. \DR{check ref} \DR{also we might want to move this Remark below the next remark. The next remark is more important}
\end{remark}
}

\begin{remark} \label{rem:onlyquantumautos}The classification in Corollary~\ref{cor:bigclassification} only depends on the monoidal category $\QAut(\Gamma)$ and its fibre functor $F:\QAut(\Gamma) \to \Hilb$. In the language of compact quantum groups, the classification of quantum graphs quantum isomorphic to a classical graph $\Gamma$ depends only on the quantum automorphism group of $\Gamma$, and not on its action on the set of vertices $V_{\Gamma}$. 
\end{remark}

\begin{remark}
Corollary~\ref{cor:bigclassification} provides a classification of all quantum graphs $\Gamma'$ which are quantum isomorphic to a quantum graph $\Gamma$, but does not classify the explicit quantum isomorphisms between $\Gamma'$ and $\Gamma$. Such a classification can in fact be achieved as follows. We take two quantum isomorphisms $(H,P): \Gamma' \to \Gamma$ and $(H',P'): \Gamma''\to \Gamma$ into $\Gamma$ to be \emph{equivalent} if there is an isomorphism of quantum graphs (Definition~\ref{def:quantumgraphiso}) $\epsilon:\Gamma'\to \Gamma''$ and a unitary map $U:H\to H'$ such that the following holds\footnote{For classical graphs $\Gamma, \Gamma'$ and $\Gamma''$, this translates into the following condition on projective permutation matrices. Two projective permutation matrices $\{P_{v',v}\}_{v'\in V_{\Gamma'}, v\in V_{\Gamma}}$ and $\{P'_{v'',v}\}_{v''\in V_{\Gamma''}, v\in V_{\Gamma}}$ on Hilbert spaces $H$ and $H'$ are equivalent if there is a graph isomorphism $\epsilon: \Gamma' \to \Gamma''$ and a unitary $U:H\to H'$ such that for all $v\in V_{\Gamma}$ and $v' \in V_{\Gamma'}$ it holds that $P_{v',v} = U^\dagger P'_{\epsilon(v'),v} U $ }:
\begin{equation}\label{eq:equivalencequantumiso}\begin{tz}[zx,yscale=-1,every to/.style={out=up, in=down}]
\draw (0,0) to (0,1) to (2,3) to node[zxnode=\zxwhite, pos=0.1] {$\epsilon$} (2,4);
\draw[string,arrow data={0.35}{<}] (2,0) to node[zxnode=\zxwhite, pos=0.9] {$U^\dagger$} (2,1);
\draw[string,arrow data={0.8}{<},arrow data={0.26}{<}](2,1) to node[zxnode=\zxwhite, pos=0.5]{$P'$}  (0,3);
\draw[string,arrow data={0.85}{<}] (0,3) to node[zxnode=\zxwhite, pos=0.1]{$U$} (0,4);
\node[dimension, right] at (2,0) {$H$};
\node[dimension, right] at (0,0) {$V_{\Gamma}$};
\node[dimension, right] at (2,4) {$V_{\Gamma'}$};
\node[dimension, right] at (0,4) {$H$};
\end{tz}
\quad = \quad~
\begin{tz}[zx,xscale=-1,every to/.style={out=up, in=down}]
\draw (0,0) to (0,1) to (2,3) to  (2,4);
\draw (2,0) to  (2,1);
\draw[string,arrow data={0.99}{>},arrow data={0.01}{>}](2,1) to node[zxnode=\zxwhite, pos=0.5]{$P$}  (0,3);
\draw[string](0,3) to (0,4);
\node[dimension, right] at (2,0) {$H$};
\node[dimension, right] at (0,0) {$V_{\Gamma'}$};
\node[dimension, right] at (2,4) {$V_{\Gamma}$};
\node[dimension, right] at (0,4) {$H$};
\end{tz}
\end{equation}
\noindent
It then follows from Remark~\ref{rem:classify*iso} that the constructions of Proposition~\ref{prop:quantumbijectionFrobenius} and Theorem~\ref{thm:Frobeniussplit} induce a bijection between the following sets:
\begin{itemize}
\item Quantum isomorphisms into $\Gamma$ up to the equivalence relation~\eqref{eq:equivalencequantumiso}.
\item $*$-isomorphism classes of simple dagger Frobenius monoids in $\QAut(\Gamma)$.
\end{itemize}
\noindent
In other words, $*$-isomorphism classes of simple dagger Frobenius monoids classify quantum isomorphisms into $\Gamma$ up to dagger equivalence, while the coarser Morita equivalence classes only classify quantum graphs which are quantum isomorphic to $\Gamma$, without keeping track of the quantum isomorphism itself. For applications to pseudo-telepathy, we are mainly interested in this latter, coarser classification.
\end{remark}

\noindent
The following is a first, easy application of Corollary~\ref{cor:bigclassification}.

\begin{corollary}\label{cor:easyapplication} Let $\Gamma$ be a quantum graph with trivial quantum automorphism group, that is, $\QAut(\Gamma) \cong \Hilb$. Then, every quantum graph that is quantum isomorphic to $\Gamma$ is also isomorphic to $\Gamma$.
\end{corollary} 
\begin{proof} The category $\QAut(\Gamma) \cong \Hilb$ has only one Morita equivalence class of simple dagger Frobenius monoids, corresponding to the graph $\Gamma$ itself.
\end{proof}
\subsection{Classifying quantum isomorphic classical graphs}\label{sec:quantumisomorphicclassical}

In Corollary~\ref{cor:bigclassification}, we classified quantum graphs $\Gamma'$ quantum isomorphic to a quantum graph $\Gamma$ in terms of Morita equivalence classes of simple dagger Frobenius monoids in $\QAut(\Gamma)$.
However, as discussed in Remark~\ref{rem:quantumgraphsemerge}, if $\Gamma$ is a classical graph, then the quantum graph $\Gamma_X$ corresponding to a simple dagger Frobenius monoid $X$ in $\QAut(\Gamma)$ will in general not be classical.

In this section, we prove a necessary and sufficient condition for commutativity of the algebra $V_{\Gamma_X}$, and therefore classicality of the graph $\Gamma_X$. \ignore{This is a certain commutativity condition on the simple dagger Frobenius monoid $X$. We show that} This results in a classification of classical graphs quantum isomorphic to a given classical graph $\Gamma$.  

For a quantum isomorphism $(H,P): \Gamma' \to \Gamma$, equations~\eqref{eq:rightdual} and~\eqref{eq:leftdual} are expressed in the shorthand notation~\eqref{eq:shorthand} as follows:
\def\scl{0.55}%
\begin{calign}\label{eq:biunitarybraiding}
\begin{tz}[zx,xscale=-1,scale=\scl]
\clip (-2.05,-0.05) rectangle (4.05,4.05);
\draw[arrow data={0.1}{>},arrow data={0.499}{>}, arrow data={0.9}{>}] (0,0) to [out=up, in=up, looseness=6.] (2,0);
\draw (-2,0) to [out=up, in=down]  node[zxvertex=\zxwhite, pos=0.4] {} node[zxvertex=\zxwhite, pos=0.6]{} (4,4);
\end{tz}
=
\begin{tz}[zx,xscale=-1,scale=\scl]
\clip (-2.05,-0.05) rectangle (4.05,4.05);
\draw[arrow data={0.2}{>}, arrow data={0.8}{>}] (0,0) to [out=up, in=up, looseness=2.] (2,0);
\draw (-2,0) to [out=up, in=down]  (4,4);
\end{tz}
&
\begin{tz}[zx,xscale=-1,scale=-1,scale=\scl]
\clip (-2.05,-0.05) rectangle (4.05,4.05);
\draw[arrow data={0.1}{<},arrow data={0.499}{<}, arrow data={0.9}{<}] (0,0) to [out=up, in=up, looseness=6.] (2,0);
\draw (-2,0) to [out=up, in=down]  node[zxvertex=\zxwhite, pos=0.4] {} node[zxvertex=\zxwhite, pos=0.6] {}(4,4);
\end{tz}
=
\begin{tz}[zx,xscale=-1,scale=-1,scale=\scl]
\clip (-2.05,-0.05) rectangle (4.05,4.05);
\draw[arrow data={0.2}{<}, arrow data={0.8}{<}] (0,0) to [out=up, in=up, looseness=2.] (2,0);
\draw (-2,0) to [out=up, in=down]  (4,4);
\end{tz}
\\[2pt]\label{eq:biunitarybraiding2}
\begin{tz}[zx,xscale=-1,master,scale=-1,scale=\scl]
\clip (-2.05,-0.05) rectangle (4.05,4.05);
\draw[arrow data={0.1}{>},arrow data={0.499}{>}, arrow data={0.9}{>}] (0,0) to [out=up, in=up, looseness=6.] (2,0);
\draw (-2,0) to [out=up, in=down]  node[zxvertex=\zxwhite, pos=0.4] {} node[zxvertex=\zxwhite, pos=0.6] {}(4,4);
\end{tz}
=
\begin{tz}[zx,xscale=-1,scale=-1,scale=\scl]
\clip (-2.05,-0.05) rectangle (4.05,4.05);
\draw[arrow data={0.2}{>}, arrow data={0.8}{>}] (0,0) to [out=up, in=up, looseness=2.] (2,0);
\draw (-2,0) to [out=up, in=down]  (4,4);
\end{tz}
&
\begin{tz}[zx,xscale=-1,scale=1,scale=\scl]
\clip (-2.05,-0.05) rectangle (4.05,4.05);
\draw[arrow data={0.1}{<},arrow data={0.499}{<}, arrow data={0.9}{<}] (0,0) to [out=up, in=up, looseness=6.] (2,0);
\draw (-2,0) to [out=up, in=down]  node[zxvertex=\zxwhite, pos=0.4] {} node[zxvertex=\zxwhite, pos=0.6] {}(4,4);
\end{tz}
=
\begin{tz}[zx,xscale=-1,scale=1,scale=\scl]
\clip (-2.05,-0.05) rectangle (4.05,4.05);
\draw[arrow data={0.2}{<}, arrow data={0.8}{<}] (0,0) to [out=up, in=up, looseness=2.] (2,0);
\draw (-2,0) to [out=up, in=down]  (4,4);
\end{tz}
\end{calign}
These equations look exactly like the second Reidemeister move from knot theory. Together with equations~\eqref{eq:quantumfunction} and~\eqref{eq:quantumfunction2}, this leads to a very flexible topological calculus, allowing us to move oriented Hilbert space wires almost freely through our diagrams, interconverting the algebra $V_\Gamma$ (in the following depicted by white nodes) and the algebra $V_{\Gamma'}$ (depicted by grey nodes) when passing through the corresponding nodes.
\ignore{
\begin{proof}
Putting the shorthand of~\eqref{eq:shorthand} together with the equations~\eqref{eq:daggerdual}, \eqref{eq:rightdual} and \eqref{eq:leftdual}, we immediately obtain equations~\eqref{eq:biunitarybraiding} and~\eqref{eq:biunitarybraiding2}. The equations~\eqref{eq:pullthroughdoubledwire} are the defining equations~\eqref{eq:quantumfunction} and \eqref{eq:quantumfunction2} of a quantum isomorphism for $X= \overline{P} \circ P: \Gamma \to \Gamma$. Equation~\eqref{eq:removebubble} follows immediately from the definition of the cups and caps in $\Hilb$~\eqref{eq:cupscapsHilb}.
\end{proof}
\noindent
Equations~\eqref{eq:biunitarybraiding} and~\eqref{eq:biunitarybraiding2} look exactly like the second Reidemeister move from knot theory, and allow us to pull Hilbert space wires through quantum set wires. The equations~\eqref{eq:pullthroughdoubledwire} allow us to pull the multiplication, comultiplication, unit, counit and adjacency matrix of $\Gamma$ through a doubled wire. Finally, the equation~\eqref{eq:removebubble} allows us to remove or add a bubble in exchange for a scalar factor.
}

We also recall the following piece of folklore~\cite{Baez2004}.
\begin{proposition}\label{prop:projectorcentre} Let $A$ be a\ignore{ finite-dimensional $C^*$-algebra, or equivalently a } special symmetric dagger Frobenius algebra\ignore{ in $\Hilb$} (depicted as a grey node). Then, the following endomorphism $P_{Z(A)}:A\to A$ is a projector onto the centre of $A$:
\begin{equation}\label{eq:projectorcenter}\begin{tz}[zx]
\draw (0,0) to (0,1) to [out=135, in=down] (-0.5,1.5) to [out=up, in=down] (0.5, 3) to [out=up, in=-45] (0,3.5) to (0,4.5);
\draw (0,1) to [out=45, in=down] (0.5,1.5) to [out=up, in=down] (-0.5, 3) to [out=up, in=-135] (0,3.5);
\node[zxvertex=\zxblack,zxup] at (0,1){};
\node[zxvertex=\zxblack,zxdown] at (0,3.5){};
\end{tz}
\end{equation}
\end{proposition}
\begin{proof} For an appropriately normalised matrix algebra (see e.g.~\eqref{eq:endomorphismalgebra}), Proposition~\ref{prop:projectorcentre} can easily be verified. Semisimplicity then extends this formula to general finite-dimensional $C^*$-algebras. 
\end{proof}
\noindent 
In particular, $\dim(Z(A)) = \Tr(P_{Z(A)})$, and $A$ is commutative if and only if $\Tr(P_{Z(A)}) = \dim(A)$. We use this fact to derive our commutativity condition.
\begin{theorem} \label{thm:commutativitycondition} Let $\Gamma$ be a classical graph, let $(H\otimes H^*,X)$ be a simple dagger Frobenius monoid in $\QAut(\Gamma)$ and let $\Gamma_X$ be the associated quantum graph. Then, the dimension of the centre of $V_{\Gamma_X}$ can be expressed as follows, where $X_{v,v}$ are the diagonal components of the projective permutation matrix underlying $X$ (see~\eqref{eq:componentPPM}):
\begin{equation}\label{eq:centerquantumgraph}\dim(Z(V_{\Gamma_X}))~~=~\frac{1}{\dim(H)}~
\sum_{v\in V_{\Gamma}}~~~\begin{tz}[zx]
\draw[arrow data={0.25}{>}, arrow data={0.8}{>}] (-0.3,-2) to (-0.3,2) ; 
\draw[arrow data={0.13}{<}, arrow data={0.9}{<}] (0.3,-2) to [out=up, in=down] (-1.5, 2) to [out=up, in=up, looseness=3.5] (-0.3,2);
\draw[arrow data={0.14}{>}, arrow data={0.9}{>}]  (0.3, 2) to [out=down, in=up] (-1.5,-2) to [out=down, in=down, looseness=3.5] (-0.3,-2);
\draw[arrow data={0.08}{<}, arrow data={0.92}{<}]  (0.3,2) to [out=up, in=up, looseness=3.5] (1.5,2) to (1.5,-2) to [out=down, in=down, looseness=3.5] (0.3,-2);
\node[zxnode=\zxwhite] at (0,2) {$X_{v,v}$};
\node[zxnode=\zxwhite] at (0,-2) {$X_{v,v}$};
\end{tz}
\end{equation}
In particular, $\Gamma_X$ is classical if and only if $\dim(Z(V_{\Gamma_X})) = |V_{\Gamma}|$.
\end{theorem}
\begin{proof} Note that for a special symmetric dagger Frobenius algebra $A$ (depicted as a grey node) and a linear map $f:A\to A$, the trace $\Tr(f)$ can be computed as follows:
\[ \Tr(f) ~~= \begin{tz}[zx]
\draw (0,0) to (0,0.5) to  [out=up, in=up, looseness=2] (0.8, 0.5) to (0.8,-0.5) to [out=down, in=down, looseness=2] (0,-0.5) to (0,0);
\node[zxnode] at (0,0) {$f$};
\node[zxvertex=\zxblack] at (0.4, 0.92) {};
\node[zxvertex=\zxblack] at (0.4, -0.92) {};
\end{tz}
\]
Let $(H,P):\Gamma_X \to \Gamma$ be a quantum isomorphism such that $P\circ \overline{P}= X$ (see Theorem~\ref{thm:Frobeniussplit}). Using the shorthand notation~\eqref{eq:shorthand} for $P$, the trace of the projector~\eqref{eq:projectorcenter} for the algebra $V_{\Gamma_X}$ (depcited as a grey node) can be expressed as follows, where $n=\dim(H)$:
\def\scl{0.9}
\[
\begin{tz}[zx,scale=\scl]
\draw (0,0.5) to (0,1) to [out=135, in=down] (-0.75,2) to [out=up, in=down] (0.75, 3.5) to [out=up, in=-45] (0,4.5) to (0,5)  to [out=up, in=up,looseness=2] node[zxvertex=\zxblack, pos=0.5] {} (2.5, 5) to (2.5,0.5) to [out=down, in=down,looseness=2] node[zxvertex=\zxblack, pos=0.5]{} (0,0.5);
\draw (0,1) to [out=45, in=down] (0.75,2) to [out=up, in=down] (-0.75, 3.5) to [out=up, in=-135] (0,4.5);
\node[zxvertex=\zxblack,zxup] at (0,1){};
\node[zxvertex=\zxblack,zxdown] at (0,4.5){};
\end{tz}\ignore{
~\superequalseq{eq:secondalgebra}~\frac{1}{n^4}~
\begin{tz}[zx,scale=\scl]
\draw[arrow data={0.26}{>}, arrow data={0.5}{>}, arrow data ={1}{>}] (0,1.1) ellipse (1 and 0.7);
\draw[arrow data={0.5}{>},arrow data={0.76}{>},  arrow data ={1}{>}] (0,4.4) ellipse (1 and 0.7);
\draw[arrow data={0.26}{>}, arrow data={0.76}{>}] (1.25, 6.4) ellipse (1.75 and 0.7);
\draw[arrow data={0.26}{>}, arrow data={0.76}{>}] (1.25, -0.9) ellipse (1.75 and 0.7);
\draw (0,0.5) to  (0,1) to [out=135, in=down] node[zxvertex=\zxwhite, pos=0.7]{}(-0.75,2) to [out=up, in=down] (0.75, 3.5) to [out=up, in=-45] node[zxvertex=\zxwhite, pos=0.3]{} (0,4.5) to  (0,5)  to [out=up, in=up,looseness=2]  node[zxvertex=\zxwhite, pos=0.01]{}node[zxvertex=\zxwhite, pos=0.175] {}node[zxvertex=\zxwhite, pos=0.5] {}node[zxvertex=\zxwhite, pos=0.825] {} (2.5, 5) to (2.5,0.5) to [out=down, in=down,looseness=2] node[zxvertex=\zxwhite, pos=0.175]{}node[zxvertex=\zxwhite, pos=0.5]{} node[zxvertex=\zxwhite, pos=0.825]{} node[zxvertex=\zxwhite, pos=0.99]{} (0,0.5);
\draw (0,1) to [out=45, in=down] node[zxvertex=\zxwhite, pos=0.7]{} (0.75,2) to [out=up, in=down] (-0.75, 3.5) to[out=up, in=-135]  node[zxvertex=\zxwhite, pos=0.3]{} (0,4.5);
\node[zxvertex=\zxwhite,zxup] at (0,1){};
\node[zxvertex=\zxwhite,zxdown] at (0,4.5){};
\end{tz}
}
~~=~\frac{1}{n}~
\begin{tz}[zx,scale=\scl]
\draw[arrow data={0}{<}, arrow data ={0.06}{<},arrow data={0.13}{<},arrow data ={0.21}{<},arrow data ={0.27}{<},arrow data={0.635}{<}] (-1, 1) to [out=up, in=down] (1.5, 2.75) to [out=up, in=down] (-1, 4.5) to [out=up, in=up, looseness=2] (3.5,4.5) to (3.5, 1) to [out=down, in=down, looseness=2] (-1,1);
\draw (0,0.5) to (0,1) to [out=135, in=down] node[zxvertex=\zxwhite, pos=0.65]{} (-0.75,2) to [out=up, in=down] (0.75, 3.5) to [out=up, in=-45]node[zxvertex=\zxwhite, pos=0]{} (0,4.5) to (0,5)  to [out=up, in=up,looseness=2] node[zxvertex=\zxwhite, pos=0.5] {} (2.5, 5) to (2.5,0.5) to [out=down, in=down,looseness=2] node[zxvertex=\zxwhite, pos=0.5]{} (0,0.5);
\draw (0,1) to [out=45, in=down] node[zxvertex=\zxwhite, pos=1]{} (0.75,2) to [out=up, in=down]  (-0.75, 3.5) to [out=up, in=-135]node[zxvertex=\zxwhite, pos=0.35]{} (0,4.5);
\node[zxvertex=\zxwhite,zxup] at (0,1){};
\node[zxvertex=\zxwhite,zxdown] at (0,4.5){};
\end{tz}
~~\overset{\eqref{eq:untangleunwindtwists}}{=}\frac{1}{n}~
\begin{tz}[zx,scale=\scl]
\draw[arrow data ={0.25}{<}, arrow data ={0.5}{<},arrow data={0.75}{<}, arrow data={1}{<}] (0, 2.75) ellipse (1.5 and 0.8);
\draw (0,0.5) to (0,1) to [out=135, in=down] (-0.75,2) to [out=up, in=down] node[zxvertex=\zxwhite, pos=0.03]{} node[zxvertex=\zxwhite, pos=0.97]{} (0.75, 3.5) to [out=up, in=-45] (0,4.5) to (0,5)  to [out=up, in=up,looseness=2] node[zxvertex=\zxwhite, pos=0.5] {} (2.5, 5) to (2.5,0.5) to [out=down, in=down,looseness=2] node[zxvertex=\zxwhite, pos=0.5]{} (0,0.5);
\draw (0,1) to [out=45, in=down] (0.75,2) to [out=up, in=down] node[zxvertex=\zxwhite, pos=0.03]{} node[zxvertex=\zxwhite, pos=0.97]{}(-0.75, 3.5) to [out=up, in=-135] (0,4.5);
\node[zxvertex=\zxwhite,zxup] at (0,1){};
\node[zxvertex=\zxwhite,zxdown] at (0,4.5){};
\end{tz}
~~=\frac{1}{n}~
\sum_{v\in V_{\Gamma}} ~
\begin{tz}[zx,scale=\scl]
\draw[arrow data ={0.25}{<}, arrow data ={0.5}{<},arrow data={0.75}{<}, arrow data={1}{<}] (0, 2.75) ellipse (1.5 and 0.8);
\draw  (-0.75,1) to (-0.75,2) to [out=up, in=down] node[zxvertex=\zxwhite, pos=0.03]{} node[zxvertex=\zxwhite, pos=0.97]{} (0.75, 3.5) to (0.75, 4.5);
\draw (0.75, 1) to  (0.75,2) to [out=up, in=down] node[zxvertex=\zxwhite, pos=0.03]{} node[zxvertex=\zxwhite, pos=0.97]{}(-0.75, 3.5) to  (-0.75,4.5);
\node[zxnode] at (-0.75,1) {$v$};
\node[zxnode] at (-0.75,4.5) {$v$};
\node[zxnode] at (0.75,1) {$v$};
\node[zxnode] at (0.75,4.5) {$v$};
\end{tz}
\]
In the first equation, we have introduced a circle~\eqref{eq:closedcircle} to the right of the diagram and then enlarged this circle over parts of the diagram, converting grey $V_{\Gamma_X}$-nodes into white $V_{\Gamma}$-nodes in the process. 
In the last equation we used the expression~\eqref{eq:classicalcopy} for the commutative special dagger Frobenius algebra $V_{\Gamma}$. 
\ignore{
\DR{suggestion: get rid of first picture? Suggestion: Write equation for quantum iso in beginning as white black. Say: we already know that P is a quantum iso from white to black. Together with the equatins from beginning can freely move bubbles, etc. }For the second equation, we pulled the comultiplication downwards through a doubled wire, pulled the multiplication upwards through a doubled wire, and removed two bubbles; \DRcomm{we then pulled the bottom Hilbert space loop up around the braiding and over a quantum set wire, before pulling the cap of the Frobenius algebra downwards through a doubled wire and removing a third bubble.}\DR{where's a braiding here? Not entirely clear what's going on} }

Using $X=P \circ \conj{P}$, and untangling the above equation leads to the formula~\eqref{eq:centerquantumgraph}.
\ignore{\def\ang{-17}
\[ 
\begin{tz}[zx,yscale=-1,every to/.style={out=up, in=down}]
\draw  (0.75, 0.75) to  (0.75,1.75) to (0.75,2) to  (-0.75, 3.) to (-0.75,4);
\draw (-0.75,0.75) to (-0.75,1.75) to (-0.75,2) to  (0.75, 3.) to  (0.75,4);
\draw[string,arrow data={0.1}{<}, arrow data={0.95}{<}]  (1.75,0.75) to  [looseness=1.1, in looseness=0.9] node[zxvertex=\zxwhite, pos=0.32] {} node[zxvertex=\zxwhite, pos=0.55] {} (-2.5, 3.2) to (-2.5,4);
\draw[string,arrow data={0.1}{>}, arrow data={0.95}{>}] (2.75,0.75) to (2.75,1.5) to  [looseness=1.1, in looseness=0.9] node[zxvertex=\zxwhite, pos=0.48] {} node[zxvertex=\zxwhite, pos=0.73] {}(-1.5,4);
\end{tz}
~~=~~
\begin{tz}[zx,xscale=-1, every to/.style={out=up, in=down}]
\node[zxnode=\zxwhite] (XR) at (0.75, 2.3) {$X$};
\node[zxnode=\zxwhite] (XL)  at (-0.75, 2.65) {$X$};
\draw[string]  (0.75,0.75) to  (0.75,1.75) to (0.75,4);
\draw[string] (-0.75,0.75) to (-0.75,1.75) to (-0.75,4);
\draw[string,arrow data={0.1}{>}, arrow data={0.95}{>}] (2,0.75) to (2,0.9) to  [looseness=1.1, out looseness=0.9] (-2.25, 3.7) to (-2.25,4);
\draw[string, arrow data={0.5}{<}] (XL) to [out=135, in=\ang,in looseness=2] (0.5, 3.1) to [out=\ang-180, in=25, out looseness=2.5] (XR);
\draw[string,arrow data={0.9}{<}] (XR) to [out=135] (-1.5,4);
\draw[string,arrow data={0.95}{>}] (XL) to [out=5, in =up] (2.75,2) to (2.75,0.75);
\end{tz}
\]}%
\ignore{Plugging this into the above equation and untangling the result leads to the formula~\eqref{eq:centerquantumgraph}. }In particular, $\Gamma_X$ is classical if $V_{\Gamma_X}$ is commutative, that is, if $\dim(Z(V_{\Gamma_X})) = \dim(V_{\Gamma_X})$. Since quantum isomorphisms preserve dimensions (Proposition~\ref{prop:dimensionpreserved}), we have $\dim(V_{\Gamma_X}) = |V_{\Gamma}|$. Thus, $\Gamma_X$ is classical if and only if $\dim(Z(V_{\Gamma_X})) = |V_{\Gamma}|$.
\end{proof}

\ignore{
\DRcomm{Maybe remove:

\begin{remark} Since commutativity of a monoid is preserved under isomorphism, the classification statement in Corollary~\ref{cor:bigclassification} implies that the condition of Theorem~\ref{thm:commutativitycondition} is invariant under Morita equivalence of Frobenius monoids $X$.\ignore{This means that \eqref{eq:commutativeFrobenius} will usually have a more conceptual expression.} \ignore{For example, in Section~\ref{sec:construction} we will see that certain Frobenius monoids of central type come from actual subgroups of central type $(L,\psi)$ of $\Aut(\Gamma)$. In this case, the condition of Theorem~\ref{thm:commutativitycondition} can be expressed purely in terms of the subgroup structure $L\subseteq \Aut(\Gamma)$ and the cohomology class of the $2$-cocycle $\psi$.}
\end{remark}

}
}

\begin{remark} In contrast to the classification of quantum isomorphic quantum graphs (see Remark~\ref{rem:onlyquantumautos}), the condition in Theorem~\ref{thm:commutativitycondition} does not only depend on the abstract monoidal category with fibre functor $\QAut(\Gamma)$. In the language of compact quantum groups, the classification of classical graphs $\Gamma'$ which are quantum isomorphic to a classical graph $\Gamma$ depends both on the quantum automorphism group of $\Gamma$ and its action on $V_{\Gamma}$. 
\ignore{
From a categorical perspective, we therefore expect the commutativity condition of Theorem~\ref{thm:commutativitycondition} to be expressible as an abstract condition on the monoidal category $\mathcal{C}:= \QAut(\Gamma)$ and its module category $\mathcal{M} := \QGraph (*, \Gamma)$. \ignore{\DRcomm{but have not yet identified such a condition which only makes use of the monoidal structure on}}\DR{clear?}\DV{Yes. But why is it a useful perspective?}\DR{because we cannot yet phrase the condition of theorem above for general concrete monoidal categories. Somehow it depends on our model of QGraph which is always bad. But can get rid of this remark }
}%
\end{remark}
\noindent
We therefore obtain a classification of classical graphs which are quantum isomorphic to a classical graph $\Gamma$ in terms of simple dagger Frobenius monoids in $\QAut(\Gamma)$.
\begin{corollary} \label{cor:superclassification}Let $\Gamma$ be a classical graph. Then, the construction of Proposition~\ref{prop:quantumbijectionFrobenius} induces a bijective correspondence between the following structures:
\begin{itemize}
\item Isomorphism classes of classical graphs $\Gamma'$ such that there exists a quantum isomorphism {$\Gamma' \to \Gamma$}.
\item Morita equivalence classes of simple dagger Frobenius monoids in $\QAut(\Gamma)$ for which the expression~\eqref{eq:centerquantumgraph} evaluates to $|V_{\Gamma}|$.
\ignore{
\item Pseudo-telepathic classical graph pairs $(G,H)$ up to isomorphism of $H$.
\item Morita equivalence classes of Morita non-trivial Frobenius monoids of central type in $\QAut(G)$ fulfilling the conditions of Theorem~\ref{thm:commutativitycondition}.}
\end{itemize}
\end{corollary}
\begin{proof}Corollary~\ref{cor:superclassification} follows from restricting the classification of quantum isomorphic quantum graphs (Corollary~\ref{cor:bigclassification}) to Morita equivalence classes of simple dagger Frobenius monoids fulfilling the conditions of Theorem~\ref{thm:commutativitycondition} and their associated classical graphs.
\end{proof}

\subsection{Proof of Theorem~\ref{thm:Frobeniussplit}}\label{sec:proof}
In this section, we prove Theorem~\ref{thm:Frobeniussplit}. We first introduce two technical propositions.
\begin{proposition} \label{prop:quantumcondition}Let $A$ and $B$ be special symmetric dagger Frobenius algebras, and let $P:H \otimes A\to B\otimes H $ be a linear map fulfilling the first two equations of~\eqref{eq:quantumfunction} and the first two equations of~\eqref{eq:quantumfunction2}. Then, $P$ is unitary if and only if it also fulfils the last equation of~\eqref{eq:quantumfunction}.
\end{proposition}
\begin{proof} The `if'-direction follows immediately from Theorem~\ref{thm:dualisable} (in the special case where the quantum adjacency matrices are identities). For the other direction, observe that if $P$ is unitary, then the following implication holds:
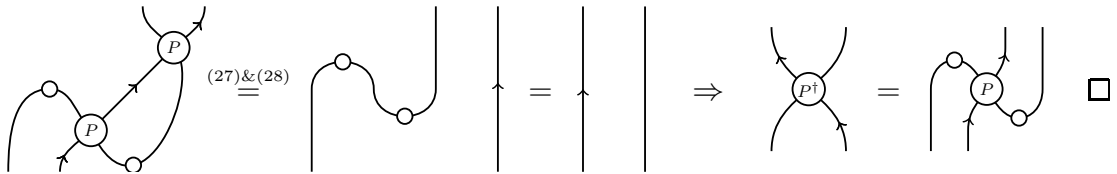
\[
\begin{tz}[zx,xscale=-1,scale=-1]
\clip (-3.05,0) rectangle (1.8,4);
\draw[arrow data={0.1}{<},arrow data={0.5}{<}, arrow data={0.95}{<}] (1.75,0) to [out=up,in=-45] (1,1) to (-1,3) to [out=135, in=down] (-1.75,4);
\draw (0.25,0) to [out=up, in=-135] (1,1) to [in looseness=2.5,out=45, in=65] node[pos=0.72](c){} (-1,3) to [out=-135, in=right] (-2,2) to [out=left, in=down] (-3,4) ;
\node[zxnode=\zxwhite] at (1,1) {$P$};
\node[zxnode=\zxwhite] at (-1,3) {$P$};
\node[zxvertex=\zxwhite] at (c.center){};
\node[zxvertex=\zxwhite] at (-2,2){};
\end{tz}
~~\superequals{\eqref{eq:quantumfunction}\&\eqref{eq:quantumfunction2}}~~~~
\begin{tz}[zx,xscale=-1,scale=-1]
\draw (0,0) to (0, 2) to [looseness=1.5,out=up, in=up] node[zxvertex=\zxwhite,pos=0.5]{} (-1.5,2) to [looseness=1.5,out=down, in=down]node[zxvertex=\zxwhite,pos=0.5]{} (-3,2) to (-3,4);
\draw[string,arrow data={0.5}{<}] (1.5,0) to + (0,4);
\end{tz}
~~=~~
\begin{tz}[zx,xscale=-1]
\draw (0,0) to + (0,4);
\draw[string,arrow data={0.5}{>}] (1.5,0) to + (0,4);
\end{tz}
\hspace{0.5cm} \Rightarrow \hspace{0.5cm}
\begin{tz}[zx,xscale=-1,every to/.style={out=up, in=down},xscale=0.8,scale=1]
\draw [arrow data={0.2}{>},arrow data={0.8}{>}]  (0,0) to (2.25,3);
\draw (2.25,0) to node[zxnode=\zxwhite, pos=0.5] {$P^{\dagger}$} (0,3);
\end{tz}
~~=~~
\begin{tz}[zx,xscale=-0.6,yscale=-1]
\draw[arrow data={0.5}{<}] (0.25,-0.5) to (0.25,0) to [out=up, in=-135] (1,1);
\draw (1,1) to [out=135, in=right] node[zxvertex=\zxwhite, pos=1]{} (-0.3, 1.7) to [out=left, in=up] (-1.25,1) to (-1.25,-0.5);
\draw[arrow data={0.5}{>}] (1.75,2.5) to (1.75,2) to [out=down, in=45] (1,1);
\draw (1,1) to [out= -45, in= left] node[zxvertex=\zxwhite, pos=1] {} (2.3,0.3) to [out=right, in=down] (3.25,1) to (3.25,2.5);
\node [zxnode=\zxwhite] at (1,1) {$P$};
\end{tz} \qedhere
\]
\end{proof}
\noindent
We adopt the following terminology, originally introduced in~\cite{Ocneanu1989,Jones1999} and adapted to a categorical setting in~\cite{Vicary2012hq, Reutter2016}:
\begin{definition} Let $A,B$ and $H$ be Hilbert spaces. A linear map $P:H \otimes A \to B \otimes H$ is \textit{biunitary}, if it and the following `quarter-rotation' are unitary:
\begin{equation}\label{eq:biinv}\begin{tz}[zx,xscale=-0.6,yscale=1,scale=1]
\draw(0.25,-0.5) to (0.25,0) to [out=up, in=-135] (1,1);
\draw [arrow data={0.34}{>}]  (1,1) to [out=135, in=right]  (-0.3, 1.7) to [out=left, in=up] (-1.25,1) to (-1.25,-0.5);
\draw(1.75,2.5) to (1.75,2) to [out=down, in=45] (1,1);
\draw[arrow data={0.32}{<}] (1,1) to [out= -45, in= left]   (2.3,0.3) to [out=right, in=down] (3.25,1) to (3.25,2.5);
\node [zxnode=\zxwhite] at (1,1) {$P$};
\end{tz}
\end{equation}
\end{definition}
\noindent

\noindent
From now on, we will use the shorthand~\eqref{eq:shorthand} for $P$. \ignore{Although in all cases the partially transposed and/or conjugated biunitary is represented as a white dot, as long as we only use duality of the oriented wire, we can always determine the resulting 2-morphism from its type. \ignore{This is exactly as in Remark~\ref{rem:convention}, except that now we only transpose the oriented wire.}
It will be useful to note that, as one travels along the $H$-wire in the direction of its orientation, the Hilbert space $A$ is always to the left, while $B$ is always on the right.
}It can straightforwardly be verified that a morphism $P:H \otimes A \to B \otimes H$ is biunitary if and only if the equations~\eqref{eq:biunitarybraiding} and~\eqref{eq:biunitarybraiding2} hold. 

The following proposition allows linear maps that can pull through a double wire to `jump' over a single wire, acquiring a surrounding bubble as they do so.
\def\d{2}
\def\halfheightbox{0.4}
\def\wwhite{1}
\def\ang{-17}
\begin{proposition}\label{prop:jumpinglemma}Let $S:H \otimes A \to B \otimes H$ be a biunitary linear map, written using the conventions above, and let $n=\dim(H)$. Let $e:B^{\otimes k} \to B^{\otimes r}$ be a linear map between tensor powers of $A$ fulfilling the following:
\begin{equation}\label{eq:jump}
\begin{tz}[zx,scale=-1,xscale=-1]
\draw[string] (0,0) to node[front,zxvertex=\zxwhite,pos=0.59] {} node[front,zxvertex=\zxwhite, pos=0.78] {} (0,4.5);
\draw[string] (\d,0) to node[front,zxvertex=\zxwhite, pos=0.44] {} node[front,zxvertex=\zxwhite, pos=0.65] {} (\d,4.5);
\draw[string,fill=white] (-0.2, 1-\halfheightbox) rectangle (\d+0.2, 1+\halfheightbox);
\node[scale=0.8] at (0.5*\d, 1) {$e$};
\draw[string, arrow data={0.08}{<}, arrow data ={0.95}{<}] (3, 0 )  to (3,0.5) to [out=up, in=down] (-2,4.5);
\draw[string, arrow data={0.08}{>}, arrow data={0.97}{>}] (4,0) to (4,0.5) to [out=up, in=down, out looseness=1.3, in looseness=0.7] (-1, 4.5);
\fill[white] (0.5*\d-0.5*\wwhite,1.8) rectangle (0.5*\d+0.5*\wwhite, 4.5);
\node[rotate=\ang] at (0.5*\d, 2.25+\halfheightbox) {$\cdots$};
\node at (0.5*\d, 0.2) {$\ldots$};
\end{tz}
~~=~~
\begin{tz}[zx,xscale=-1]
\draw[string] (0,0) to node[front,zxvertex=\zxwhite,pos=0.59] {} node[front,zxvertex=\zxwhite, pos=0.78] {} (0,4.5);
\draw[string] (\d,0) to node[front,zxvertex=\zxwhite, pos=0.44] {} node[front,zxvertex=\zxwhite, pos=0.65] {} (\d,4.5);
\draw[string,fill=white] (-0.2, 1-\halfheightbox) rectangle (\d+0.2, 1+\halfheightbox);
\node[scale=0.8] at (0.5*\d, 1) {$e$};
\draw[string, arrow data={0.08}{<}, arrow data ={0.95}{<}] (3, 0 )  to (3,0.5) to [out=up, in=down] (-2,4.5);
\draw[string, arrow data={0.08}{>}, arrow data={0.97}{>}] (4,0) to (4,0.5) to [out=up, in=down, out looseness=1.3, in looseness=0.7] (-1, 4.5);
\fill[white] (0.5*\d-0.5*\wwhite,1.8) rectangle (0.5*\d+0.5*\wwhite, 4.5);
\node[rotate=\ang] at (0.5*\d, 2.25+\halfheightbox) {$\cdots$};
\node at (0.5*\d, 0.2) {$\ldots$};
\end{tz}
\end{equation}
Then, the following holds:
\begin{equation}\nonumber
\begin{tz}[zx,scale=-1,xscale=-1]
\draw[string] (0,0) to node[front,zxvertex=\zxwhite,pos=0.59] {} (0,4.5);\draw[string] (\d,0) to node[front,zxvertex=\zxwhite, pos=0.41] {} (\d,4.5);
\draw[string,fill=white] (-0.2, 1-\halfheightbox) rectangle (\d+0.2, 1+\halfheightbox);
\node[scale=0.8] at (0.5*\d, 1) {$e$};
\draw[string, arrow data={0.08}{<}, arrow data ={0.95}{<}] (3, 0 )  to (3,0.5) to [out=up, in=down] (-1.,4) to (-1.,4.5);
\fill[white] (0.5*\d-0.5*\wwhite,1.8) rectangle (0.5*\d+0.5*\wwhite, 4.5);
\node at (0.5*\d, 0.2) {$\ldots$};
\node at (0.5*\d, 3.5) {$\ldots$};
\end{tz}
=
\begin{tz}[zx,xscale=-1]
\draw[string] (0,0) to node[front,zxvertex=\zxwhite,pos=0.59] {} (0,4.5);\draw[string] (\d,0) to node[front,zxvertex=\zxwhite, pos=0.41] {} (\d,4.5);
\draw[string,fill=white] (-0.2, 1-\halfheightbox) rectangle (\d+0.2, 1+\halfheightbox);
\node[scale=0.8] at (0.5*\d, 1) {$e'$};
\draw[string, arrow data={0.08}{>}, arrow data ={0.95}{>}] (3, 0 )  to (3,0.5) to [out=up, in=down] (-1.,4) to (-1.,4.5);
\fill[white] (0.5*\d-0.5*\wwhite,1.8) rectangle (0.5*\d+0.5*\wwhite, 4.5);
\node at (0.5*\d, 0.2) {$\ldots$};
\node at (0.5*\d, 3.5) {$\ldots$};
\end{tz}
\hspace{0.2cm} \text{where } \hspace{0.2cm} \begin{tz}[zx]
\draw[string] (0,0) to  (0,4.5);
\draw[string] (\d,0) to  (\d,4.5);
\draw[string,fill=white] (-0.2, 2.25-\halfheightbox) rectangle (\d+0.2, 2.25+\halfheightbox);
\node[scale=0.8] at (0.5*\d, 2.25) {$e'$};
\fill[white] (0.5*\d-0.5*\wwhite,0.1) rectangle (0.5*\d+0.5*\wwhite, 1);
\fill[white] (0.5*\d-0.5*\wwhite,3.5) rectangle (0.5*\d+0.5*\wwhite, 4.4);
\node at (0.5*\d, 0.7) {$\ldots$};
\node at (0.5*\d, 3.8) {$\ldots$};
\end{tz}
~~=~\frac{1}{n}~
\begin{tz}[zx]
\draw[string] (0,0) to node[front,zxvertex=\zxwhite,pos=0.15] {} node[front,zxvertex=\zxwhite,pos=0.855] {}  (0,4.5);\draw[string] (\d,0) to node[front,zxvertex=\zxwhite, pos=0.23] {} node[front,zxvertex=\zxwhite,pos=0.765] {} (\d,4.5);
\draw[string,fill=white] (-0.2, 2.25-\halfheightbox) rectangle (\d+0.2, 2.25+\halfheightbox);
\node[scale=0.8] at (0.5*\d, 2.25) {$e$};
\draw[string,arrow data={0.01}{>},arrow data = {0.33}{>},arrow data={0.655}{>}] (-0.5, 0.5) to [out=20, in=down] (3,2.25) to [out=up, in=-20] (-0.5, 4)  to [out=160 , in=200] (-0.5, 0.5) ;
\fill[white] (0.5*\d-0.5*\wwhite,0.1) rectangle (0.5*\d+0.5*\wwhite, 1);
\fill[white] (0.5*\d-0.5*\wwhite,3.5) rectangle (0.5*\d+0.5*\wwhite, 4.4);
\node at (0.5*\d, 0.2) {$\ldots$};
\node at (0.5*\d, 4.3) {$\ldots$};
\end{tz}
~=\frac{1}{n}~
\begin{tz}[zx,xscale=-1]
\draw[string] (0,0) to node[front,zxvertex=\zxwhite,pos=0.15] {} node[front,zxvertex=\zxwhite,pos=0.855] {}  (0,4.5);\draw[string] (\d,0) to node[front,zxvertex=\zxwhite, pos=0.23] {} node[front,zxvertex=\zxwhite,pos=0.765] {} (\d,4.5);
\draw[string,fill=white] (-0.2, 2.25-\halfheightbox) rectangle (\d+0.2, 2.25+\halfheightbox);
\node[scale=0.8] at (0.5*\d, 2.25) {$e$};
\draw[string,arrow data={0.01}{<},arrow data = {0.33}{<},arrow data={0.655}{<}] (-0.5, 0.5) to [out=20, in=down] (3,2.25) to [out=up, in=-20] (-0.5, 4)  to [out=160 , in=200] (-0.5, 0.5) ;
\fill[white] (0.5*\d-0.5*\wwhite,0.1) rectangle (0.5*\d+0.5*\wwhite, 1);
\fill[white] (0.5*\d-0.5*\wwhite,3.5) rectangle (0.5*\d+0.5*\wwhite, 4.4);
\node at (0.5*\d, 0.2) {$\ldots$};
\node at (0.5*\d, 4.3) {$\ldots$};
\end{tz}
\end{equation}
Moreover, if $f:B^{\otimes l} \to B^{\otimes k}$ and $e:B^{\otimes k}\to B^{\otimes r}$  both fulfil~\eqref{eq:jump}, it follows that $(1_{B^{\otimes k}})' = 1_{A^{\otimes k}}$ and $(ef)' = e'f'$.\ignore{ and $e''=e$ where the second $'$ is with respect to the biunitary $\overline{S}:B\otimes H^* \to A\otimes H^*$ defined in~\eqref{eq:daggerdual}.}
\end{proposition}
\begin{proof}
\[
\begin{tz}[zx,scale=-1]
\draw[string] (0,-2) to node[front,zxvertex=\zxwhite, pos=0.88]{} (0,0);
\draw[string] (\d,-2) to node[front,zxvertex=\zxwhite, pos=1.12]{} (\d,0);
\draw[string, arrow data={0.05}{<},arrow data ={0.9}{<}] (-1,-2) to (-1,-1.5) to [out=up, in=down] (4, 1.8) to (4, 4.5);
\draw[string,fill=white] (-0.2, -1.2-\halfheightbox) rectangle (\d+0.2, -1.2+\halfheightbox);
\node[scale=0.8] at (0.5*\d, -1.2){$e$};
\draw[string] (0,0) to  (0,4.5);
\draw[string] (\d,0) to (\d,4.5);
\ignore{
\draw[string,fill=white] (-0.2, 2.25-\halfheightbox) rectangle (\d+0.2, 2.25+\halfheightbox);
\node[scale=0.8] at (0.5*\d, 2.25) {$e$};}
\ignore{
\draw[string,arrow data={0.01}{<},arrow data = {0.33}{<},arrow data={0.655}{<}] (-0.5, 0.5) to [out=20, in=down] (3,2.25) to [out=up, in=-20] (-0.5, 4)  to [out=160 , in=200] (-0.5, 0.5) ;}
\fill[white] (0.5*\d-0.5*\wwhite,-0.6) rectangle (0.5*\d+0.5*\wwhite, 1);
\fill[white] (0.5*\d-0.5*\wwhite,3.5) rectangle (0.5*\d+0.5*\wwhite, 4.4);
\node at (0.5*\d, -1.88) {$\ldots$};
\node at (0.5*\d, 4.3) {$\ldots$};
\end{tz}
\!\!\superequalseq{eq:closedcircle}~
\frac{1}{n}~~
\begin{tz}[zx,scale=-1]
\draw[arrow data={0.51}{<}] (-1,2.25) circle (0.5);
\draw[string] (0,-2) to node[front,zxvertex=\zxwhite, pos=0.88]{} (0,0);
\draw[string] (\d,-2) to node[front,zxvertex=\zxwhite, pos=1.12]{} (\d,0);
\draw[string, arrow data={0.05}{<},arrow data ={0.9}{<}] (-1,-2) to (-1,-1.5) to [out=up, in=down] (4, 1.8) to (4, 4.5);
\draw[string,fill=white] (-0.2, -1.2-\halfheightbox) rectangle (\d+0.2, -1.2+\halfheightbox);
\node[scale=0.8] at (0.5*\d, -1.2){$e$};
\draw[string] (0,0) to  (0,4.5);
\draw[string] (\d,0) to (\d,4.5);
\ignore{
\draw[string,fill=white] (-0.2, 2.25-\halfheightbox) rectangle (\d+0.2, 2.25+\halfheightbox);
\node[scale=0.8] at (0.5*\d, 2.25) {$e$};}
\ignore{
\draw[string,arrow data={0.01}{<},arrow data = {0.33}{<},arrow data={0.655}{<}] (-0.5, 0.5) to [out=20, in=down] (3,2.25) to [out=up, in=-20] (-0.5, 4)  to [out=160 , in=200] (-0.5, 0.5) ;}
\fill[white] (0.5*\d-0.5*\wwhite,-0.6) rectangle (0.5*\d+0.5*\wwhite, 1);
\fill[white] (0.5*\d-0.5*\wwhite,3.5) rectangle (0.5*\d+0.5*\wwhite, 4.4);
\node at (0.5*\d, -1.88) {$\ldots$};
\node at (0.5*\d, 4.3) {$\ldots$};
\end{tz}
~~\superequals{\eqref{eq:biunitarybraiding}\&\eqref{eq:biunitarybraiding2}}~~~
\frac{1}{n}~
\begin{tz}[zx,scale=-1]
\draw[string] (0,-2) to node[front,zxvertex=\zxwhite, pos=0.88]{} (0,0);
\draw[string] (\d,-2) to node[front,zxvertex=\zxwhite, pos=1.12]{} (\d,0);
\draw[string, arrow data={0.05}{<},arrow data ={0.9}{<}] (-1,-2) to (-1,-1.5) to [out=up, in=down] (4, 1.8) to (4, 4.5);
\draw[string,fill=white] (-0.2, -1.2-\halfheightbox) rectangle (\d+0.2, -1.2+\halfheightbox);
\node[scale=0.8] at (0.5*\d, -1.2){$e$};
\draw[string] (0,0) to node[front,zxvertex=\zxwhite,pos=0.15] {} node[front,zxvertex=\zxwhite,pos=0.855] {}  (0,4.5);\draw[string] (\d,0) to node[front,zxvertex=\zxwhite, pos=0.23] {} node[front,zxvertex=\zxwhite,pos=0.765] {} (\d,4.5);
\ignore{
\draw[string,fill=white] (-0.2, 2.25-\halfheightbox) rectangle (\d+0.2, 2.25+\halfheightbox);
\node[scale=0.8] at (0.5*\d, 2.25) {$e$};}
\draw[string,arrow data={0.01}{>},arrow data = {0.33}{>},arrow data={0.655}{>}] (-0.5, 0.5) to [out=20, in=down] (3,2.25) to [out=up, in=-20] (-0.5, 4)  to [out=160 , in=200] (-0.5, 0.5) ;
\fill[white] (0.5*\d-0.5*\wwhite,-0.6) rectangle (0.5*\d+0.5*\wwhite, 1);
\fill[white] (0.5*\d-0.5*\wwhite,3.5) rectangle (0.5*\d+0.5*\wwhite, 4.4);
\node at (0.5*\d, -1.88) {$\ldots$};
\node at (0.5*\d, 4.3) {$\ldots$};
\end{tz}
~\superequalseq{eq:jump}~\frac{1}{n}~
\begin{tz}[zx,scale=-1]
\draw[string] (0,-2) to node[front,zxvertex=\zxwhite, pos=0.88]{} (0,0);
\draw[string] (\d,-2) to node[front,zxvertex=\zxwhite, pos=1.12]{} (\d,0);
\draw[string, arrow data={0.05}{<},arrow data ={0.9}{<}] (-1,-2) to (-1,-1.5) to [out=up, in=down] (4, 1.8) to (4, 4.5);
\draw[string] (0,0) to node[front,zxvertex=\zxwhite,pos=0.15] {} node[front,zxvertex=\zxwhite,pos=0.855] {}  (0,4.5);\draw[string] (\d,0) to node[front,zxvertex=\zxwhite, pos=0.23] {} node[front,zxvertex=\zxwhite,pos=0.765] {} (\d,4.5);
\draw[string,fill=white] (-0.2, 2.25-\halfheightbox) rectangle (\d+0.2, 2.25+\halfheightbox);
\node[scale=0.8] at (0.5*\d, 2.25) {$e$};
\draw[string,arrow data={0.01}{>},arrow data = {0.33}{>},arrow data={0.655}{>}] (-0.5, 0.5) to [out=20, in=down] (3,2.25) to [out=up, in=-20] (-0.5, 4)  to [out=160 , in=200] (-0.5, 0.5) ;
\fill[white] (0.5*\d-0.5*\wwhite,-0.6) rectangle (0.5*\d+0.5*\wwhite, 1);
\fill[white] (0.5*\d-0.5*\wwhite,3.5) rectangle (0.5*\d+0.5*\wwhite, 4.4);
\node at (0.5*\d, -1.88) {$\ldots$};
\node at (0.5*\d, 4.3) {$\ldots$};
\end{tz}
\] The statements $(1_{B^{\otimes k}})' = 1_{A^{\otimes k}}$ and $(ef)' = e'f'$ are verified analogously.
\end{proof}
\begin{remark} Proposition~\ref{prop:jumpinglemma} is closely related to standard techniques in the setting of \emph{planar algebras}. In particular, it is analogous to~\cite[Proposition~2.11.7 and Theorem~2.11.8]{Jones1999}.
\end{remark}
\noindent
We now prove Theorem~\ref{thm:Frobeniussplit}.
\splitfrobenius*
\begin{proof}
A simple dagger Frobenius monoid in $\QAut(\Gamma)$ is $*$-isomorphic to a quantum isomorphism $(H\otimes H^*, X): \Gamma \to \Gamma$, represented by a linear map $X: (H\otimes H^*) \otimes V_{\Gamma}\to V_{\Gamma} \otimes (H\otimes H^*)$, fulfilling: 
 
 \def\dt{0.2}
 \def\db{0.2}
 \begin{calign}\label{eq:proofFrobenius1}
 \hspace{-0.9cm}
 \begin{tz}[zx,xscale=-1, every to/.style={out=up, in=down}]
   \clip (-1.1,0) rectangle (1.6,3);
 \draw[arrow data={0.2}{>}, arrow data={0.83}{>}] (\db, 0) to [out=up, in=up, looseness=4] (1-\db,0);
 \draw[arrow data={0.3}{<}, arrow data={0.91}{<}] (-\db, 0) to (0.5-\dt, 2) to (0.5-\dt, 3);
 \draw [arrow data={0.3}{>}, arrow data={0.91}{>}](1+\db, 0) to (0.5+\dt,2) to (0.5+\dt, 3);
 \draw (-1,0) to (-1,0.25)to[in looseness=0.8] (1.5,3);
 \node[zxnode=\zxwhite] at (0.5,1.85) {$X$};
 \ignore{
 \node[dimension,left] at (-1,0) {$A$};
 \node[dimension, right] at (1.5,3) {$A$};
 \node[dimension, left] at (
 }
 \end{tz}
 =
  \begin{tz}[zx,xscale=-1,every to/.style={out=up, in=down}]
    \clip (-1.1,0) rectangle (1.85,3);
 \draw[arrow data={0.1}{>}, arrow data={0.5}{>}, arrow data={0.95}{>}] (\db, 0) to (\db,1.25) to  [out=up, in=up, looseness=4] (1-\db,1.25) to (1-\db, 0);
 \draw[arrow data={0.13}{<}, arrow data={0.87}{<}] (-\db, 0) to (-\db, 1.25) to  (0.5-\dt, 3);
 \draw [arrow data={0.13}{>}, arrow data={0.87}{>}](1+\db, 0) to  (1+\db, 1.25) to (0.5+\dt,3) ;
 \draw (-1,0) to [in looseness=0.8] node[zxnode=\zxwhite, pos=0.405] {$X$} node[zxnode=\zxwhite, pos=0.65]{$X$} (1.75,2.) to  (1.75,3);
 \end{tz}
 &
  \begin{tz}[zx,xscale=-1,every to/.style={out=up, in=down}]
  \clip (-1.1,0) rectangle (1.6,3);
  \draw[arrow data={0.13}{>},arrow data={0.44}{>},arrow data={0.65}{>},arrow data={0.93}{>}](0.5-\dt, 3) to (0.5-\dt, 1.4) to [out=down, in=down, looseness=6] (0.5+\dt, 1.4) to (0.5+\dt, 3);
 \draw (-1,0) to (-1,0.25)to[in looseness=0.8] (1.5,3);
 \node[zxnode=\zxwhite] at (0.5,1.85) {$X$};
 \ignore{
 \node[dimension,left] at (-1,0) {$A$};
 \node[dimension, right] at (1.5,3) {$A$};
 \node[dimension, left] at (
 }
 \end{tz}
=
  \begin{tz}[zx,xscale=-1,every to/.style={out=up, in=down}]
  \clip (-1.1,0) rectangle (1.6,3);
  \draw[arrow data={0.2}{>},arrow data={0.86}{>}](0.5-\dt, 3) to (0.5-\dt, 2.7) to [out=down, in=down, looseness=6] (0.5+\dt, 2.7) to (0.5+\dt, 3);
 \draw (-1,0) to[in looseness=0.8]  (1.5,2.3) to (1.5,3);
 \end{tz}
 &
  \begin{tz}[zx,xscale=-1,every to/.style={out=up, in=down},scale=-1]
    \clip (-1.1,0) rectangle (1.85,3);
 \draw[arrow data={0.1}{>}, arrow data={0.5}{>}, arrow data={0.95}{>}] (\db, 0) to (\db,1.25) to  [out=up, in=up, looseness=4] (1-\db,1.25) to (1-\db, 0);
 \draw[arrow data={0.13}{<}, arrow data={0.87}{<}] (-\db, 0) to (-\db, 1.25) to  (0.5-\dt, 3);
 \draw [arrow data={0.13}{>}, arrow data={0.87}{>}](1+\db, 0) to  (1+\db, 1.25) to (0.5+\dt,3) ;
 \draw (-1,0) to [in looseness=0.8] node[zxnode=\zxwhite, pos=0.405] {$X$} node[zxnode=\zxwhite, pos=0.65]{$X$} (1.75,2.) to  (1.75,3);
 \end{tz}
 =
 \begin{tz}[zx,xscale=-1, every to/.style={out=up, in=down},scale=-1]
   \clip (-1.1,0) rectangle (1.6,3);
 \draw[arrow data={0.2}{>}, arrow data={0.83}{>}] (\db, 0) to [out=up, in=up, looseness=4] (1-\db,0);
 \draw[arrow data={0.3}{<}, arrow data={0.91}{<}] (-\db, 0) to (0.5-\dt, 2) to (0.5-\dt, 3);
 \draw [arrow data={0.3}{>}, arrow data={0.91}{>}](1+\db, 0) to (0.5+\dt,2) to (0.5+\dt, 3);
 \draw (-1,0) to (-1,0.25)to[in looseness=0.8] (1.5,3);
 \node[zxnode=\zxwhite] at (0.5,1.85) {$X$};
 \ignore{
 \node[dimension,left] at (-1,0) {$A$};
 \node[dimension, right] at (1.5,3) {$A$};
 \node[dimension, left] at (
 }
 \end{tz}
 &
  \begin{tz}[zx,xscale=-1,every to/.style={out=up, in=down},scale=-1]
  \clip (-1.1,0) rectangle (1.6,3);
  \draw[arrow data={0.2}{>},arrow data={0.86}{>}](0.5-\dt, 3) to (0.5-\dt, 2.7) to [out=down, in=down, looseness=6] (0.5+\dt, 2.7) to (0.5+\dt, 3);
 \draw (-1,0) to[in looseness=0.8]  (1.5,2.3) to (1.5,3);
 \end{tz}
 =
  \begin{tz}[zx,xscale=-1,every to/.style={out=up, in=down},scale=-1]
  \clip (-1.1,0) rectangle (1.6,3);
  \draw[arrow data={0.13}{>},arrow data={0.44}{>},arrow data={0.65}{>},arrow data={0.93}{>}](0.5-\dt, 3) to (0.5-\dt, 1.4) to [out=down, in=down, looseness=6] (0.5+\dt, 1.4) to (0.5+\dt, 3);
 \draw (-1,0) to (-1,0.25)to[in looseness=0.8] (1.5,3);
 \node[zxnode=\zxwhite] at (0.5,1.85) {$X$};
 \ignore{
 \node[dimension,left] at (-1,0) {$A$};
 \node[dimension, right] at (1.5,3) {$A$};
 \node[dimension, left] at (
 }
 \end{tz}
 \end{calign}
 We first note that since $X$ is a quantum isomorphism and therefore unitary (Theorem~\ref{thm:dualisable}), the following holds:
  \begin{calign}\label{eq:proofFrobenius2}
  \hspace{-0.5cm}
   \begin{tz}[zx,xscale=-1,every to/.style={out=up, in=down}]
    \clip (-1.1,0) rectangle (1.85,3);
 \draw[arrow data={0.1}{>}, arrow data={0.5}{>}, arrow data={0.97}{>}] (\db, 0) to (\db,1.25) to  [out=up, in=up, looseness=4] (1-\db,1.25) to (1-\db, 0);
 \draw[arrow data={0.07}{<}, arrow data={0.5}{<},arrow data={0.95}{<}] (-\db, 0) to (-\db, 1.25) to  [out=up, in=up, looseness=3] (1+\db, 1.25) to [out=down, in=up] (1+\db, 0);
 \draw (-1,0) to [in looseness=0.8] node[zxnode=\zxwhite, pos=0.405] {$X$} node[zxnode=\zxwhite, pos=0.65]{$X$} (1.75,2.) to  (1.75,3);
 \end{tz}
~ \superequalseq{eq:proofFrobenius1}~
 \begin{tz}[zx,xscale=-1, every to/.style={out=up, in=down}]
   \clip (-1.1,0) rectangle (1.6,3);
 \draw[arrow data={0.2}{>}, arrow data={0.83}{>}] (\db, 0) to [out=up, in=up, looseness=4] (1-\db,0);
 \draw[arrow data={0.06}{<},arrow data={0.44}{<},arrow data={0.59}{<}, arrow data={0.96}{<}] (-\db, 0) to (0.5-\dt, 2.4) to [out=up, in=up, looseness=4] (0.5+\dt, 2.4) to [out=down, in=up] (1+\db, 0); 
 \draw (-1,0) to (-1,0.25)to[in looseness=0.8] (1.5,3);
 \node[zxnode=\zxwhite] at (0.5,1.85) {$X$};
 \end{tz}
~ \superequalseq{eq:proofFrobenius1}~
  \begin{tz}[zx,xscale=-1, every to/.style={out=up, in=down}]
   \clip (-1.1,0) rectangle (1.6,3);
 \draw[arrow data={0.2}{>}, arrow data={0.83}{>}] (\db, 0) to [out=up, in=up, looseness=4] (1-\db,0);
 \draw[arrow data={0.11}{<}, arrow data={0.92}{<}] (-\db, 0) to [out=up, in=up, looseness=3] (1+\db, 0); 
 \draw (-1,0) to (-1,0.25)to[in looseness=0.8] (1.5,3);
 \end{tz}
 \hspace{0.3cm}\Rightarrow 
 \hspace{0.3cm}
 \begin{tz}[zx,xscale=-1,every to/.style={out=up, in=down},xscale=-1]
 \draw[string] (0,0) to node[zxnode=\zxwhite, pos=0.5] {$X^\dagger$}(2,3);
 \draw[string,arrow data={0.1}{<}, arrow data={0.93}{<}] (2.2,0) to (0.2,3);
 \draw[string,arrow data={0.08}{>}, arrow data={0.9}{>}] (1.8,0) to (-0.2,3);
 \end{tz}
 \overset{\text{unitary}}{=}
  \begin{tz}[zx,xscale=-1,every to/.style={out=up, in=down},xscale=-1]
 \draw[string] (0,0) to node[zxnode=\zxwhite, pos=0.5] {$X^{-1}$}(2,3);
 \draw[string,arrow data={0.1}{<}, arrow data={0.93}{<}] (2.2,0) to (0.2,3);
 \draw[string,arrow data={0.08}{>}, arrow data={0.9}{>}] (1.8,0) to (-0.2,3);
 \end{tz}
 ~=~~
 \begin{tz}[zx,xscale=-1,xscale=0.6,yscale=1]
\draw(0.25,-0.5) to (0.25,0) to [out=up, in=-135] (1,1);
\draw [arrow data={0.88}{<}]  (1,1) to [out=135, in=right]  (-0.3, 1.7) to [out=left, in=up] (-1.25,1) to (-1.25,-0.5);
\draw(1.75,2.5) to (1.75,2) to [out=down, in=45] (1,1);
\draw[arrow data={0.88}{<}] (1,1) to [out= -45, in= left]   (2.3,0.3) to [out=right, in=down] (3.25,1) to (3.25,2.5);
\draw[arrow data={0.9}{>}] (1,1) to [out=up, in=up,looseness=1.25] (-2,1) to (-2,-0.5);
\draw[arrow data={0.9}{>}] (1,1) to [out=down, in=down,looseness=1.25] (4,1) to  (4,2.5);
\node [zxnode=\zxwhite] at (1,1) {$X$};
\end{tz} 
\end{calign}
It then follows straightforwardly from \eqref{eq:proofFrobenius1} and \eqref{eq:proofFrobenius2}, that the following linear map $x\in \End(H^*\otimes V_\Gamma \otimes H)$ is a dagger idempotent, i.e. it is self-adjoint and fulfils $x^2= x$:
 \[\frac{1}{n}~~~
\begin{tz}[zx,xscale=-0.6,yscale=1]
\draw(0.25,-0.5) to (0.25,0) to [out=up, in=-135] (1,1);
\draw [arrow data={0.23}{<},arrow data={0.88}{<}]  (1,1) to [out=135, in=right]  (-0.3, 1.7) to [out=left, in=up] (-1.25,1) to (-1.25,-0.5);
\draw(1.75,2.5) to (1.75,2) to [out=down, in=45] (1,1);
\draw[arrow data={0.23}{<},arrow data={0.9}{<}] (1,1) to [out= -45, in= left]   (2.3,0.3) to [out=right, in=down] (3.25,1) to (3.25,2.5);
\draw[arrow data={0.83}{>}] (1,1) to [out=up, in=down] (0.25,2.5);
\draw[arrow data={0.83}{>}] (1,1) to [out=down, in=up] (1.75,-0.5);
\node [zxnode=\zxwhite] at (1,1) {$X$};
\end{tz} 
 \]
 Splitting this idempotent (see Section~\ref{app:Frobenius}) produces an isometry $i:A\to H^*\otimes V_{\Gamma} \otimes H$ from some Hilbert space $A$ such that:
 \begin{calign}\label{eq:idempotentsplitting}
 \frac{1}{n}~~~
\begin{tz}[zx,xscale=-0.6,scale=4/3]
\draw(0.25,-0.5) to (0.25,0) to [out=up, in=-135] (1,1);
\draw [arrow data={0.23}{<},arrow data={0.88}{<}]  (1,1) to [out=135, in=right]  (-0.3, 1.7) to [out=left, in=up] (-1.25,1) to (-1.25,-0.5);
\draw(1.75,2.5) to (1.75,2) to [out=down, in=45] (1,1);
\draw[arrow data={0.23}{<},arrow data={0.9}{<}] (1,1) to [out= -45, in= left]   (2.3,0.3) to [out=right, in=down] (3.25,1) to (3.25,2.5);
\draw[arrow data={0.83}{>}] (1,1) to [out=up, in=down] (0.25,2.5);
\draw[arrow data={0.83}{>}] (1,1) to [out=down, in=up] (1.75,-0.5);
\node [zxnode=\zxwhite] at (1,1) {$X$};
\end{tz} 
~~~=~
\def\d{0.8}
\begin{tz}[zx,xscale=-1,scale=4/3,every to/.style={out=up, in=down}]
\draw[arrow data={0.45}{>}] (-\d,0) to (-0.25*\d, 1);
\draw (0,0) to + (0,1);
\draw[arrow data={0.45}{<}] (\d,0) to  (0.25*\d,1);
\draw[arrow data={0.55}{>}](-0.25*\d,2) to (-\d,3) ;
\draw (0, 2) to  (0,3) ;
\draw[arrow data={0.55}{<}]  (0.25*\d, 2) to (\d,3) ;
\draw (0, 1) to (0,2) ;
\node[zxnode=\zxwhite] at (0, 1) {$i^\dagger$};
\node[zxnode=\zxwhite] at (0, 2) {$i$};
\node[dimension, right] at (0,1.5) {$A$};
\end{tz}
&
\begin{tz}[zx,xscale=-1,scale=4/3,every to/.style={out=up, in=down}]
\draw[arrow data={0.5}{>}] (0,0.55) to (-0.25*\d, 1.5) to (0,2.45);
\draw[arrow data={0.5}{<}] (0,0.55) to (0.25*\d, 1.5) to (0,2.45);
\draw (0,0) to + (0,1);
\draw (0, 2) to  (0,3) ;
\draw (0, 1) to (0,2) ;
\node[zxnode=\zxwhite] at (0, 0.75) {$i$};
\node[zxnode=\zxwhite] at (0, 2.25) {$i^\dagger$};
\node[dimension, right] at (0,0) {$A$};
\node[dimension, right] at (0,3) {$A$};
\end{tz}
~~=~~
\begin{tz}[zx,xscale=-1,scale=4/3]
\clip (-0.1, -0.2) rectangle (0.5,3.2);
\draw (0,0) to (0,3);
\node[dimension,right] at (0,0) {$A$};
\end{tz}
 \end{calign}
 \ignore{Moreover, we remark that this splitting is unique up to unitaries $A'\to A'$.}
Note that this splitting is unique up to a unitary morphism. We now claim that the following linear map $P:H \otimes A  \to V_{\Gamma} \otimes H$ is biunitary:
 \[\begin{tz}[zx,xscale=-1,every to/.style={out=up, in=down},xscale=-0.8]
\draw [arrow data={0.2}{>},arrow data={0.8}{>}]  (0,0) to (2.25,3);
\draw (2.25,0) to node[zxnode=\zxwhite, pos=0.5] {$P$} (0,3);
\node[dimension, right] at (2.25,0) {$A$};
\node[dimension, left] at (0,3) {$V_{\Gamma}$};
\end{tz}
~~:=~\sqrt{n}~~~
\begin{tz}[zx,yscale=-1,every to/.style={out=up, in=down}]
\draw[arrow data={0.35}{<}] (1., 0) to (0,1.8);
\draw[arrow data ={0.8}{<}] (0,1.8) to [out=-120, in=right] (-0.85,0.5) to [out=left, in=down] (-1.5,1.5) to (-1.5,3);
\draw (0,0) to (0, 3);
\node[zxnode=\zxwhite] at (0,1.5) {$i$};
\end{tz}\]
The unitarity equation $PP^\dagger=\mathbbm{1}$ follows from:
\[n~~~
\begin{tz}[zx,every to/.style={out=up, in=down}]
\begin{scope}[yscale=-1,yshift=-4.5cm]
\draw[arrow data={0.35}{<}] (1., 0) to (0,1.8);
\draw (0,1.8) to [out=-120, in=right] (-0.85,0.5) to [out=left, in=down] (-1.5,1.5) to (-1.5,2.3);
\draw (0,0) to (0, 3);
\node[zxnode=\zxwhite] at (0,1.5) {$i$};
\end{scope}
\draw[string,arrow data={0.35}{>}] (1., 0) to (0,1.8);
\draw[string,arrow data ={1}{>}] (0,1.8) to [out=-120, in=right] (-0.85,0.5) to [out=left, in=down] (-1.5,1.5) to (-1.5,2.3);
\draw (0,0) to (0, 3);
\node[zxnode=\zxwhite] at (0,1.5) {$i^\dagger$};
\end{tz}
~\superequalseq{eq:idempotentsplitting}~
\begin{tz}[zx,xscale=0.6,scale=1.225,yscale=-1.225]
\draw(0.25,-0.5) to (0.25,0) to [out=up, in=-135] (1,1);
\draw[arrow data={0.4}{<}] (0.95,0.95) to [out=145, in=125,looseness=25] (1.05,1.05) ;
\draw(1.75,2.5) to (1.75,2) to [out=down, in=45] (1,1);
\draw[arrow data={0.23}{<},arrow data={0.9}{<}] (1,1) to [out= -45, in= left]   (2.3,0.3) to [out=right, in=down] (3.25,1) to (3.25,2.5);
\draw[arrow data={0.83}{>}] (1,1) to [out=down, in=up] (1.75,-0.5);
\node [zxnode=\zxwhite] at (1,1) {$X$};
\end{tz} 
~~\superequalseq{eq:proofFrobenius1}~~
\begin{tz}[zx,xscale=1]
\draw (0,0) to (0,4.5);
\draw[arrow data={0.5}{>}] (1.2,0) to (1.2,4.5);
\end{tz}
\]
The other equation $P^\dagger P = \mathbbm{1}$ follows from conjugating the three right-most wires of the following by $i$ and using~\eqref{eq:idempotentsplitting}:
\def\sclP{0.75}
\def\xoffP{-0.2}
\def\yoffP{-0.3}
\begin{equation}\label{eq:endskolemnoether}
\begin{tz}[zx,yscale=-1,xscale=0.6,scale=4/3*\sclP]
\draw(0.25,-0.5) to (0.25,0) to [out=up, in=-135] (1,1);
\draw [arrow data={0.88}{<}]  (1,1) to [out=135, in=right]  (-0.3, 1.7) to [out=left, in=up] (-1.25,1) to (-1.25,-0.5);
\draw (1.75,2) to [out=down, in=45] (1,1);
\draw[arrow data={1}{<}] (1,1) to [out= -45, in= left]   (2.3,0.3) to [out=right, in=down] (3.25,1) to (3.25,1.5);
\draw[arrow data={0.9}{>}] (1,1) to [out=up, in=up,looseness=0.95] (-2.75,1) to (-2.75, -0.5);
\draw[arrow data={0.7}{>}] (1,1) to [out=down, in=up] (1.75,-0.5);
\draw (1.75,2) to [out=up, in=-135] (2.5,3) ;
\draw (3.25,1.5) to [out=up, in=down] (2.5,3);
\draw[arrow data={0.9}{<}] (2.5,3) to [out=-45, in=left] (3.8, 2.3) to [out=right, in=down] (4.75, 3) to (4.75, 4.5);
\draw[] (2.5,3) to [out=45, in=down] (3.25, 4) to (3.25, 4.5);
\draw[arrow data={0.79}{>}] (2.5,3) to [out=up, in=down] (1.75, 4.5);
\draw[arrow data={0.83}{<}] (2.25,3) to [out=up, in=down] (0.25, 4.5);
\node [zxnode=\zxwhite] at (1,1) {$X$};
\node [zxnode=\zxwhite] at (2.5,3) {$X$};
\end{tz} 
~\superequalseq{eq:proofFrobenius1}~ 
\begin{tz}[zx,yscale=-1,xscale=0.6,scale=4/3*\sclP]
\draw[arrow data={0.13}{<}, arrow data={0.93}{<}] (-2.75, -0.5) to [out=up, in=down] (0.25,4.5);
\begin{scope}[xshift=\xoffP cm, yshift=\yoffP cm]
\draw (1.75,2) to [out=135, in=right]  (0.45, 2.7) to [out=left, in=up] (-0.5,2);
\draw (1.75,2) to [out= -45, in= left]   (3.05,1.3) to [out=right, in=down] (4,2);
\draw (1.75,2) to [out=up, in=down] (1,3.5);
\draw (1.75,2) to [out=down, in=up] (2.1,0.5);
\node [zxnode=\zxwhite] at (1.75,2) {$X$};
\end{scope}
\draw[string,arrow data={0.79}{<}] (-0.5+\xoffP,2+\yoffP) to [out=down, in=up] (-1.25, -0.5);
\draw (0.25,-0.5) to [out=up, in=-135] (1.75+\xoffP, 2+\yoffP);
\draw[arrow data={0.35}{>}] (2.1+\xoffP,0.5+\yoffP) to [out=down, in=up] (1.75,-0.5);
\draw[arrow data={0.68}{>},string] (1+\xoffP,3.5+\yoffP) to [out=up, in=down](1.75,4.5);
\draw[string]  (1.75+\xoffP, 2+\yoffP) to [out=45, in=down] (3.24,4.5);
\draw[string, arrow data={0.85}{<}] (4+\xoffP, 2+\yoffP) to [out=up, in=down]  (4.75, 4.5);
\end{tz} 
\end{equation}
Unitarity of the quarter-rotation follows analogously:
\[
\begin{tz}[zx,xscale=-0.6,yscale=1,scale=1]
\draw(0.25,-0.5) to (0.25,0) to [out=up, in=-135] (1,1);
\draw [arrow data={0.34}{>}]  (1,1) to [out=135, in=right]  (-0.3, 1.7) to [out=left, in=up] (-1.25,1) to (-1.25,-0.5);
\draw(1.75,2.5) to (1.75,2) to [out=down, in=45] (1,1);
\draw[arrow data={0.32}{<}] (1,1) to [out= -45, in= left]   (2.3,0.3) to [out=right, in=down] (3.25,1) to (3.25,2.5);
\node [zxnode=\zxwhite] at (1,1) {$P$};
\end{tz}
~=~\sqrt{n}~
\begin{tz}[zx,yscale=-1,every to/.style={out=up, in=down},xscale=-1]
\draw[arrow data={0.35}{>}] (1., 0) to (0,1.8);
\draw[arrow data ={0.8}{>}] (0,1.8) to [out=-120, in=right] (-0.85,0.5) to [out=left, in=down] (-1.5,1.5) to (-1.5,3);
\draw (0,0) to (0, 3);
\node[zxnode=\zxwhite] at (0,1.5) {$i$};
\end{tz}
\]
From now on we will use the short-hand notation for $P$ introduced in~\eqref{eq:shorthand}. Using the algebra on $V_{\Gamma}$, we define the following linear maps on $A$:
\def\yoff{-0.15}
\def\loff{-0.15}
\def\scl{0.75}
\begin{calign}\nonumber
\begin{tz}[zx,scale=\scl]
\draw (0,0) to (0,1) to [out=up, in=-135] (0.75,2) to (0.75,3.75);
\draw (1.5,0) to (1.5,1) to  [out=up, in=-45] (0.75,2);
\node[zxvertex=\zxblack, zxdown] at  (0.75,2){};
\node[dimension, right] at (1.5,0) {$A$};
\node[dimension, right] at (0,0) {$A$};
\node[dimension, right] at (0.75,3.75) {$A$};
\end{tz}
\!\!\!:=~\frac{1}{n}~
\begin{tz}[zx,scale=\scl]
\draw[string] (0,0) to node[front,zxvertex=\zxwhite, pos=0.69]{}(0,1) to [out=up, in=-135] (0.75,2) to node[front, zxvertex=\zxwhite, pos=0.5]{} (0.75,3.75);
\draw[string] (1.5,0) to (1.5,1) to [out=up, in=-45] node[front,zxvertex=\zxwhite, pos=0.0]{} (0.75, 2);
\draw[string,arrow data={0.01}{>},arrow data = {0.33}{>},arrow data={0.665}{>}] (-0.25+\loff, 0.75+\yoff) to [out=20, in=down] (2.25,2+\yoff) to [out=up, in=-20] (-0.25+\loff, 3.25+\yoff)  to [out=160 , in=200] (-0.25+\loff, 0.75+\yoff) ;
\node[zxvertex=\zxwhite,zxdown] at (0.75,2){};
\end{tz}
&
\begin{tz}[zx,scale=\scl]
\clip (0.45, -0.25) rectangle (1.55, 4);
\draw  (0.75,2) to (0.75,3.75);
\node[zxvertex=\zxblack] at  (0.75,2){};
\node[dimension, right] at (0.75,3.75) {$A$};
\end{tz}
\!:=~\frac{1}{n}~
\begin{tz}[zx,scale=\scl]
\clip (-1+\loff, 0) rectangle (2.45,3.75); 
\draw[string] (0.75,2) to node[front, zxvertex=\zxwhite, pos=0.5]{} (0.75,3.75);
\draw[string,arrow data={0.01}{>},arrow data = {0.33}{>},arrow data={0.665}{>}] (-0.25+\loff, 0.75+\yoff) to [out=20, in=down] (2.25,2+\yoff) to [out=up, in=-20] (-0.25+\loff, 3.25+\yoff)  to [out=160 , in=200] (-0.25+\loff, 0.75+\yoff) ;
\node[zxvertex=\zxwhite] at (0.75,2){};
\end{tz}
&
\begin{tz}[zx,scale=\scl, yscale=-1]
\draw (0,0) to (0,1) to [out=up, in=-135] (0.75,2) to (0.75,3.75);
\draw (1.5,0) to (1.5,1) to  [out=up, in=-45] (0.75,2);
\node[zxvertex=\zxblack, zxup] at  (0.75,2){};
\node[dimension, right] at (1.5,0) {$A$};
\node[dimension, right] at (0,0) {$A$};
\node[dimension, right] at (0.75,3.75) {$A$};
\end{tz}
\!\!\!:=~\frac{1}{n}~
\begin{tz}[zx,scale=\scl,yscale=-1]
\draw[string] (0,0) to node[front,zxvertex=\zxwhite, pos=0.69]{}(0,1) to [out=up, in=-135] (0.75,2) to node[front, zxvertex=\zxwhite, pos=0.5]{} (0.75,3.75);
\draw[string] (1.5,0) to (1.5,1) to [out=up, in=-45] node[front,zxvertex=\zxwhite, pos=0.0]{} (0.75, 2);
\draw[string,arrow data={0.0}{<},arrow data = {0.33}{<},arrow data={0.67}{<}] (-0.25+\loff, 0.75+\yoff) to [out=20, in=down] (2.25,2+\yoff) to [out=up, in=-20] (-0.25+\loff, 3.25+\yoff)  to [out=160 , in=200] (-0.25+\loff, 0.75+\yoff) ;
\node[zxvertex=\zxwhite,zxup] at (0.75,2){};
\end{tz}
&
\begin{tz}[zx,scale=\scl, yscale=-1]
\clip (0.45, -0.25) rectangle (1.55, 4);
\draw  (0.75,2) to (0.75,3.75);
\node[zxvertex=\zxblack] at  (0.75,2){};
\node[dimension, right] at (0.75,3.75) {$A$};
\end{tz}
\!:=~\frac{1}{n}~
\begin{tz}[zx,scale=\scl, yscale=-1]
\clip (-1+\loff, 0) rectangle (2.45,3.75); 
\draw[string] (0.75,2) to node[front, zxvertex=\zxwhite, pos=0.5]{} (0.75,3.75);
\draw[string,arrow data={0.01}{<},arrow data = {0.33}{<},arrow data={0.665}{<}] (-0.25+\loff, 0.75+\yoff) to [out=20, in=down] (2.25,2+\yoff) to [out=up, in=-20] (-0.25+\loff, 3.25+\yoff)  to [out=160 , in=200] (-0.25+\loff, 0.75+\yoff) ;
\node[zxvertex=\zxwhite] at (0.75,2){};
\end{tz}
\end{calign}
It follows from Proposition~\ref{prop:jumpinglemma} that these structures form a special dagger Frobenius algebra. In fact, they form a special \textit{symmetric} dagger Frobenius algebra, since we also have that
\def\scl{0.8}
\[
\begin{tz}[zx, scale=\scl]
\clip (-0.2, -0.5) rectangle (1.8, 3.4+\yoff);
\draw (0,-0.5) to (0,1) to [out=up, in=-135] (0.75,2);
\draw (1.5,-0.5) to (1.5,1) to  [out=up, in=-45] (0.75,2);
\node[zxvertex=\zxblack, zxdown] at  (0.75,2){};
\end{tz}
~=~\frac{1}{n}~\begin{tz}[zx,scale=\scl]
\clip (-1.2, -0.5) rectangle (2.4, 3.4+\yoff);
\draw[string] (0, -0.5) to (0,0) to node[front,zxvertex=\zxwhite, pos=0.69]{}(0,1) to [out=up, in=-135] (0.75,2);
\draw[string]  (1.5,-0.5) to (1.5,0) to (1.5,1) to [out=up, in=-45] node[front,zxvertex=\zxwhite, pos=0.0]{} (0.75, 2);
\draw[string,arrow data={0.01}{>},arrow data = {0.33}{>},arrow data={0.665}{>}] (-0.25+\loff, 0.75+\yoff) to [out=20, in=down] (2.25,2+\yoff) to [out=up, in=-20] (-0.25+\loff, 3.25+\yoff)  to [out=160 , in=200] (-0.25+\loff, 0.75+\yoff) ;
\node[zxvertex=\zxwhite] at (0.75,2){};
\end{tz}
~=~\frac{1}{n}~
\begin{tz}[zx,scale=\scl]
\clip (-1.2, -0.5) rectangle (2.4, 3.4+\yoff);
\draw[string] (1.5,-0.5) to [out=up, in=down]  (0,0.5) to node[front,zxvertex=\zxwhite, pos=0.4]{}(0,1) to [out=up, in=-25, in looseness=2.5] (0.75,2);
\draw[string] (0,-0.5) to [out=up, in=down] (1.5,0.5) to (1.5,1) to [out=up, in=-155, in looseness=2.5] node[front,zxvertex=\zxwhite, pos=0.0]{} (0.75, 2);
\draw[string,arrow data={0.01}{>},arrow data = {0.33}{>},arrow data={0.665}{>}] (-0.25+\loff, 0.75+\yoff) to [out=20, in=down] (2.25,2+\yoff) to [out=up, in=-20] (-0.25+\loff, 3.25+\yoff)  to [out=160 , in=200] (-0.25+\loff, 0.75+\yoff) ;
\node[zxvertex=\zxwhite] at (0.75,2){};
\end{tz}
~=~\frac{1}{n}~
\begin{tz}[zx,scale=\scl]
\clip (-1.2, -0.5) rectangle (2.4, 3.4+\yoff);
\draw[string] (1.5,-0.5) to [out=up, in=down]  (0,0.5) to node[front,zxvertex=\zxwhite, pos=0.4]{}(0,1) to [out=up, in=-135] (0.75,2);
\draw[string] (0,-0.5) to [out=up, in=down] (1.5,0.5) to (1.5,1) to [out=up, in=-45] node[front,zxvertex=\zxwhite, pos=0.0]{} (0.75, 2);
\draw[string,arrow data={0.01}{>},arrow data = {0.33}{>},arrow data={0.665}{>}] (-0.25+\loff, 0.75+\yoff) to [out=20, in=down] (2.25,2+\yoff) to [out=up, in=-20] (-0.25+\loff, 3.25+\yoff)  to [out=160 , in=200] (-0.25+\loff, 0.75+\yoff) ;
\node[zxvertex=\zxwhite] at (0.75,2){};
\end{tz}
~=~
\begin{tz}[zx, scale=\scl]
\clip (-0.2, -0.5) rectangle (1.8, 3.4+\yoff);
\draw (1.5,-0.5) to[out=up, in=down] (0,1) to [out=up, in=-135] (0.75,2);
\draw (0,-0.5) to [out=up, in=down] (1.5,1) to  [out=up, in=-45] (0.75,2);
\node[zxvertex=\zxblack, zxdown] at  (0.75,2){};
\end{tz}
\]
Here, the second equation is a direct consequence of the graphical calculus, moving the bottom right node all the way around the oriented loop to the left. The third equation is symmetry of the algebra on $V_{\Gamma}$.

We also define the following endomorphism on $A$, which is --- again due to Proposition~\ref{prop:jumpinglemma} --- an adjacency matrix of a quantum graph.
\[\begin{tz}[zx,scale=\scl, yscale=1]
\draw[string] (0.75,0) to (0.75,3.75);
\node[zxnode=\zxwhite] at (0.75, 1.875) {$\Gamma_X$};
\node[dimension, right] at (0.75,0) {$A$};
\node[dimension, right] at (0.75, 3.75) {$A$};
\end{tz}
~:=~\frac{1}{n}~
\begin{tz}[zx,scale=\scl, yscale=1]
\draw[string] (0.75,0) to node[front, zxvertex=\zxwhite, pos=0.225]{}node[front,zxvertex=\zxwhite, pos=0.765]{} (0.75,3.75);
\draw[string,arrow data={0.01}{>},arrow data = {0.33}{>},arrow data={0.665}{>}] (-0.25+\loff, 0.75+\yoff) to [out=20, in=down] (2.25,2+\yoff) to [out=up, in=-20] (-0.25+\loff, 3.25+\yoff)  to [out=160 , in=200] (-0.25+\loff, 0.75+\yoff) ;
\node[zxnode=\zxwhite] at (0.75, 1.875) {$\Gamma$};
\end{tz}
\]
It follows from Proposition~\ref{prop:quantumcondition} and Proposition~\ref{prop:jumpinglemma} that $P$ is a quantum graph isomorphism from $\Gamma_X$ to $\Gamma$ and it follows from~\eqref{eq:idempotentsplitting} that $X=P \circ \overline{P}$.
\end{proof}

\section{Frobenius monoids in classical subcategories}
\label{sec:frobmonclassical}

In Section~\ref{sec:classification}, we classified quantum and classical graphs which are quantum isomorphic to a given quantum or classical graph $\Gamma$ in terms of Morita equivalence classes of simple dagger Frobenius monoids in the monoidal category $\QAut(\Gamma)$. Although some of these categories $\QAut(\Gamma)$ have been studied in the framework of compact quantum groups~\cite{Banica2008,Banica2009}, the general classification of Morita equivalence classes of Frobenius monoids in such categories seems unfeasible using current techniques.

We therefore focus on the much more tractable classical subcategories $\Hilb_{\Aut(\Gamma)} \subseteq \QAut(\Gamma)$ (see Definition~\ref{def:classical}), where the Morita equivalence classes of Frobenius monoids are well known~\cite{Ostrik2003}. Although these Frobenius monoids are in a sense `classical', being sums of classical automorphisms, we will see in Section~\ref{sec:bcsarkhipov} (and have already seen in the introduction) that they can still give rise to quantum but not classically isomorphic graphs.

Moreover, if a quantum graph $\Gamma$ has no quantum symmetries (Definition~\ref{def:noquantumsymmetries}) --- that is, if $\QAut(\Gamma) \cong \Hilb_{\Aut(\Gamma)}$ --- we are able to completely classify quantum graphs quantum isomorphic to $\Gamma$ in terms of straightforward group theoretical properties of the automorphism group of $\Gamma$.

\subsection{Quantum isomorphic quantum graphs from groups}\label{sec:ClassicalFrobPart1}
We recall the well known Morita classification of special Frobenius monoids on graded vector spaces. 
\begin{proposition}[{\cite[Example 2.1]{Ostrik2003}}] \label{prop:subgroupsclassification}Let $G$ be a finite group. Up to Morita equivalence, indecomposable,\footnote{A Frobenius monoid is \emph{indecomposable} if it is not a direct sum of non-trivial Frobenius monoids. We observe that all simple dagger Frobenius monoids are indecomposable.} special dagger Frobenius monoids in $\Hilb_G$ correspond to pairs $(L, \psi)$ where $L$ is a subgroup of $G$ and $\psi:L\times L \to U(1)$ is a 2-cocycle up to the equivalence relation:
\begin{equation} \label{eq:equivalence} \begin{split}
(L, \psi) \sim (L', \psi') \hspace{0.1cm}\Leftrightarrow\hspace{0.1cm} L' = gLg^{-1}&\text{ and }\psi'\text{ is cohomologous to }\\
&\psi^g(x,y):=\psi(gxg^{-1}, gyg^{-1})\text{ for some }g\in G
\end{split}
\end{equation}
\end{proposition}
\begin{proof}
Morita equivalence classes of indecomposable, special Frobenius monoids in $\Vect_G$ correspond to equivalence classes of semisimple indecomposable module categories over $\Vect_G$ whose classification in terms of pairs $(L, \psi)$ up to the equivalence relation~\eqref{eq:equivalence} is well-known~\cite[Example 2.1]{Ostrik2003}.
\ignore{
The classification of Morita equivalence classes of indecomposable symmetric, special Frobenius monoids in the category $\Vect_G$ by pairs $(L, \psi)$ up to the equivalence relation~\eqref{eq:equivalence} is well-known~\cite[Example 2.1]{Ostrik2003}. }%
\ignore{That this classification also applies in the dagger setting follows from the fact that any $G$-graded bimodule for two $G$-graded special Frobenius algebras can be endowed with a graded inner product giving it the structure of a dagger bimodule (Definition~\ref{def:daggerbimodule}). (For this, take any inner product compatible with the grading, pick a homogeneous orthonormal basis for both algebras, and perform an analogue of Weyl's unitarian trick.)}%
That this classification also applies in the dagger setting follows from the fact that a version of Weyl's unitarian trick can be used to endow any $G$-graded bimodule between two $G$-graded special dagger Frobenius algebras with a compatible inner product giving it the structure of a dagger bimodule (Definition~\ref{def:daggerbimodule}).  \ignore{ any $G$-graded bimodule for two $G$-graded special Frobenius algebras can be endowed with a graded inner product giving it the structure of a dagger bimodule (Definition~\ref{def:daggerbimodule}). (For this, take any inner product compatible }
\end{proof}
\noindent
The underlying algebra of the Frobenius monoid associated to $(L,\psi)$ is the twisted group algebra $\mathbb{C}L^\psi$ defined on the Hilbert space $\mathbb{C}L$ with orthonormal basis given by the group elements and algebra structure defined as:
\begin{calign}g\star_\psi h :=\frac{1}{\sqrt{|L|}}~ \psi(g,h) gh & e_\psi :=\sqrt{|L|}~\conj{\psi}(e,e)~ e
\end{calign}
Here again, the normalisation factors are chosen to make $\mathbb{C}L^\psi$ special (see Remark~\ref{rem:normalisation}).

The Frobenius monoid $(L,\psi)$ is simple in the sense of Definition~\ref{def:simpleFrob}, if the algebra $\mathbb{C}L^\psi$ is simple. Groups with 2-cocycles $\psi$ such that $\mathbb{C}L^\psi$ is simple have a long history and are known as \textit{groups of central type}, while the corresponding 2-cocycles are said to be \emph{non-degenerate} (see~\cite[Definition 7.12.21]{Etingof2015}). This leads to the following consequence of Corollary~\ref{cor:bigclassification}.

\begin{corollary}\label{cor:classification} Let $\Gamma$ be a quantum graph. Every subgroup of central type $(L, \psi)$ of $\Aut(\Gamma)$ induces a quantum graph $\Gamma_{L,\psi}$ and a quantum isomorphism $\Gamma_{L, \psi} \to \Gamma$. 
Moreover, if $\Gamma$ has no quantum symmetries, this gives rise to a bijective correspondence between the following structures:
\begin{itemize}
\item Isomorphism classes of quantum graphs $\Gamma'$ such that there exists a quantum isomorphism $\Gamma'\to\Gamma$.
\item  Subgroups of central type $(L, \psi)$ of $\Aut(\Gamma)$ up to the equivalence relation~\eqref{eq:equivalence}.
\end{itemize}
\end{corollary}
\begin{proof} The statement is a direct consequence of Proposition~\ref{prop:subgroupsclassification}; the classification of Morita equivalence classes of simple dagger Frobenius monoids in the category of $\Aut(\Gamma)$-graded Hilbert spaces.
\end{proof}
\begin{remark} Two Frobenius monoids in $\Hilb_{\Aut(\Gamma)}\subseteq \QAut(\Gamma)$ might be Morita equivalent in $\QAut(\Gamma)$ even if they are not in $\Hilb_{\Aut(\Gamma)}$. Therefore, the bijective correspondence of Corollary~\ref{cor:classification} holds only if $\Gamma$ has no quantum symmetries, that is, if $\QAut(\Gamma) \cong \Hilb_{\Aut(\Gamma)}$.
\end{remark}
\noindent
This makes the classification of quantum isomorphic quantum graphs quite concrete, particularly if one of the graphs has no quantum symmetries.

\begin{example}\label{exm:Cnquantumiso1}Let $C_n$ be the cycle graph with $n\geq 5$ vertices. It is known~\cite[Lemma 3.5]{Banica2005a} that $C_n$ has no quantum symmetries. Therefore, quantum isomorphic quantum graphs $\Gamma'$ are in correspondence with subgroups of central type of $\Aut(C_n) = D_n$.  All subgroups of $D_n$ are either cyclic or dihedral. For odd $n$, $D_n$ has no subgroup of central type. For even $n$, the only such subgroups are the abelian groups $D_2 \cong \mathbb{Z}_2\times \mathbb{Z}_2$, acting by 180-degree rotations and reflections on the cycle graph. Since there is only one nondegenerate second cohomology class of $\mathbb{Z}_2 \times \mathbb{Z}_2$, the equivalence relation~\eqref{eq:equivalence} reduces to conjugacy of subgroups. If $4$ does not divide $n$, there is only one conjugacy class of $\mathbb{Z}_2 \times \mathbb{Z}_2$ subgroups; if $4$ divides $n$, there are two such conjugacy classes, depending on whether the line of reflection is through opposing edges or through opposing vertices. We therefore conclude the following:
\begin{itemize}
\item For odd $n$, $C_n$ is only isomorphic to itself. 
\item For even $n$ not divisible by $4$, there is exactly one other quantum graph quantum isomorphic to $C_n$. 
\item For $n$ divisible by $4$, there are exactly two other quantum graphs quantum isomorphic to $C_n$.
\end{itemize}
We will show in Example~\ref{exm:Cnquantumiso2} that none of these quantum graphs are classical graphs.
\end{example}

\noindent
We now explicitly construct the simple dagger Frobenius monoid in $\QAut(\Gamma)$ corresponding to a subgroup of central type $(L,\psi)$ of the automorphism group $\Aut(\Gamma)$ of a quantum graph~$\Gamma$.

Note that every quantum isomorphism $X$ in $\Hilb_{\Aut(\Gamma)} \subseteq \QAut(\Gamma)$ is \textit{classical} in the sense of Definition~\ref{def:classical}. In particular, it is of the form~\eqref{eq:classicalquantumfunction} with some orthonormal basis $\left\{\ket{l}~|~l \in L\right\}$ of the underlying Hilbert space $H$ and permutations $\{l\in L\}$ where $L\subseteq \Aut(\Gamma)$ is some subset of the automorphism group of $\Gamma$. If $X$ is moreover a simple dagger Frobenius monoid, by Proposition~\ref{prop:subgroupsclassification} we can assume without loss of generality that $L$ is a subgroup of central type of the automorphism group and that the basis $\{\ket{l}\}$ is determined by a *-isomorphism of algebras $\mathbb{C}L^{\psi} \cong \End(H)$. The data defining such an isomorphism is known in the quantum information community as a nice unitary error basis.
\begin{definition}[{\cite{Klappenecker2003}}]\label{def:niceueb}  A \textit{nice unitary error basis} (nice UEB) for a group of central type $(L,\psi)$ is a family of unitary endomorphisms $\{U_a~|~a \in L\}$ on some Hilbert space $H$ with $|L|=\dim(H)^2$ and such that for all $a,b\in L$:\looseness=-1
\begin{calign} \label{eq:UEBcondition}\Tr(U_a^\dagger U_{b}) = \dim(H)~\delta_{a,b}& U_aU_{b} = \psi(a,b) U_{ab}
\end{calign}
The group $L$ is called the \textit{index group} of the nice UEB. From now on, and without loss of generality, we will always assume that $\psi(e,h) = 1 = \psi(h,e)$ and therefore that $U_e= \mathbbm{1}_H$.
\end{definition}
\noindent
A nice UEB induces a *-isomorphism of algebras $\mathbb{C}L^\psi \to \End(H),~a \mapsto \sqrt{\dim(H)}^{\hspace{1pt}-1} U_a$  (see Remark~\ref{rem:normalisation} for our normalisation of the endomorphism algebra) and every {$*$\-isomorphism} between $\mathbb{C}L^\psi$ and $\End(H)$ is of this form. 

We summarise this discussion in the following proposition.
\begin{prop} Let $\Gamma$ be a quantum graph. Every simple dagger Frobenius monoid in $\Hilb_{\Aut(\Gamma)}\subseteq \QAut(\Gamma)$ is Morita equivalent to a simple dagger Frobenius monoid $(H\otimes H^*, X_{L, \psi})$ for some Hilbert space $H$, where the underlying linear map $X_{L,\psi}: (H\otimes H^*) \otimes V_{\Gamma} \to V_{\Gamma} \otimes (H\otimes H^*)$ is defined as follows:
\def\d{0.5}
\def\h{0.6}
\def\l{3}
\def\hU{0.8}
\begin{equation}\label{eq:UEBquantumbijection}\begin{tz}[zx,xscale=-1, every to/.style={out=up, in=down},yscale=4/3]
\draw (0,0) to (2,3);
\draw[arrow data={0.1}{<}, arrow data={0.93}{<}] (2-\d, 0) to [out looseness=1.25, in looseness=0.75] (0.-\d,3);
\draw[arrow data={0.1}{>}, arrow data={0.93}{>}] (2+\d,0) to [out looseness=0.75, in looseness=1.25] (0+\d,3);
\node[zxnode=\zxwhite] at (1,1.5) {$X_{L,\psi}$};
\end{tz}
~=~
\frac{1}{\sqrt{|L|}}~\sum_{a\in L\subseteq \Aut(\Gamma)}~ 
\begin{tz}[zx,xscale=-1, every to/.style={out=up, in=down}]
\draw (0,0) to (2,4);
\draw[arrow data={0.08}{<},arrow data={0.99}{<}] (2-\d,0) to (2-\d,\h) to  [out=up, in=up,looseness=\l] (2+\d,\h) to  (2+\d,0);
\draw[arrow data={0.1}{>}, arrow data={0.95}{>}] (0-\d,4) to (0-\d,4-\h) to  [out=down, in=down,looseness=\l] (0+\d,4-\h) to  (0+\d,4);
\node[zxnode=\zxwhite] at (1,2) {$a$};
\node[zxnode=\zxwhite] at (2+\d, \hU) {$U_a^\dagger$};
\node[zxnode=\zxwhite] at (0+\d, 4-\hU) {$U_a$};
\node[dimension,left] at (0,0) {$V_\Gamma$};
\node[dimension, right] at (2,4) {$V_\Gamma$};
\end{tz}
\end{equation}
Here, $(L, \psi)$ is a subgroup of central type of $\Aut(\Gamma)$ and $\{U_a~|~a\in L\}$ is a corresponding nice UEB. The endomorphism $a:V_\Gamma\to V_\Gamma$ denotes the action of $a\in L \subseteq \Aut(\Gamma)$ on the quantum set of vertices $V_\Gamma$.
\end{prop}

\begin{remark} Different nice UEBs for the same subgroup of central type $(L,\psi)$ --- that is, different $*$-isomorphisms $\mathbb{C}L^\psi \cong \End(H)$ --- give rise to $*$-isomorphic, and in particular Morita equivalent, simple dagger Frobenius monoids $X_{L, \psi}$, and thus to isomorphic induced quantum graphs. Therefore, the particular choice of UEB does not play a role in the following classification.
\end{remark}

\begin{remark} The fact that $X_{L,\psi}$ is classical in the sense of Definition~\ref{def:classical} does not mean that its splitting --- the induced quantum isomorphism from some quantum graph $\Gamma_{L, \psi}$ to $\Gamma$ --- is classical. If this were not the case, we could never generate any non-isomorphic graph from Frobenius monoids in the classical subcategory. In fact, it is a direct consequence of Corollary~\ref{cor:bigclassification} that the splitting is only classical if $X_{L,\psi}$ is Morita trivial in $\QAut(\Gamma)$; if $\Gamma$ has no quantum symmetries, this only happens if $L$ is trivial.
\end{remark}
\begin{remark} For a classical graph $\Gamma$, the underlying projective permutation matrix of the quantum isomorphism~\eqref{eq:UEBquantumbijection} is the following, for $v,w\in V_{\Gamma}$:
\begin{equation} \left(X_{L, \psi} \right)_{v,w} := \frac{1}{\sqrt{|L|}}~ \sum_{a\in L \subseteq \Aut(\Gamma)}\! \delta_{a(v), w}~ P_{U_a} 
\end{equation}
Here $P_{U_a} : \End(H) \to \End(H)$ denotes the projector on the one-dimensional subspace spanned by $U_a\in \End(H)$.
\end{remark}
\subsection{Quantum isomorphic classical graphs from groups}\label{sec:ClassicalFrobPart2}
We now consider the conditions under which a central type subgroup of the automorphism group of a classical graph gives rise to a quantum isomorphic classical graph.
 In particular, we translate the classicality condition of Theorem~\ref{thm:commutativitycondition} into a condition on subgroups of central type.

We first discuss some properties of non-degenerate $2$-cocycles.
We denote the \emph{centraliser} of a group element $a\in L$ by $Z_a:= \left\{b \in L~|~ ab=ba\right\}$ and the \textit{commutator} of two group elements $a,b\in L$ by $[a,b]:= ab\inv{a}\inv{b}$.

For a 2-cocycle $\psi:L \otimes L\to U(1)$, we define the following function:
\begin{calign}\rho_\psi: L\otimes L \to U(1) & \rho_{\psi} (a,b) := \psi(a,b) \conj{\psi} (aba^{-1},a)
\end{calign}
If $L$ is abelian, it is well-known that $\rho_{\psi}$ is an \emph{alternating bicharacter} (that is, a homomorphism in both arguments such that $\rho_{\psi}(a,b) =\overline{\rho}_{\psi}(b,a))$. In the general setting, and for non-degenerate 2-cocycle $\psi$, the following still holds.

\begin{proposition}[{\cite[Exercise 7.12.22.v]{Etingof2015}}]\label{prop:nontriviality} Let $(L,\psi)$ be a group of central type and let $x\in L$. Then $\rho_{\psi} (x, - )|_{Z_x}: Z_x \to U(1)$ is a multiplicative character of the centralizer $Z_x$ and $\rho_{\psi}(x,-)|_{Z_x}$ is non-trivial for every $x\neq e_L$, that is:
\begin{equation}\label{eq:nontriviality} \rho_{\psi}(x,a) = 1 ~~~\forall a \in Z_x  \hspace{1cm} \Rightarrow\hspace{1cm} x = e_L 
\end{equation}
\end{proposition}

\noindent
If $(L,\psi)$ is a group of central type, we may therefore think of $\rho_\psi$ as a non-degenerate alternating form on $L$. In particular, we borrow the following definitions and terminology from the theory of symplectic forms on groups~\cite{BenDavid:2014}. 

\begin{definition} Let $(L,\psi)$ be a group of central type and let $S\subseteq L$ be a subset. The \emph{orthogonal complement} $S^\bot$ of $S$ is the following subset of $L$:
\begin{equation} S^\bot := \left\{ g\in L~|~\rho_{\psi} (g,a) = 1 ~~\forall a\in Z_g \cap S \right\}
\end{equation}
We say that a subset $S$ is \emph{coisotropic} if $S^\bot \subseteq S$.
\end{definition}

\noindent
Proposition~\ref{prop:nontriviality} leads to the following observation.

\begin{prop}\label{prop:centraltypeformula} For a group of central type $(L,\psi)$ and a subgroup $H \subseteq L$ we define:
\begin{equation}\Phi^{L,\psi}_H~:= \sum_{\substack{a,b\in H\\ [a,b] = e}} \rho_{\psi}(a,b)\end{equation} 
Then, $\Phi^{L,\psi}_H\in \mathbb{N}$ and $\Phi^{L,\psi}_H \leq |L|$ with equality if and only if $H$ is coisotropic. \ignore{If moreover $H$ is abelian, then $\Phi_H^{G,\psi} = |H| |H^\bot|$ and $\Phi_H^{G,\psi}=|G|$ if and only if $H$ is Lagrangian and $|G|=|H|^2$.}
\end{prop}
\begin{proof} 
Using orthogonality of characters of the group $Z_a\cap H$, we calculate:
\begin{equation}\label{eq:Phiformula}\Phi_H^{L,\psi} ~=~\sum_{a\in H} \sum_{b \in Z_a\cap H} \rho_{\psi}(a,b) ~=\!\!\!\!\!\!\!\sum_{\substack{a\in H\\ \rho_{\psi}(a,b) =1~\forall b \in Z_a\cap H}}\!\!\!\!\!\!\! |Z_a \cap H|
\end{equation}
Thus, $\Phi_{H}^{L,\psi}$ is a natural number.
Again using orthogonality of characters and equation~\eqref{eq:nontriviality} we note the following:
\[\sum_{\substack{a\in L, b \in H\\ [a,b]=e}} \rho_{\psi}(a,b) ~=~ \sum_{b\in H} \sum_{a\in Z_b } \rho_{\psi}(a,b) ~\superequalseq{eq:nontriviality}~ \sum_{b\in H} |Z_b| ~\delta_{b,e} ~=~ |L|
\]
On the other hand, we find:
\[\sum_{\substack{a\in L, b \in H\\ [a,b]=e}} \rho_{\psi}(a,b) ~=~\sum_{\substack{a,b\in H\\ [a,b]=e}} \rho_{\psi}(a,b) + \sum_{a\in L\setminus H} \sum_{b \in Z_a\cap H } \rho_{\psi}(a,b) ~=~ \Phi_H^{L,\psi} ~+\!\!\!\!\!\! \!\!\!\sum_{\substack{a\in L \setminus H\\ \rho_{\psi}(a,b) =1~\forall b\in Z_a\cap H}}\!\!\!\!\!\!\!\!\! |Z_a\cap H|
\]
Therefore, we obtain the following formula for $\Phi_{H}^{L,\psi}$:
\[\Phi_{H}^{L,\psi} ~=~ |L| ~- \!\!\!\!\!\!\!\!\! \!\!\!\sum_{\substack{a\in L \setminus H\\ \rho_{\psi}(a,b) =1~\forall b\in Z_a\cap H}}\!\!\!\!\!\!\!\!\! |Z_a\cap H|
\]
Proposition~\ref{prop:centraltypeformula} is an immediate consequence.\end{proof}

\noindent We now turn our attention back to graphs. For a vertex $v$ of a classical graph $\Gamma$ and a subgroup $L\subseteq \Aut(\Gamma)$, we denote the \textit{stabiliser subgroup} of $L$ by $\Stab_L(v):=\{h \in L~|~ h(v) = v \}$.

\begin{prop} \label{prop:formulacenter}Let $\Gamma$ be a classical graph and let $(L,\psi)$ be a subgroup of central type of $\Aut(\Gamma)$. Then, the dimension of the center of the algebra $V_{\Gamma_{L,\psi}}$ can be expressed as follows:
\begin{equation} \label{eq:dimcentergroup}\dim(Z(V_{\Gamma_{L,\psi}})) ~=~\frac{1}{|L|}~\sum_{v \in V_{\Gamma}}~\Phi_{\Stab_L(v)}^{L, \psi}
\end{equation}
\end{prop}
\begin{proof} Inserting the Frobenius algebra $X_{L,\psi}$~\eqref{eq:UEBquantumbijection} into the expression~\eqref{eq:centerquantumgraph} results in the following formula for the dimension of the center of the algebra $V_{\Gamma_{L, \psi}}$:
\[ \dim(Z(V_{\Gamma_{L,\psi}})) ~=~\frac{1}{|L|^{\frac{3}{2}}}\sum_{v\in V_{\Gamma}}\sum_{a,b\in \Stab_L(v)} \Tr(U_bU_a^\dagger U_b^\dagger U_a)
\]
It is a direct consequence of~\eqref{eq:UEBcondition} that the trace is only non-zero if $[a,b]=e$. In this case, $U_aU_b = \rho_\psi(a,b) U_b U_a$ and therefore $\Tr(U_b U_a^\dagger U_b^\dagger U_a) = \sqrt{|L|}~\rho_{\psi}(a,b)$. This proves the theorem:
\[
\dim(Z(V_{\Gamma_{L,\psi}})) ~=~\frac{1}{|L|} \sum_{v\in V_{\Gamma}}~\sum_{\substack{a,b\in \Stab_L(v)\\ [a,b]=e}} \rho_{\psi}(a,b) ~=~\frac{1}{|L|} \sum_{v\in V_{\Gamma}} ~\Phi_{\Stab_L(v)}^{L,\psi}
\qedhere\]
\end{proof}
\noindent Combining the formula of Proposition~\ref{prop:formulacenter} with Proposition~\ref{prop:centraltypeformula} leads to a necessary and sufficient condition for the quantum graph $\Gamma_{L,\psi}$ to be classical.

\begin{theorem}\label{thm:coisotropic} Let $\Gamma$ be a classical graph and let $(L, \psi)$ be a subgroup of central type of $\Aut(\Gamma)$. Then, $\Gamma_{L, \psi}$ is a classical graph if and only if all stabiliser subgroups are coisotropic; that is, for every vertex $v\in V_{\Gamma}$ the following holds:
\begin{equation}\label{eq:coisotropic} \Stab_L(v)^\bot := \left\{ a\in L~|~ \rho_{\psi}(a, b) = 1~\forall b \in Z_a \cap \Stab_L(v) \right\} ~\subseteq~ \Stab_L(v)
\end{equation}
\end{theorem}

\begin{proof}  The graph $\Gamma_{L, \psi}$ is classical if $V_{\Gamma_{L, \psi}}$ is commutative, that is if  $\dim(Z(V_{\Gamma_{L, \psi}})) = \dim(V_{\Gamma_{L, \psi}})~~ \superequals{\text{Prop.~\ref{prop:dimensionpreserved}}}~~ \dim(V_{\Gamma}) = |V_{\Gamma}|$. Using equation~\eqref{eq:dimcentergroup}, $\Gamma_{L, \psi}$ is therefore classical if and only if the following holds:
\begin{equation}\label{eq:proofcommutativegroup} \frac{1}{|L|} \sum_{v\in V_{\Gamma}}~ \Phi_{\Stab_L(v)}^{L, \psi} ~=~ |V_{\Gamma}|
\end{equation}
It follows from Proposition~\ref{prop:centraltypeformula} that $\Phi_{\Stab_L(v)}^{L, \psi} \leq |L|$. Thus, equation~\eqref{eq:proofcommutativegroup} holds if and only if $\Phi_{\Stab_L(v)}^{L, \psi} = |L|$ for every vertex $v\in V_{\Gamma}$ which in turn holds, again by Proposition~\ref{prop:centraltypeformula}, if and only if $\Stab_L(v)$ is coisotropic. 
\end{proof}

\noindent We now summarise our results on quantum isomorphic classical graphs obtained from simple dagger Frobenius monoids in the classical subcategory.
\begin{corollary}\label{cor:classgroup} Let $\Gamma$ be a classical graph. Then, every subgroup of central type $(L,\psi)$ of $\Aut(\Gamma)$ with coisotropic stabilisers induces a classical graph $\Gamma_{L,\psi}$ and a quantum isomorphism $\Gamma_{L, \psi}\to \Gamma$. Moreover, if $\Gamma$ has no quantum symmetries, this gives rise to a bijective correspondence between the following structures:
\begin{itemize}
\item Isomorphism classes of classical graphs $\Gamma'$ such that there exists a quantum isomorphism $\Gamma'\to \Gamma$.
\item Subgroups of central type $(L, \psi)$ of $\Aut(\Gamma)$ with coisotropic stabilisers up to the equivalence relation~\eqref{eq:equivalence}.
\end{itemize}
\end{corollary}
\begin{proof}Corollary~\ref{cor:classgroup} is a direct consequence of Theorem~\ref{thm:coisotropic} and Corollary~\ref{cor:superclassification}.
\end{proof}

\ignore{\begin{remark} \DRcomm{The authors have not yet found a graph without quantum symmetries that admits a quantum isomorphism to another non-isomorphic graph , or equivalently, a graph without quantum symmetries which has a central type subgroup of its automorphism group with coisotropic stabilizers. Nevertheless, there are graphs with quantum symmetries for which....}
\end{remark}}
\noindent We immediately make a simple observation based on the fact that trivial subgroups can never be coisotropic.
\begin{prop}\label{prop:notstabilised} Let $\Gamma$ be a classical graph, and let $(L, \psi)$ be a non-trivial subgroup of central type of $\Aut(\Gamma)$ such that $\Gamma_{L, \psi}$ is a classical graph. Then, every vertex is stabilised by some non-trivial element of $L$, that is $\Stab_L(v) \neq \{e\}$.
\end{prop}
\begin{proof} Note that $\{e\}^\bot = L$. Thus, if $v$ is a vertex of $\Gamma$ such that $\Stab_L(v) =\{e\}$, it follows from Theorem~\ref{thm:coisotropic} that $L = \Stab_L(v)^\bot \super{\eqref{eq:coisotropic}}{\subseteq} \Stab_L(v)=\{e\}$, and thus that $L=\{e\}$ contradicting non-triviality of $L$.
\end{proof} 

\begin{example}\label{exm:Cnquantumiso2} \label{exm:Cnquantumiso2}Let $C_n$ be the cycle graph with $n\geq 5$ vertices. We have seen in Example~\ref{exm:Cnquantumiso1} that for even $n$ there are either one or two quantum graphs $\Gamma'$ which are quantum isomorphic to $C_n$, corresponding to conjugacy classes of central type subgroups $\mathbb{Z}_2\times \mathbb{Z}_2 \subset D_n$. These subgroups act by 180 degree rotation and reflection along some axis through either opposite vertices or opposite edges of $C_n$. In both cases, there are vertices with trivial stabiliser. It therefore follows from Proposition~\ref{prop:notstabilised} that all quantum graphs quantum isomorphic to $C_n$ are non-classical.
\ignore{  Let $\{s,r\}$ be generators of $D_2 \subset D_n$; any vertex $v$ not satisfying $s \cdot v = v$ or $r s \cdot v = v$ will have trivial stabiliser. There are at most two vertices satisfying each equation, and there are $\geq 5$ vertices in the graph. Therefore, at least one vertex  has trivial stabiliser, and it follows from Proposition~\ref{prop:notstabilised} that all quantum graphs quantum isomorphic to $C_n$ are non-classical.}%
\end{example}

\ignore{
As we briefly discussed in the introduction, two non-isomorphic graphs G and H which nevertheless have a quantum isomorphism between them can be used to define a nonlocal game in which quantum pseudo-telepathy is exhibited [3]. The existence of a quantum graph isomorphism ensures a winning strategy provided that the players share sufficient quantum entanglement; however, the non-existence of a classical isomorphism means that there is no such perfect strategy without this shared entanglement.
This inspires the following definition:
Definition 7.1. A pair of non-isomorphic graphs (G,H) with a quantum graph isomorphism G →− H
between them will be called pseudo-telepathic.
An example of such a pair of pseudo-telepathic graphs was obtained in [3]. In contrast to the situation with quantum graph isomorphisms, the categories of quantum graph automorphisms are well understood, thanks to the work of Banica, Bichon and others on quantum graph automorphism groups [5, 6, 7, 9, 10? ? ].
In this section, we will show how we can classify pairs of pseudo-telepathic graphs in terms of certain algebraic structures in these quantum graph automorphism categories.
The following corollary is a direct consequence of Proposition 4.1 and Theorem 4.8.
}

\section{Quantum pseudo-telepathy}
Quantum pseudo-telepathy is a well-studied phenomenon in quantum information theory, where two non-communicating parties can use pre-shared entanglement to perform a task classically impossible without communication~\cite{Brassard2003,Brassard2005,Cleve2004}. Such tasks are usually formulated as games, where isolated players Alice and Bob are provided with inputs, and must return outputs satisfying some winning condition.

One such game is is the \emph{graph isomorphism game}~\cite{Atserias2016}, whose instances correspond to pairs of classical graphs $\Gamma$ and $\Gamma'$, and whose winning classical strategies are precisely graph isomorphisms $\Gamma' \to \Gamma$. Winning quantum strategies correspond to quantum isomorphisms.
\begin{proposition}[{\cite[Theorem 5.4]{Atserias2016}}]
Given classical graphs $\Gamma$ and $\Gamma'$, there is a winning quantum strategy for the graph isomorphism game if and only if there is a quantum isomorphism $(H,P):\Gamma' \to \Gamma$.
\end{proposition}
\noindent
Therefore, two non-isomorphic graphs with a quantum isomorphism between them exhibit pseudo-telepathic behaviour.

\begin{definition} A pair of non-isomorphic graphs $(\Gamma, \Gamma')$ will be called \emph{pseudo-telepathic} if there is a quantum isomorphism $\Gamma' \to\Gamma$.
\end{definition}
\noindent
We can therefore apply the results of Sections~\ref{sec:classification} and~\ref{sec:frobmonclassical} to obtain the following classification of pseudo-telepathic graph pairs $(\Gamma, \Gamma')$ in terms of structures in the monoidal category $\QAut(\Gamma)$.

\begin{corollary} \label{cor:fullpseudo}Let $\Gamma$ be a classical graph. There is a bijective correspondence between the following sets:
\begin{itemize}
\item Isomorphism classes of classical graphs $\Gamma'$ such that $(\Gamma, \Gamma')$ are pseudo-telepathic.
\item Non-trivial Morita equivalence classes of simple dagger Frobenius monoids in $\QAut(\Gamma)$ for which the expression~\eqref{eq:centerquantumgraph} evaluates to $|V_{\Gamma}|$.
\end{itemize}
\end{corollary}
\begin{proof}This is essentially the statement of Corollary~\ref{cor:superclassification} with the additional condition of non-triviality. Note that a simple dagger Frobenius monoid is Morita trivial if it is Morita equivalent to the monoidal unit $I$. On the other hand, under the correspondence of Corollary~\ref{cor:superclassification}, the monoidal unit of $\QAut(\Gamma)$ corresponds to the isomorphism class of $\Gamma$ itself. Excluding this trivial class leads to Corollary~\ref{cor:fullpseudo}.
\end{proof}

\noindent
Similarly, we can translate the statement of Corollary~\ref{cor:classgroup} into a statement about pseudo-telepathic graph pairs.

\begin{corollary}\label{cor:grouppseudo}
Let $\Gamma$ be a classical graph with no quantum symmetries. There is a bijective correspondence between the following sets:
\begin{itemize}
\item Isomorphism classes of classical graphs $\Gamma'$ such that the pair $(\Gamma, \Gamma')$ is pseudo-telepathic.
\item Non-trivial subgroups of central type $(L, \psi)$ of $\Aut(\Gamma)$ with coisotropic stabilisers up to the equivalence relation~\eqref{eq:equivalence}.
\end{itemize} 
\end{corollary}
\ignore{\noindent
We will now show how these results may be used to demonstrate the existence or non-existence of pseudotelepathic graph pairs.
}
\subsection{Ruling out pseudo-telepathy}
In this section, we demonstrate how Corollary~\ref{cor:fullpseudo} and Corollary~\ref{cor:grouppseudo} can be used to show that a graph $\Gamma$ cannot exhibit pseudo-telepathy. We begin by showing that almost all graphs are not part of a pseudo-telepathic graph pair. We recall a result of Lupini et al. showing that almost all graphs have trivial quantum automorphism group.
\begin{theorem}[{\cite[Theorem 3.14]{Lupini2017}}]\label{prop:randomgraph}
Let $G_n$ be the number of isomorphism classes of classical graphs with $n$ vertices and let $Q_n$ be the number of isomorphism classes of classical graphs with non-trivial quantum automorphism group. Then $Q_n/G_n$ goes to zero as $n$ goes to infinity. 
\end{theorem}
\noindent We combine this with our results to obtain the following corollary.
\begin{corollary}\label{cor:goestozero}
Let $G_n$ be the number of isomorphism classes of classical graphs with $n$ vertices and let $PT_n$ be the number of isomorphism classes of classical graphs which are part of a pseudo-telepathic pair. Then $PT_n/G_n$ goes to zero as $n$ goes to infinity.
\ignore{
Let $\Gamma$ be a random classical graph on $n$ vertices. The probability that there exists a quantum graph $\Gamma'$ such that $\Gamma$ and $\Gamma'$ are quantum isomorphic goes to zero as $n$ goes to infinity.}%
\end{corollary}
\begin{proof}
If $\Gamma$ has trivial quantum automorphism group, then it has no quantum symmetries and trival automorphism group $\Aut(\Gamma)$. There are therefore no non-trivial Morita equivalence classes of simple dagger Frobenius monoids in $\QAut(G)$; the result then follows from Corollary~\ref{cor:grouppseudo} and Theorem~\ref{prop:randomgraph}.
\end{proof}
\ignore{
\begin{corollary}
Let $\Gamma$ be a random classical graph on $n$ vertices. The probability that $\Gamma$ is part of a pseudo-telepathic graph pair goes to zero as $n$ goes to infinity.
\end{corollary}}
\noindent We now consider various graphs known to have no quantum symmetries. We recall the following result. 
\begin{theorem}[{\citetext{\citealp{Banica2007_2}, Section 7; \citealp{Schmidt2018}}}] The following is a complete list of all vertex-transitive graphs of order $\leq 11$ with no quantum symmetries.
\end{theorem}
\begin{center}
\begin{tabular}{l | l}
Graph & Automorphism group \\
\hline \hline 
$C_{11}$, $C_{11}(2)$, $C_{11}(3)$ & $D_{11}$ \\
Petersen & $S_5$ \\
$C_{10}$, $C_{10}(2)$, $C_{10}^{+}$, $\text{Pr}(C_5)$ & $D_{10}$ \\
Torus & $S_3 \wr \mathbb{Z}_2$ \\
$C_9$, $C_9(3)$ & $D_9$ \\
$C_8$, $C_8^+$ & $D_8$\\
$C_7$ & $D_7$ \\
$C_6$ & $D_6$ \\
$C_5$ & $D_5$ \\
$K_3$ & $S_3$ \\
$K_2$ & $\mathbb{Z}_2$
\end{tabular}
\end{center}
Here the graphs $C_{n}$, $C_n(m)$ and $C_{2n}^+=C_{2n}(n)$ are circulant graphs; $K_n$ are complete graphs; the Petersen graph is well-known; $\text{Pr}(C_5)$ is the graph $C_5 \times K_2$; and Torus is the graph $K_3 \times K_3$, where $\times$ is the direct product; see~\cite{Banica2007_2} for more detail.
\begin{theorem}\label{thm:vertextransitive}
Vertex-transitive graphs of order $\leq 11$ with no quantum symmetries cannot be part of a pseudo-telepathic graph pair. \end{theorem}
\begin{proof}
In this proof we make extensive use of the fact that the trivial subgroup of a group of central type cannot be coisotropic (see Proposition~\ref{prop:notstabilised}).

The automorphism groups of the complete graphs $K_2$ and $K_3$ have no nontrivial subgroups of central type, so by Corollary~\ref{cor:grouppseudo} cannot be part of a pseudo-telepathic graph pair.\footnote{In fact, it is well known that all quantum isomorphisms between graphs with fewer than four vertices are direct sums of classical isomorphisms~\cite{Wang1998}.}

The circulant graphs all have dihedral automorphism group, which acts on them as on any cycle graph. As with the cycle graph (Examples~\ref{exm:Cnquantumiso1} and~\ref{exm:Cnquantumiso2}), there are up to two conjugacy classes of nontrivial central type subgroups (all isomorphic to $\mathbb{Z}_2 \times \mathbb{Z}_2$), all of which have some trivial vertex stabilisers; so, by Corollary~\ref{cor:grouppseudo}, they cannot be part of a pseudo-telepathic graph pair. 

Similarly, $\text{Pr}(C_5)$ has trivial vertex stabilisers under the action of the  unique up-to-conjugacy central type subgroup $\mathbb{Z}_2 \times \mathbb{Z}_2$.

For the Petersen graph, all central type subgroups of $S_5$ are isomorphic to $\mathbb{Z}_2 \times \mathbb{Z}_2$; there are two conjugacy classes of these subgroups. However, each of these conjugacy classes have vertices with trivial stabiliser.

For the torus graph, $S_3 \wr \mathbb{Z}_2$ has three conjugacy classes of central type subgroups, two isomorphic to $\mathbb{Z}_2 \times \mathbb{Z}_2$ and one isomorphic to $\mathbb{Z}_3 \times \mathbb{Z}_3$. Again, it is straightforward to check that all three conjugacy classes have vertices with trivial stabiliser; the corresponding quantum graphs are therefore non-classical. 
\end{proof}
\begin{remark}
By Corollary~\ref{cor:bigclassification}, we also obtain a classification of quantum graphs which are quantum isomorphic to a classical graph with no quantum symmetries. We show how this works in the vertex-transitive case. The central type subgroups appearing in the proof of Theorem~\ref{thm:vertextransitive} are of the form $\mathbb{Z}_n\times \mathbb{Z}_n$ with $n=2,3$. There is only one cohomology class of non-degenerate $2$-cocycles on $\mathbb{Z}_2 \times \mathbb{Z}_2$, so for those graphs with only $\mathbb{Z}_2 \times \mathbb{Z}_2$ central type subgroups, quantum isomorphic quantum graphs are in bijective correspondence with conjugacy classes of these subgroups. This implies that the circulant graphs of odd order have no quantum isomorphic quantum graph, $\text{Pr}(C_5)$ and the circulant graphs of even order not divisible by $4$ have one quantum isomorphic quantum graph, and the Petersen graph and the circulant graphs of order divisible by $4$ have two quantum isomorphic quantum graphs.

We must be slightly more careful with the torus graph, since the central type subgroup $\mathbb{Z}_3 \times \mathbb{Z}_3$ has two cohomology classes $[\phi_1]$ and $[\phi_2]$ of nondegenerate $2$-cocyles. It is straightforward to check that, for a subgroup $L\cong \mathbb{Z}_3 \times \mathbb{Z}_3$ of $\Aut(\text{Torus}) \cong S_3 \wr \mathbb{Z}_2$, the pairs $(L, [\phi_1])$ and $(L, [\phi_2])$ are equivalent under the relation~\eqref{eq:equivalence}. The torus graph therefore has three quantum isomorphic quantum graphs, corresponding to the two conjugacy classes of $\mathbb{Z}_2 \times \mathbb{Z}_2$ subgroups, and the single conjugacy class of $\mathbb{Z}_3 \times \mathbb{Z}_3$ subgroups with either of the equivalent cohomology classes of 2-cocyles.

These quantum isomorphisms may have some interpretation in the theory of zero-error quantum communication~\cite{Stahlke2016}.
\end{remark}

\subsection{Binary constraint systems and Arkhipov's construction}\label{sec:bcsarkhipov}

In~\cite{Arkhipov2012}, Arkhipov describes a construction of a non-local game from a connected non-planar graph $Z$ and a specified vertex $l^*$, generalising the famous magic square and magic pentagram games~\cite{Mermin1990}.  In~\cite[Definition 4.4 and Theorem 4.5]{Lupini2017}, Lupini et al. translate this construction into a construction of a pseudo-telepathic graph pair $\left(X_0(Z), X(Z,l^*)\right)$. 

In this section, we show that the graph $X(Z,l^*)$ and the quantum isomorphism $X(Z,l^*)\to X_0(Z) $ always arise from subgroups of central type\footnote{In particular, all these graph pairs correspond to Frobenius monoids in the classical subcategory of one of the graphs.} of the automorphism group of the graph $X_0(Z)$, following the construction of Corollary~\ref{cor:classgroup}. Moreover, these subgroups can always be taken to be isomorphic to either $\mathbb{Z}_2^4$ or $\mathbb{Z}_2^6$.
The observations and constructions in this section generalise the binary magic square example from the introduction.

We first establish the following proposition, which allows us to recognise whether a graph $\Gamma'$ which is quantum isomorphic to another graph $\Gamma$ comes from a given central type subgroup of $\Aut(\Gamma)$.
\begin{prop} \label{prop:recognize}Let $\Gamma$ and $\Gamma'$ be classical graphs, let $(L, \psi)$ be a subgroup of central type of $\Aut(\Gamma)$ with coisotropic stabilisers and let $\{ U_a\in U(H) |~ a\in L\}$ be a corresponding nice unitary error basis. Then $\Gamma'$ is isomorphic to $\Gamma_{L, \psi}$ if and only if there exists a quantum isomorphism $(H,P): \Gamma'\to \Gamma$ such that the following holds, for all $a\in L \subseteq \Aut(\Gamma)$:
\begin{equation}\label{eq:recognise}\begin{tz}[zx, xscale=-1,scale=1.7,every to/.style={out=up, in=down},yscale=1.25]
\draw (0,0) to (2,2);
\draw[arrow data={0.95}{>},arrow data={0.085}{>}] (2,0) to node[zxnode=\zxwhite,pos=0.5] {$P$} node[zxnode=\zxwhite, pos=0.2] {$U_a^\dagger$}node[zxnode=\zxwhite, pos=0.8] {$U_a$} (0,2);
\node[dimension,left] at (2,0) {$H$};
\node[dimension,right] at (0,2) {$H$};
\node[dimension, right] at (0,0) {$V_{\Gamma'}$};
\node[dimension, left] at (2,2) {$V_{\Gamma}$};
\end{tz}
~~~=~~~
\begin{tz}[zx, xscale=-1,scale=1.7,every to/.style={out=up, in=down},yscale=1.25]
\draw (0,0) to node[zxnode=\zxwhite, pos=0.8] {$a^{\minus 1}$} (2,2);
\draw[string,arrow data={0.95}{>},arrow data={0.085}{>}] (2,0) to node[zxnode=\zxwhite,pos=0.5] {$P$}  (0,2);
\node[dimension,left] at (2,0) {$H$};
\node[dimension,right] at (0,2) {$H$};
\node[dimension, right] at (0,0) {$V_{\Gamma'}$};
\node[dimension, left] at (2,2) {$V_{\Gamma}$};
\end{tz}
\end{equation}
\end{prop}
\begin{proof} It follows from Proposition~\ref{prop:quantumbijectionFrobenius} that $\Gamma_{L, \psi}$ is isomorphic to $\Gamma'$ if and only if there exists a quantum isomorphism $(H,P): \Gamma' \to \Gamma$ such that $P\circ \overline{P} = X_{L, \psi}$:
\def\d{0.5}
\def\h{0.6}
\def\l{3}
\def\hU{0.8}
\begin{equation}\label{eq:proofMP}
\begin{tz}[zx, xscale=-1,every to/.style={out=up, in=down},yscale=1.143]
\draw (0,0) to (2.5,3.5);
\draw[arrow data={0.15}{<}, arrow data={0.8}{<}] (1.25,0) to node[zxnode=\zxwhite, pos=0.4] {$\conj{P}$} (0,3.5);
\draw[string,arrow data={0.2}{>}, arrow data={0.9}{>}] (2.5,0) to node[zxnode=\zxwhite, pos=0.58] {$P$} (1.25,3.5);
\node[dimension,left] at (1.25,0) {$H^*$};
\node[dimension, left] at (2.5,0) {$H$};
\end{tz}
~~~=~~~
\frac{1}{\sqrt{|L|}}~\sum_{a\in L\subseteq \Aut(\Gamma)}~ 
\begin{tz}[zx,xscale=-1, every to/.style={out=up, in=down}]
\draw (0,0) to (2,4);
\draw[arrow data={0.08}{<},arrow data={0.99}{<}] (2-\d,0) to (2-\d,\h) to  [out=up, in=up,looseness=\l] (2+\d,\h) to  (2+\d,0);
\draw[arrow data={0.1}{>}, arrow data={0.95}{>}] (0-\d,4) to (0-\d,4-\h) to  [out=down, in=down,looseness=\l] (0+\d,4-\h) to  (0+\d,4);
\node[zxnode=\zxwhite] at (1,2) {$a$};
\node[zxnode=\zxwhite] at (2+\d, \hU) {$U_a^\dagger$};
\node[zxnode=\zxwhite] at (0+\d, 4-\hU) {$U_a$};
\node[dimension,right] at (0,0) {$V_\Gamma$};
\node[dimension, left] at (2,4) {$V_\Gamma$};
\end{tz}
\end{equation}
Using the shorthand notation~\eqref{eq:shorthand} for the quantum isomorphism $P$, and~\eqref{eq:biunitarybraiding2}, this is equivalent to the following:
\[
\begin{tz}[zx,yscale=-1]
\draw[arrow data={0.15}{<}, arrow data={0.8}{<}] (1.25,0) to [out=up, in=down] (0, 4);
\draw[arrow data ={0.15}{<}, arrow data ={0.85}{<}] (1.25,4) to [out=down, in=down,looseness=3.5] (2.5,4);
\draw (0,0) to [out=up, in=down] node[zxvertex=\zxwhite, pos=0.31] {} (4,4);
\end{tz}
~~~=~~~
\begin{tz}[zx,yscale=-1]
\draw[arrow data={0.15}{<}, arrow data={0.8}{<}] (1.25,0) to [out=up, in=down] (0, 4);
\draw[arrow data ={0.1}{<}, arrow data = {0.5}{<}, arrow data ={0.9}{<}] (1.25,4) to[out=down, in=up] (1.9,2.0) to [out=down, in=down,looseness=3.5] (3.15,2.25) to [out=up, in=down] (2.5,4);
\draw (0,0) to [out=up, in=down] node[zxvertex=\zxwhite, pos=0.31] {} node[zxvertex=\zxwhite, pos=0.488]{} node[zxvertex=\zxwhite, pos=0.695] {}(4,4);
\end{tz}
~~~\superequalseq{eq:proofMP}~ ~~
\frac{1}{\sqrt{|L|}}~\sum_{b\in L} ~
\begin{tz}[zx,yscale=-1]
\draw[arrow data={0.11}{<}, arrow data={0.98}{<}] (3.75,0) to [out=up, in=down] node[zxnode=\zxwhite, pos=0.3] {$U_b$} (2.5, 4);
\draw[string,arrow data ={0.08}{>}, arrow data ={0.95}{>}] (0,4) to [out=down, in=down,looseness=3.5] node[zxnode=\zxwhite, pos=0.18] {$U_b^\dagger$}(1.25,4);
\draw[string] (0,0) to [out=up, in=down, in looseness=1.1] node[zxnode=\zxwhite, pos=0.47] {$b$} node[zxvertex=\zxwhite, pos=0.67] {} (4,4);
\end{tz}
\]
Contracting the first two bottom wires with $U_a$ for $a\in L$ and using~\eqref{eq:UEBcondition} completes the proof.
\end{proof}
\ignore{
\begin{corollary} The graphs $\Gamma_{\mathrm{BMS}}$ and $\Gamma_{\mathbb{Z}_2^4, \psi}$ are not classically isomorphic. 
\end{corollary}
\begin{proof} \DRcomm{Make sure we have assigned 1Z and 1X to the correct generators.}
Suppose $\Gamma_{\mathrm{BMS}}$ and $\Gamma_{\mathbb{Z}_2^4, \psi}$ are isomorphic. Then, according to Theorem~\ref{}, there is a quantum automorphism $(H,P): \Gamma_{\mathrm{BMS}} \to \Gamma_{\mathrm{BMS}}$ such that~\eqref{eq:} holds. In terms of the underlying matrix of projectors, this reads 
\[ U_a^\dagger P_{v,w} U_a = P_{a(v),w}
\]
Here $\{U_a~|~ a\in \mathbb{Z}_2^4\}$ is the UEB obtained from tensor products of Pauli matrices and $a(v)$ denotes the action of the group $\mathbb{Z}_2^4$ on the graph $\Gamma_{\mathrm{BMS}}$ as described in ...
\end{proof}}
\noindent
In terms of the underlying projective permutation matrix $P$, condition~\eqref{eq:recognise} can be stated as follows, for all $a\in L\subseteq\Aut(\Gamma), v'\in V_{\Gamma'}$ and $v\in V_{\Gamma}$: \begin{equation}\label{eq:recognisePPM} U_a P_{v',v} U_a^\dagger = P_{v',a(v)} \end{equation}

We give a brief summary of the construction of pseudo-telepathic graphs from binary constraint systems as developed in~\cite[Section 6]{Atserias2016}.

Let $\mathcal{F}$ be a linear binary constraint system (see~\cite[Section~6.1 and 6.2]{Atserias2016}) with binary variables $x_1,\ldots, x_m \in \{+1,-1\}$ and constraints $C_1,\ldots, C_p$, where each $C_l$ is an equation of the form $\prod_{x_i \in S_l} x_i = b_l$ for $S_l \subseteq \{ x_1,\ldots,x_n\}$ and $b_l \in \{+ 1,-1\}$.\footnote{Unlike~\cite{Atserias2016}, we write our constraint systems in multiplicative form.} A \emph{classical solution} is a solution of the constraint system with $x_i \in \{+ 1,-1\}$. A \emph{quantum solution} is a solution for which the $x_i$ are self-adjoint operators with eigenvalues $\pm 1$ acting on some finite-dimensional Hilbert space $H$, and such that all operators appearing in the same constraint commute. A linear binary constraint system which admits a quantum but not a classical solution will be called \emph{pseudo-telepathic}. 

The \emph{homogenisation} $\mathcal{F}_0$ of $\mathcal{F}$ is the constraint system in which we set the right hand side of every constraint equation to $+1$. For every linear binary constraint system $\mathcal{F}$, Atserias et al. construct a graph $\Gamma^{\mathcal{F}}$ whose vertices are pairs $(C_l, f)$ of a constraint equation $C_l$ of $\mathcal{F}$ together with a `local' classical solution $f : S_l \to \{+1,-1\}$ of this equation, and with an edge between $(C_l,f)$ and $(C_k, g)$ if and only if the local solutions $f:S_l \to \{+1,-1\}$ and $g: S_k \to \{+ 1,-1\}$ are inconsistent on $S_l \cap S_k$. They show that a constraint system $\mathcal{F}$ has a classical solution if and only if the graphs $\Gamma^{\mathcal{F}_0}$ and $\Gamma^{\mathcal{F}}$ are isomorphic, and that if $\mathcal{F}$ has a quantum solution then these graphs are quantum isomorphic. (See~\cite[Proof of Theorem 6.3]{Atserias2016} or the proof of Proposition~\ref{prop:binaryconstraint} below for the construction of the quantum isomorphism arising from a quantum solution.)

We now show that all pseudo-telepathic graph pairs arising from a binary constraint system possessing a quantum solution satisfying a certain pair of conditions can be obtained from central type subgroups.

\begin{prop} \label{prop:binaryconstraint}Let $\mathcal{F}$ be a linear binary constraint system and suppose that this system has a quantum solution $\{X_i \in \End(H) \}_{1\leq i \leq m}$, acting on some Hilbert space $H$, with the following two properties:
\begin{itemize}
\item If $A\in \End(H)$ is such that $AX_i = X_i A$ for all $1\leq i\leq m$, then $A \propto \mathbbm{1}_H$.
\item There is a group of central type $(L, \psi)$ and a corresponding nice unitary error basis $\{ U_a\in U(H)~|~ a\in L\}$ such that the following holds for all $a\in L$ and $1\leq i\leq m$:
\begin{equation} \label{eq:conjugating}U_a^\dagger X_i U_a = p_i ^a ~X_i\hspace{1cm}\text{where } p_i^a \in \{+1,-1\}
\end{equation}
\end{itemize}
Then, there is an embedding $L \hookrightarrow \Aut(\Gamma^{\mathcal{F}_0})$ and $\Gamma^{\mathcal{F}}$ is isomorphic to $\Gamma^{\mathcal{F}_0}_{L, \psi}$.
\end{prop}
\begin{proof}  We first note that for each $a\in L$, $\{p_i^a\}_{1\leq i \leq m}$ forms a (global) classical solution of the homogenous constraint system $\mathcal{F}_0$ and thus gives rise to `local' assignments which we denote by $[p^a]^l: S_l \to \{ +1,-1\}$. This in turn gives rise to an automorphism $p^a$ of $\Gamma^{\mathcal{F}_0}$, mapping a vertex $(C_l, f)$ to the vertex $(C_l,[p^a]^l \cdot f)$ where $[p^a]^l \cdot f : S_l \to \{ + 1,-1\}$ denotes the pointwise multiplication of the assignments $[p^a]^l, f: S_l \to \{ + 1,-1\}$. This results in a group homomorphism $L \to \Aut(\Gamma^{\mathcal{F}_0}),~a\mapsto p^a$, which is injective, since if $p^a = p^b $ it follows that $p_i^a = p_i ^b$ for all $1\leq i \leq m$ and thus that $U_a^\dagger X_i U_a = U_b^\dagger X_i U_b$, or equivalently that $U_a U_b^\dagger$ commutes with each $x_i$. Thus, by the first assumption, $U_a U_b^\dagger \propto \mathbbm{1}_H$ and therefore $a=b$. Therefore, $a\mapsto p^a$ defines an embedding $L \hookrightarrow \Aut(\Gamma^{\mathcal{F}_0})$.

We now show that $\Gamma^\mathcal{F}$ is isomorphic to $\Gamma^{\mathcal{F}_0}_{L, \psi}$. In the proof of~\cite[Theorem 6.3]{Atserias2016}, from a quantum solution $\{X_i \in \End(H)\}_{1\leq i \leq m}$ a quantum isomorphism $(H,P): \Gamma^{\mathcal{F}} \to \Gamma^{\mathcal{F}_0}$ is constructed as follows. Given a vertex $(C_l ,f)$ of $\Gamma^{\mathcal{F}}$, define the projector $Q_{(C_l,f)}$ on $H$ as the projector onto the joint eigenspace of the commuting operators $\{X_i~|~ x_i \in S_l \}$ with respective eigenvalues determined by $f:S_l \to \{ +1,-1\}$. The quantum isomorphism $(H, P):  \Gamma^{\mathcal{F}} \to \Gamma^{\mathcal{F}_0} $ is then defined as the following projective permutation matrix, where $(C_k, f) \in V_{\Gamma^{\mathcal{F}}}$ and $(C_l, g) \in V_{\Gamma^{\mathcal{F}_0}}$:
\[ P_{(C_k, f),(C_l, g)} := \delta_{k,l} ~Q_{(C_l, fg)}
\] If the given quantum solution fulfils the second condition of the theorem, then the projectors onto the joint eigenspaces fulfil the following equation:
\[U_a Q_{(C_l,g)} U_a^\dagger = Q_{(C_l, [p^a]^l\cdot g)}\]
Therefore, the following holds for the just defined projective permutation matrix $P$ and for all $a\in L$ and vertices $v\in V_{\Gamma^{\mathcal{F}}}$ and $w\in V_{\Gamma^{\mathcal{F}_0}}$:
\[ U_a P_{v,w} U_a^\dagger = P_{v,p^a(w)}
\]
This is precisely condition~\eqref{eq:recognisePPM}. It thus follows from Proposition~\ref{prop:recognize} that $\Gamma^{\mathcal{F}}$ and $\Gamma^{\mathcal{F}_0}_{L, \psi}$ are isomorphic.
\end{proof}

\begin{remark} 
The first paragraph of the proof of Propositon~\ref{prop:binaryconstraint} shows how automorphisms of the graph $\Gamma^{\mathcal{F}_0}$ arise from global classical solutions of the homogenous constraint system $\mathcal{F}_0$. This generalises how the automorphism subgroup $\mathbb{Z}_2^4$ of the graph $\Gamma$ in the example in the introduction arises from bit flip symmetries --- or equivalently from global classical solutions of the binary magic square constraint system.
\end{remark}
\begin{remark}
The global classical solutions~\eqref{eq:symmetrytransformations} of the magic square constraint system discussed in the introduction arise as in equation~\eqref{eq:conjugating} as the matrices of signs obtained from conjugating the entries of the following well-known quantum solution of the inhomogenous\footnote{In the inhomogenous magic square constraint system, all rows and columns multiply to $1$ except for the middle column which multiplies to $-1$.} magic square constraint system by $U_{1,0,0,0} = \sigma_X \otimes \mathbbm{1}_2,~ U_{0,1,0,0} = \sigma_Z \otimes \mathbbm{1}_2, ~U_{0,0,1,0} =\mathbbm{1}_2 \otimes \sigma_X $ and $U_{0,0,0,1} = \mathbbm{1}_2 \otimes \sigma_Z$, respectively: \begin{equation}\label{eq:quantumsolution}\left(\begin{array}{ccc}
\mathbbm{1}_2 \otimes \sigma_Z & \sigma_Z \otimes \sigma_Z & \sigma_Z \otimes \mathbbm{1}_2\\
\sigma_X\otimes \sigma_Z & \sigma_Y \otimes \sigma_Y & \sigma_Z \otimes \sigma_X \\ \sigma_X \otimes \mathbbm{1}_2 & \sigma_X \otimes \sigma_X & \mathbbm{1}_2 \otimes \sigma_X \end{array}\right)\end{equation} 
In particular, the inhomogenous magic square constraint system fulfils the conditions of Proposition~\ref{prop:binaryconstraint} which leads to the proof of Theorem~\ref{thm:arkhipov}.
\end{remark}

\noindent
We now show that all pseudo-telepathic graph pairs generated from Lupini et al.'s translation of Arkhipov's construction arise from a central type subgroup of the automorphism group of one of the graphs. Recall that, in the introduction, we used tensor products of the Pauli UEB to define the $2$-cocycles $\psi_{\mathrm{P}}$ on $\mathbb{Z}_2^2$~\eqref{eq:paulicocycle} and $\psi_{\mathrm{P}^2}$ on $\mathbb{Z}_2^4$~\eqref{eq:Pauliproduct}. We define the $2$-cocycle $\psi_{\mathrm{P}^3}$ on $\mathbb{Z}_2^6$ analogously.

\begin{theorem}\label{thm:arkhipov} Let $Z$ be a connected non-planar graph, let $l^*$ be a specified vertex of $Z$ and let $X_0(Z)$ and $X(Z,l^*)$ be the induced pseudo-telepathic graphs~\cite[Definition 4.4]{Lupini2017}. Then, there is a subgroup of central type $\left(L , \psi\right)$ of $\Aut(X_0(Z))$,which is isomorphic to either $\left(\mathbb{Z}_2^4, \psi_{\mathrm{P}^2}\right)$ or $\left(\mathbb{Z}_2^6, \psi_{\mathrm{P}^3}\right)$ such that $X(Z, l^*)$ is isomorphic to the graph $X_0(Z)_{L, \psi}$.
\end{theorem}
\begin{proof} If $Z$ is the bipartite complete graph $K_{3,3}$ or the complete graph $K_5$ with arbitrary specified vertex $l^*$, the associated pseudo-telepathic pair $(X_0(Z), X(Z, l^*))$ arise, respectively, from the well-known magic square and magic pentagram binary constraint system (see~\cite{Arkhipov2012} for more details). In both cases, there are quantum solutions (see \eqref{eq:quantumsolution} and \cite[Figure II.2]{Arkhipov2012}) consisting of two-fold and three-fold tensor products of Pauli matrices, fulfilling the conditions of Proposition~\ref{prop:binaryconstraint} for the subgroups $\mathbb{Z}_2^4$ and $\mathbb{Z}_2^6$ with 2-cocycle $\psi_{\mathrm{P}^2}$ and $\psi_{\mathrm{P}^3}$, respectively, with the corresponding Pauli tensor product UEBs.  
Thus, $X(K_{3,3}, l^*)$ is isomorphic to $X_0(K_{3,3})_{\mathbb{Z}_2^4, \psi_{\mathrm{P}^2}}$, while $X(K_5, l^*)$ is isomorphic to $X_0(K_5)_{\mathbb{Z}_2^6, \psi_{\mathrm{P}^3}}$.

For a general connected non-planar graph $Z$, Arkhipov chooses a topological minor isomorphic either to $K_{3,3}$ or $K_5$ (such a topological minor exists due to the Pontryagin-Kuratowski theorem) and constructs a quantum solution which contains precisely the operators from the quantum solution to $K_{3,3}$ or $K_{5}$, respectively, reducing the problem to either the magic square or magic pentagram. Thus, the obtained quantum solution again fulfils the conditions of Proposition~\ref{prop:binaryconstraint} with $(L.\psi)$ isomorphic to either $\left(\mathbb{Z}_2^4, \psi_{\mathrm{P}^2}\right)$ or $\left(\mathbb{Z}_2^6, \psi_{\mathrm{P}^3}\right)$. \end{proof}


\bibliographystyle{plainurl}
\bibliography{QSET}

\appendix


\section{Morita equivalence and dagger 2-categories}\label{app:2categorypictures}
In this appendix, we prove the correspondence between equivalence classes of certain objects and Morita equivalence classes of certain Frobenius monads in dagger $2$-categories in which dagger idempotents split. This is the main technical result needed for our classification of quantum isomorphic quantum graphs in Corollary~\ref{cor:bigclassification}.

The basic idea of this section can be summarised as follows.
A dualisable 1-morphism $S:B\to A$ in a dagger $2$-category $\mathbb{B}$ gives rise to a dagger Frobenius monoid $S\circ \conj{S}$ in $\mathbb{B}(A,A)$.
It can be shown that two such Frobenius monoids are $*$-isomorphic if and only if the underlying 1-morphisms $S:B\to A$ and $S':B'\to A$ are equivalent (in the sense that there is an equivalence $\epsilon:B\to B'$ such that $ S' \circ \epsilon \cong S$). If we only consider Frobenius monoids up to Morita equivalence, then we cannot recover the 1-morphism $S$ but we can still recover $B$, the source of $S$, up to equivalence. This is the content of Theorem~\ref{thm:maintechnical}.

Here, we use the graphical calculus of 2-categories; objects are depicted as shaded regions, 1-morphisms $f:A\to B$ are depicted as wires bounded by $A$ on the right and $B$ on the left and 2-morphisms are depicted by vertices. Our diagrams should be read from bottom to top and from right to left to match the conventional right-to-left notation of function --- and 1-morphism --- composition. For an introduction to this calculus, see~\cite{Selinger2010,Marsden2014}. 

Recall that a 1-morphism $F:A\to B$ in a 2-category is an \textit{equivalence} if there is a 1-morphism $G:B\to A$ such that $F{\circ}G \cong 1_B$ and $G\circ F\cong 1_A$. Similarly, we say that a 1-morphism in a dagger 2-category is a \textit{dagger equivalence} if these 2-isomorphisms are also unitary.

In the following, we depict the objects $A$ and $B$ as white and blue regions, respectively. We say that a 1-morphism $S:B\to A$ in a dagger 2-category has a \textit{dual} $\conj{S}:A\to B$ if there are $2$-morphisms $\epsilon:\conj{S}{\circ}S\To 1_B$ and $\eta:1_A\To S{\circ}\conj{S}$, depicted as follows:
\def\ps{\vphantom{\conj{S}}}
\def\scl{1.25}
\begin{calign}
\label{eq:cupsandcaps}
\begin{tz}[zx,scale=\scl,yscale=1.1]
\clip (-0.5, -0.4) rectangle (1.5,1.4);
\draw [ string, draw] (0,0) to [out=up, in=up, looseness=2] (1,0);
\draw [blueregion] (0,0) to (-0.5,0) to (-0.5,1) to (1.5,1) to (1.5,0) to (1,0) to [out=up, in=up, looseness=2] (0,0);
\node[dimension,below] at (1,0) {$\ps S$};
\node[dimension, below] at (0,0) {$\conj{S}$}; 
\end{tz}
&
\begin{tz}[zx,scale=\scl,yscale=1.1, yscale=-1]
\clip (-0.5, -0.4) rectangle (1.5,1.4);
\draw [blueregion, string, draw] (0,0) to [out=up, in=up, looseness=2] (1,0);
\node[dimension,above] at (1,0) {$\conj{S}$};
\node[dimension, above] at (0,0) {$\ps S$}; 
\end{tz}
&
\begin{tz}[zx,scale=\scl,yscale=1.1]
\clip (-0.5, -0.4) rectangle (1.5,1.4);
\draw [blueregion, string, draw] (0,0) to [out=up, in=up, looseness=2] (1,0);
\node[dimension,below] at (1,0) {$\conj{S}$};
\node[dimension, below] at (0,0) {$\ps S$}; 
\end{tz} 
&
\begin{tz}[zx,scale=\scl,yscale=1.1, yscale=-1]
\clip (-0.5, -0.4) rectangle (1.5,1.4);
\draw [string, draw] (0,0) to [out=up, in=up, looseness=2] (1,0);
\draw [blueregion] (0,0) to (-0.5,0) to (-0.5,1) to (1.5,1) to (1.5,0) to (1,0) to [out=up, in=up, looseness=2] (0,0);
\node[dimension,above] at (1,0) {$\ps S$};
\node[dimension, above] at (0,0) {$\conj{S}$}; 
\end{tz}
\\\nonumber
\epsilon:\conj{S}{\circ}S\To 1_B & \eta:1_A\To S{\circ}\conj{S} & \eta^\dagger:S{\circ}\conj{S}\To 1_A & \epsilon^\dagger:1_B\To \conj{S}{\circ}S
\end{calign}
These must satisfy the following `snake' equations:
\def\scl{1.4}
\begin{calign}
\label{eq:colouredcupscaps}
\begin{tz}[zx,scale=\scl,scale=1,xscale=-1]
\draw [blueregion] (0,0.25) to (0,1) to [out=up, in=up, looseness=2] (-0.5,1) to [out=down, in=down, looseness=2] (-1,1) to (-1,1.75) to (0.5,1.75) to (0.5,0.25);
\draw [string] (0,0.25) to (0,1) to [out=up, in=up, looseness=2] (-0.5,1) to [out=down, in=down, looseness=2] (-1,1) to (-1,1.75);
\end{tz}
~=~
\begin{tz}[zx,scale=\scl,scale=1,xscale=-1]
\draw [blueregion] (0,0.25) to (0,1.75) to (0.5,1.75) to (0.5,0.25);
\draw [string] (0,0.25) to (0,1.75);
\end{tz}
~=~
\begin{tz}[zx,scale=\scl,scale=1,scale=-1]
\draw [blueregion] (0,0.25) to (0,1) to [out=up, in=up, looseness=2] (-0.5,1) to [out=down, in=down, looseness=2] (-1,1) to (-1,1.75) to (0.5,1.75) to (0.5,0.25);
\draw [string] (0,0.25) to (0,1) to [out=up, in=up, looseness=2] (-0.5,1) to [out=down, in=down, looseness=2] (-1,1) to (-1,1.75);
\end{tz}
&
\begin{tz}[zx,scale=\scl,scale=1]
\draw [blueregion] (0,0.25) to (0,1) to [out=up, in=up, looseness=2] (-0.5,1) to [out=down, in=down, looseness=2] (-1,1) to (-1,1.75) to (0.5,1.75) to (0.5,0.25);
\draw [string] (0,0.25) to (0,1) to [out=up, in=up, looseness=2] (-0.5,1) to [out=down, in=down, looseness=2] (-1,1) to (-1,1.75);
\end{tz}
~=~
\begin{tz}[zx,scale=\scl,scale=1]
\draw [blueregion] (0,0.25) to (0,1.75) to (0.5,1.75) to (0.5,0.25);
\draw [string] (0,0.25) to (0,1.75);
\end{tz}
~=~
\begin{tz}[zx,scale=\scl,scale=1,yscale=-1]
\draw [blueregion] (0,0.25) to (0,1) to [out=up, in=up, looseness=2] (-0.5,1) to [out=down, in=down, looseness=2] (-1,1) to (-1,1.75) to (0.5,1.75) to (0.5,0.25);
\draw [string] (0,0.25) to (0,1) to [out=up, in=up, looseness=2] (-0.5,1) to [out=down, in=down, looseness=2] (-1,1) to (-1,1.75);
\end{tz}
\end{calign}
We say that a 1-morphism $S:B\to A$ is \textit{special} if it has a dual $\conj{S}:A\to B$ such that, in addition, the following holds:
\begin{equation}\label{eq:specialdagger}
\def\scl{1.25}
\begin{tz}[zx,scale=\scl,yscale=1., yscale=-1]
\draw[blueregion] (0.,-1) rectangle (2,1);
\draw [fill=white, string, draw]  (1,0) circle (0.5); 
\end{tz}
~~=~~
\begin{tz}[zx,scale=\scl,yscale=1., yscale=-1]
\draw[blueregion] (0.,-1) rectangle (2,1);
\end{tz}
\end{equation}
\noindent
If $S:B\to A$ is a special $1$-morphism, then $S{\circ}\conj{S} $ is a special dagger Frobenius monoid in $\mathbb{B}(A,A)$. Conversely, we say that a special dagger Frobenius monoid is \emph{split} if it is $*$-isomorphic to $S\circ\conj{S}$ for some special $1$-morphism.

We now state and prove the main technical result needed for the classification of quantum isomorphic graphs in Corollary~\ref{cor:bigclassification}. We believe the content and ideas of the following theorem to be well known; however, we could not find a similar statement and proof of appropriate generality in the literature.
\ignore{
\begin{definition} Let $\mathbb{B}$ be a $2$-category and let $A$ be an object of $\mathbb{B}$. We say that two $1$-morphisms $S:A\to B$ and $S':A\to B'$ are \emph{equivalent} if there is an equivalence $\epsilon: B \to B'$ and a $2$-isomorphism $\epsilon\circ S \cong S'$. 

In a dagger $2$-category, we say that $S$ and $S'$ are \emph{dagger equivalent} if $\epsilon$ is a dagger equivalence and the $2$-isomorphism $\epsilon \circ S \cong S'$ is unitary.
\end{definition}
\begin{theorem} Let $\mathbb{B}$ be a $2$-category in which idempotents split and let $S:A\to B$ and $S':A\to B'$ be two special $1$-morphisms. Then $S$ and $S'$ are equivalent (Definition~) if and only if the special dagger Frobenius algebras $\overline{S}\circ S$ and $\overline{S}' \circ S'$ are $*$-isomorphic. 
\end{theorem}}
\begin{theorem} \label{thm:maintechnical}Let $\mathbb{B}$ be a dagger $2$-category in which all dagger idempotents split and let $S:X\to A$ and $P:Y \to A$ be special $1$-morphisms. Then, the special dagger Frobenius monoids $S\circ \overline{S} $ and $P\circ \overline{P}$ are Morita equivalent if and only if $X$ is dagger equivalent to $Y$.
\end{theorem}
\begin{proof}
1. Suppose that $S{\circ} \conj{S}$ and $P{\circ}\conj{P}$ are Morita equivalent. Then we claim that $X$ (depicted as a blue region) and $Y$ (depicted as a red region) are dagger equivalent. The object $A$ is depicted as the white region and the invertible dagger bimodules ${}_{S{\circ}\conj{S}}M_{P{\circ}\conj{P}}$ and ${}_{P{\circ}\conj{P}}N_{S{\circ}\conj{S}}$ are depicted as follows:
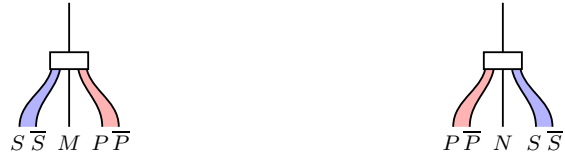
\def\vs{\vphantom{\conj{S}}}
\def\vsp{\vphantom{\conj{S}'}}
\begin{calign}\nonumber
\begin{tz}[zx,every to/.style={out=up, in=down}]
\draw (0,0) to (0,3);
\path[blueregion, draw, string] (-1.2, 0) to [in=down] (-0.4,1.6) to (-0.2,1.6) to [out=down, in=up] (-0.8,0);
\path[redregion, draw, string] (1.2, 0) to [in=down] (0.4,1.6) to (0.2,1.6) to [out=down, in=up] (0.8,0);
\node[bigbox] at (0,1.6){};
\node[dimension,below] at (0,0) {$\vs M$};
\node[dimension, below] at (1.25,0) {$\conj{P}$};
\node[dimension, below] at (0.75,0) {$\vs P$};
\node[dimension, below] at (-1.25,0) {$\vs S$};
\node[dimension, below] at (-0.75,0) {$\vs \conj{S}$};
\end{tz}
&
\begin{tz}[zx,every to/.style={out=up, in=down}]
\draw (0,0) to (0,3);
\path[redregion, draw, string] (-1.2, 0) to [in=down] (-0.4,1.6) to (-0.2,1.6) to [out=down, in=up] (-0.8,0);
\path[blueregion, draw, string] (1.2, 0) to [in=down] (0.4,1.6) to (0.2,1.6) to [out=down, in=up] (0.8,0);
\node[bigbox] at (0,1.6){};
\node[dimension,below] at (0,0) {$\vs N$};
\node[dimension, below] at (1.25,0) {$\conj{S}$};
\node[dimension, below] at (0.75,0) {$\vs S$};
\node[dimension, below] at (-1.25,0) {$\vs P$};
\node[dimension, below] at (-0.75,0) {$\conj{P}$};
\end{tz}
\end{calign}
Following~\eqref{eq:alldataMorita}, we have additional 2-morphisms fulfilling the following equations:

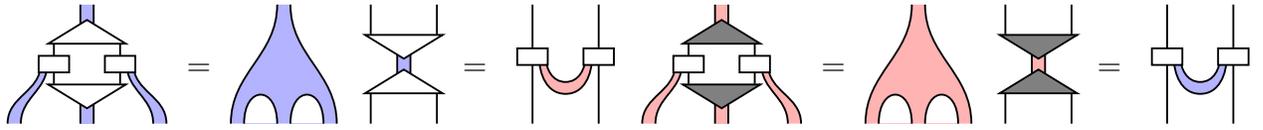
\def\scal{0.8}
\def\trianglescale{0.8}
\def\yscl{1.2}
\def\boxscale{0.8}
\begin{calign}\nonumber
\begin{tz}[zx,scale=\scal,yscale=\yscl]
\draw[blueregion,draw] (-0.2,0) to (-0.2,0.8) to (0.2,0.8) to (0.2,0);
\draw[blueregion,draw] (-0.2,3) to (-0.2,2.2) to (0.2,2.2) to (0.2,3);
\draw (-1,1) to (-1,2);
\draw (1,1) to (1,2);
\draw[blueregion, draw,string] (-2,0) to [out=up, in=down] (-1.2,1.5) to (-1.4,1.5) to [out=down, in=up] (-2.4,0);
\draw[blueregion, draw, string] (2,0) to [out=up, in=down] (1.2,1.5) to (1.4,1.5) to [out=down, in=up] (2.4,0);
\node[bigbox] at (-1,1.5){};
\node[bigbox] at (1,1.5){};
\node[triangleup=2] at (0,2){};
\node[triangledown=2] at (0,1){};
\end{tz}
~=~
\begin{tz}[zx,scale=\scal,yscale=\yscl]
\path[blueregion, even odd rule] 
(-1.6,0) to [out=up, in=down] (-0.2,3) to (0.2,3) to [out=down, in=up] (1.6,0)
(-1.2,0) to [out=up, in=up, looseness=2.5] (-0.2,0)
(0.2,0) to [out=up, in=up, looseness=2.5] (1.2,0);
\draw (-1.6,0) to [out=up, in=down] (-0.2,3) 
 (0.2,3) to [out=down, in=up] (1.6,0)
(-1.2,0) to [out=up, in=up, looseness=2.5] (-0.2,0)
(0.2,0) to [out=up, in=up, looseness=2.5] (1.2,0);
\end{tz}
&
\def\scl{3/4}
\def\xscl{4/3}
\begin{tz}[zx,scale=\scl*\scal,xscale=\xscl,yscale=\yscl]
\draw (0,0) to (0,1);
\draw (2,0) to (2,1);
\draw[blueregion, draw] (0.8,1.2) to (0.8,2.8) to (1.2,2.8) to (1.2,1.2);
\draw (0,3) to (0,4);
\draw (2,3) to (2,4);
\node[triangleup=2] at (1,1){};
\node[triangledown=2] at (1,3){};
\end{tz}
~=~~
\begin{tz}[zx,every to/.style={out=up, in=down},scale=\scl*\scal,xscale=\xscl,,yscale=\yscl]
\draw (0,0.) to (0,4);
\draw (2,0.) to (2,4);
\draw[redregion, draw] (0.2,2.25) to [out=down, in=left] (1,1.) to [out=right, in=down] (1.8, 2.25) to (1.6,2.25) to [out=down, in=right] (1,1.4) to [out=left, in=down] (0.4,2.25);
\node[bigbox] at (0,2.25){};
\node[bigbox] at (2,2.25){};
\end{tz}
&
\begin{tz}[zx,scale=\scal,yscale=\yscl]
\draw[redregion,draw] (-0.2,0) to (-0.2,0.8) to (0.2,0.8) to (0.2,0);
\draw[redregion,draw] (-0.2,3) to (-0.2,2.2) to (0.2,2.2) to (0.2,3);
\draw (-1,1) to (-1,2);
\draw (1,1) to (1,2);
\draw[redregion, draw,string] (-2,0) to [out=up, in=down] (-1.2,1.5) to (-1.4,1.5) to [out=down, in=up] (-2.4,0);
\draw[redregion, draw, string] (2,0) to [out=up, in=down] (1.2,1.5) to (1.4,1.5) to [out=down, in=up] (2.4,0);
\node[bigbox] at (-1,1.5){};
\node[bigbox] at (1,1.5){};
\node[triangleup=2,fill=\zxblack] at (0,2){};
\node[triangledown=2,fill=\zxblack] at (0,1){};
\end{tz}
~=~
\begin{tz}[zx,scale=\scal,yscale=\yscl]
\path[redregion, even odd rule] 
(-1.6,0) to [out=up, in=down] (-0.2,3) to (0.2,3) to [out=down, in=up] (1.6,0)
(-1.2,0) to [out=up, in=up, looseness=2.5] (-0.2,0)
(0.2,0) to [out=up, in=up, looseness=2.5] (1.2,0);
\draw (-1.6,0) to [out=up, in=down] (-0.2,3) 
 (0.2,3) to [out=down, in=up] (1.6,0)
(-1.2,0) to [out=up, in=up, looseness=2.5] (-0.2,0)
(0.2,0) to [out=up, in=up, looseness=2.5] (1.2,0);
\end{tz}
&
\def\scl{3/4}
\def\xscl{4/3}
\begin{tz}[zx,scale=\scl*\scal,xscale=\xscl,yscale=\yscl]
\draw (0,0) to (0,1);
\draw (2,0) to (2,1);
\draw[redregion, draw] (0.8,1.2) to (0.8,2.8) to (1.2,2.8) to (1.2,1.2);
\draw (0,3) to (0,4);
\draw (2,3) to (2,4);
\node[triangleup=2,fill=\zxblack] at (1,1){};
\node[triangledown=2,fill=\zxblack] at (1,3){};
\end{tz}
~=~~
\begin{tz}[zx,every to/.style={out=up, in=down},scale=\scl*\scal,xscale=\xscl,yscale=\yscl]
\draw (0,0.) to (0,4);
\draw (2,0.) to (2,4);
\draw[blueregion, draw] (0.2,2.25) to [out=down, in=left] (1,1.) to [out=right, in=down] (1.8, 2.25) to (1.6,2.25) to [out=down, in=right] (1,1.4) to [out=left, in=down] (0.4,2.25);
\node[bigbox] at (0,2.25){};
\node[bigbox] at (2,2.25){};
\end{tz}
\end{calign}
\def\trianglescale{1}%
\def\boxscale{1}%
\noindent
In particular, these equations imply the following:
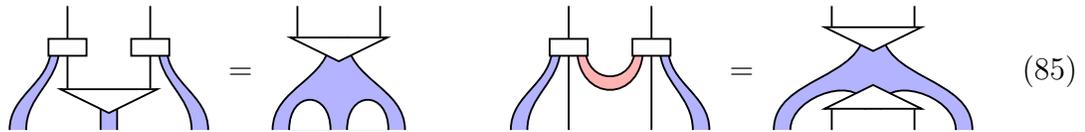
\begin{calign}\label{eq:implied}
\begin{tz}[zx]
\draw[blueregion,draw] (-0.2,0) to (-0.2,0.8) to (0.2,0.8) to (0.2,0);
\draw (-1,1) to (-1,3);
\draw (1,1) to (1,3);
\draw[blueregion, draw,string] (-2,0) to [out=up, in=down] (-1.2,2) to (-1.4,2) to [out=down, in=up] (-2.4,0);
\draw[blueregion, draw, string] (2,0) to [out=up, in=down] (1.2,2) to (1.4,2) to [out=down, in=up] (2.4,0);
\node[bigbox] at (-1,2){};
\node[bigbox] at (1,2){};
\node[triangledown=2] at (0,1){};
\end{tz}
~=~
\begin{tz}[zx]
\path[blueregion, even odd rule] 
(-1.6,0) to [out=up, in=down] (-0.2,2.25) to (0.2,2.25) to [out=down, in=up] (1.6,0)
(-1.2,0) to [out=up, in=up, looseness=2.5] (-0.2,0)
(0.2,0) to [out=up, in=up, looseness=2.5] (1.2,0);
\draw (-1.6,0) to [out=up, in=down] (-0.2,2.25) 
 (0.2,2.25) to [out=down, in=up] (1.6,0)
(-1.2,0) to [out=up, in=up, looseness=2.5] (-0.2,0)
(0.2,0) to [out=up, in=up, looseness=2.5] (1.2,0);
\draw (-1,2.25) to (-1,3);
\draw (1,2.25) to (1,3);
\node[triangledown=2] at (0, 2.25){};
\end{tz}
&
\begin{tz}[zx]
\draw (-1,0) to (-1,3);
\draw (1,0) to (1,3);
\draw[blueregion, draw,string] (-2,0) to [out=up, in=down] (-1.2,2) to (-1.4,2) to [out=down, in=up] (-2.4,0);
\draw[blueregion, draw, string] (2,0) to [out=up, in=down] (1.2,2) to (1.4,2) to [out=down, in=up] (2.4,0);
\draw[redregion, draw,string] (-0.8, 2) to [out=down, in=down,looseness=2.2]  (0.8,2) to (0.6,2) to [out=down, in=down,looseness=2] (-0.6,2);
\node[bigbox] at (-1,2){};
\node[bigbox] at (1,2){};
\end{tz}
~=~
\begin{tz}[zx]
\draw (-1,0) to (-1,0.5)
(1,0) to (1,0.5) 
(-1,2.5) to (-1,3)
(1,2.5) to (1,3);
\path[blueregion] (2.4,0) to [out=up, in=down] (0.2,2.5) to (-0.2,2.5) to [out=down, in=up]  (-2.4,0) to (-2,0) to [out=up, in=up,looseness=1.2] (-0.2,0.5) to (0.2,0.5) to [out=up, in=up,looseness=1.2] (2,0) to (2.4,0);
\draw (2.4,0) to [out=up, in=down] (0.2,2.5) 
(-0.2,2.5) to [out=down, in=up]  (-2.4,0) 
 (-2,0) to [out=up, in=up,looseness=1.2] (-0.2,0.5)
  (0.2,0.5) to [out=up, in=up,looseness=1.2] (2,0) ;
\node[triangleup=2] at (0,0.5){};
\node[triangledown=2] at (0,2.5){};
\end{tz}
\end{calign}
Since ${}_{S{\circ}\conj{S}}M_{P{\circ}\conj{P}}$ and ${}_{P{\circ}\conj{P}}N_{S{\circ}\conj{S}}$ are dagger bimodules, it follows that the following two 2-morphisms are dagger idempotent:
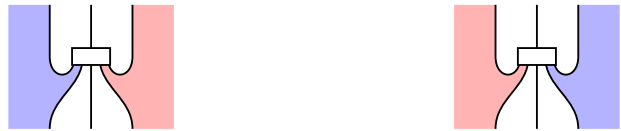
\begin{calign}\nonumber
\def\h{1.75}
\begin{tz}[zx]
\draw (0,0) to (0,3);
\path[redregion] (1,0) to [out=up, in=down] (0.2,\h) to (0.4,\h) to [out=down, in=down, looseness=2.5] (1,\h) to (1,3) to (2,3) to (2,0);
\draw (1,0) to [out=up, in=down] (0.2,\h) to (0.4,\h) to [out=down, in=down, looseness=2.5] (1,\h) to (1,3);
\path[blueregion] (-1,0) to [out=up, in=down] (-0.2,\h) to (-0.4,\h) to [out=down, in=down, looseness=2.5] (-1,\h) to (-1,3) to (-2,3) to (-2,0);
\draw (-1,0) to [out=up, in=down] (-0.2,\h) to (-0.4,\h) to [out=down, in=down, looseness=2.5] (-1,\h) to (-1,3);
\node[bigbox] at (0,\h){};
\end{tz}
&
\def\h{1.75}
\begin{tz}[zx]
\draw (0,0) to (0,3);
\path[blueregion] (1,0) to [out=up, in=down] (0.2,\h) to (0.4,\h) to [out=down, in=down, looseness=2.5] (1,\h) to (1,3) to (2,3) to (2,0);
\draw (1,0) to [out=up, in=down] (0.2,\h) to (0.4,\h) to [out=down, in=down, looseness=2.5] (1,\h) to (1,3);
\path[redregion] (-1,0) to [out=up, in=down] (-0.2,\h) to (-0.4,\h) to [out=down, in=down, looseness=2.5] (-1,\h) to (-1,3) to (-2,3) to (-2,0);
\draw (-1,0) to [out=up, in=down] (-0.2,\h) to (-0.4,\h) to [out=down, in=down, looseness=2.5] (-1,\h) to (-1,3);
\node[bigbox] at (0,\h){};
\end{tz}
\end{calign}
Splitting these idempotents produces 2-morphisms (here depicted as circular white nodes) fulfilling the following equations:
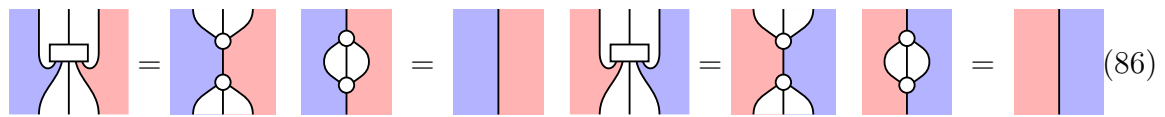
\def\xscl{0.85}
\def\scl{0.85}
\begin{calign}\label{eq:colorsplit}
\def\h{1.75}
\begin{tz}[zx,scale=\scl,xscale=\xscl]
\draw (0,0) to (0,3);
\path[redregion] (1,0) to [out=up, in=down] (0.2,\h) to (0.4,\h) to [out=down, in=down, looseness=2.5] (1,\h) to (1,3) to (2,3) to (2,0);
\draw (1,0) to [out=up, in=down] (0.2,\h) to (0.4,\h) to [out=down, in=down, looseness=2.5] (1,\h) to (1,3);
\path[blueregion] (-1,0) to [out=up, in=down] (-0.2,\h) to (-0.4,\h) to [out=down, in=down, looseness=2.5] (-1,\h) to (-1,3) to (-2,3) to (-2,0);
\draw (-1,0) to [out=up, in=down] (-0.2,\h) to (-0.4,\h) to [out=down, in=down, looseness=2.5] (-1,\h) to (-1,3);
\node[bigbox] at (0,\h){};
\end{tz}=
\begin{tz}[zx,scale=\scl,xscale=\xscl]
\clip (-1.75,0) rectangle (1.75,3);
\path[blueregion] (-2,0) rectangle (0,3);
\path[redregion] (2,0) rectangle (0,3);
\draw[fill=white] (-1,0) to [out=up, in=-135] (0,1) to [out=-45, in=up] (1,0);
\draw[fill=white] (-1,3) to [out=down, in=135] (0,2) to [out=45, in=down] (1,3);
\draw (0,0) to (0,3);
\node[zxvertex=\zxwhite, zxdown] at (0,1){};
\node[zxvertex=\zxwhite, zxup] at (0,2){};
\end{tz}
&
\begin{tz}[zx,scale=\scl,xscale=\xscl]
\clip (-1.5,0) rectangle (1.5,3);
\path[blueregion] (-2,0) rectangle (0,3);
\path[redregion] (2,0) rectangle (0,3);
\draw[fill=white] (0,0.75) to [out=135, in=down] (-0.75, 1.5) to [out=up, in=-135] (0,2.25) to [out=-45, in=up] (0.75, 1.5) to [out=down, in=45] (0,0.75);
\draw (0,0) to (0,3);
\node[zxvertex=\zxwhite, zxup] at (0,0.75){};
\node[zxvertex=\zxwhite, zxdown] at (0,2.25){};
\end{tz}
~=~
\begin{tz}[zx,scale=\scl,xscale=\xscl]
\clip (-1.5,0) rectangle (1.5,3);
\path[blueregion] (-2,0) rectangle (0,3);
\path[redregion] (2,0) rectangle (0,3);
\draw (0,0) to (0,3);
\end{tz}
&
\def\h{1.75}
\begin{tz}[zx,scale=\scl,xscale=\xscl]
\draw (0,0) to (0,3);
\path[blueregion] (1,0) to [out=up, in=down] (0.2,\h) to (0.4,\h) to [out=down, in=down, looseness=2.5] (1,\h) to (1,3) to (2,3) to (2,0);
\draw (1,0) to [out=up, in=down] (0.2,\h) to (0.4,\h) to [out=down, in=down, looseness=2.5] (1,\h) to (1,3);
\path[redregion] (-1,0) to [out=up, in=down] (-0.2,\h) to (-0.4,\h) to [out=down, in=down, looseness=2.5] (-1,\h) to (-1,3) to (-2,3) to (-2,0);
\draw (-1,0) to [out=up, in=down] (-0.2,\h) to (-0.4,\h) to [out=down, in=down, looseness=2.5] (-1,\h) to (-1,3);
\node[bigbox] at (0,\h){};
\end{tz}=
\begin{tz}[zx,scale=\scl,xscale=\xscl]
\clip (-1.75,0) rectangle (1.75,3);
\path[redregion] (-2,0) rectangle (0,3);
\path[blueregion] (2,0) rectangle (0,3);
\draw[fill=white] (-1,0) to [out=up, in=-135] (0,1) to [out=-45, in=up] (1,0);
\draw[fill=white] (-1,3) to [out=down, in=135] (0,2) to [out=45, in=down] (1,3);
\draw (0,0) to (0,3);
\node[zxvertex=\zxwhite, zxdown] at (0,1){};
\node[zxvertex=\zxwhite, zxup] at (0,2){};
\end{tz}
&
\begin{tz}[zx,scale=\scl,xscale=\xscl]
\clip (-1.5,0) rectangle (1.5,3);
\path[redregion] (-2,0) rectangle (0,3);
\path[blueregion] (2,0) rectangle (0,3);
\draw[fill=white] (0,0.75) to [out=135, in=down] (-0.75, 1.5) to [out=up, in=-135] (0,2.25) to [out=-45, in=up] (0.75, 1.5) to [out=down, in=45] (0,0.75);
\draw (0,0) to (0,3);
\node[zxvertex=\zxwhite, zxup] at (0,0.75){};
\node[zxvertex=\zxwhite, zxdown] at (0,2.25){};
\end{tz}
~=~
\begin{tz}[zx,scale=\scl,xscale=\xscl]
\clip (-1.5,0) rectangle (1.5,3);
\path[redregion] (-2,0) rectangle (0,3);
\path[blueregion] (2,0) rectangle (0,3);
\draw (0,0) to (0,3);
\end{tz}
&
\end{calign}
We claim that the resulting 1-morphisms $L:Y\to X$ and $R:X\to Y$ form a dagger equivalence. In fact, we claim that the following 2-morphisms are inverse to each other:
\begin{calign}\nonumber
\begin{tz}[zx]
\path[blueregion] (-2, -0.5) rectangle (2, 3);
\draw[fill=white, postaction={redregion,draw}] (-1,-0.5) to (-1,1) to (1,1) to (1,-0.5);
\draw[fill=white] (-1,0) to [out=45, in=135, looseness=1.5] (1,0) to [out=45, in=right,looseness=1.2] (0,2.5) to [out=left, in=135,looseness=1.2] (-1,0);
\draw (-1,0) to (-1,1);
\draw (1,0) to (1,1);
\draw[blueregion,draw] (-0.2,1.2) to (-0.2,1.5) to  [out=up, in=up, looseness=4] (0.2,1.5) to  (0.2,1.2);
\node[triangleup=2] at (0,1){};
\node[zxvertex=\zxwhite, zxup] at (-1,0){};
\node[zxvertex=\zxwhite, zxup] at (1,0){};
\end{tz}
&
\begin{tz}[zx,yscale=-1]
\path[blueregion] (-2, -0.5) rectangle (2, 3);
\draw[fill=white, postaction={redregion,draw}] (-1,-0.5) to (-1,1) to (1,1) to (1,-0.5);
\draw[fill=white] (-1,0) to [out=45, in=135, looseness=1.5] (1,0) to [out=45, in=right,looseness=1.2] (0,2.5) to [out=left, in=135,looseness=1.2] (-1,0);
\draw (-1,0) to (-1,1);
\draw (1,0) to (1,1);
\draw[blueregion,draw] (-0.2,1.2) to (-0.2,1.5) to  [out=up, in=up, looseness=4] (0.2,1.5) to  (0.2,1.2);
\node[triangledown=2] at (0,1){};
\node[zxvertex=\zxwhite, zxdown] at (-1,0){};
\node[zxvertex=\zxwhite, zxdown] at (1,0){};
\end{tz}\\\nonumber
\epsilon: L{\circ}R\To 1_X& \eta = \epsilon^\dagger:1_X\To L{\circ}R
\end{calign}
The equation $\epsilon \eta=1_{1_X}$ can be proven as follows:

\[\begin{tz}[zx]
\path[blueregion] (-2, -3.5) rectangle (2, 3.5);
\draw[fill=white, postaction={redregion,draw}] (-1,-1.5) to (-1,1.5) to (1,1.5) to (1,-1.5);
\draw[fill=white] (-1,0.5) to [out=45, in=135, looseness=1.5] (1,0.5) to [out=45, in=right,looseness=1.2] (0,3) to [out=left, in=135,looseness=1.2] (-1,0.5);
\draw[fill=white] (-1,-0.5) to [out=-45, in=-135, looseness=1.5] (1,-0.5) to [out=-45, in=right,looseness=1.2] (0,-3) to [out=left, in=-135,looseness=1.2] (-1,-0.5);
\draw (-1,0.5) to (-1,1.5);
\draw (1,0.5) to (1,1.5);
\draw (-1,-0.5) to (-1,-1.5);
\draw (1,-0.5) to (1,-1.5);
\draw[blueregion,draw] (-0.2,1.7) to (-0.2,2) to  [out=up, in=up, looseness=4] (0.2,2) to  (0.2,1.7);
\draw[blueregion,draw] (-0.2,-1.7) to (-0.2,-2) to  [out=down, in=down, looseness=4] (0.2,-2) to  (0.2,-1.7);
\node[triangleup=2] at (0,1.5){};
\node[zxvertex=\zxwhite, zxup] at (-1,0.5){};
\node[zxvertex=\zxwhite, zxup] at (1,0.5){};
\node[triangledown=2] at (0,-1.5){};
\node[zxvertex=\zxwhite, zxdown] at (-1,-0.5){};
\node[zxvertex=\zxwhite, zxdown] at (1,-0.5){};
\end{tz}
~~\superequalseq{eq:colorsplit}~~
\begin{tz}[zx]
\def\h{0.25}
\path[blueregion] (-2.5, -3.5) rectangle (2.5, 3.5);
\draw[fill=white] (1.2,\h+0.1) to  (1.,\h) to [out=-50, in=right,looseness=1.2] (0, -2.5) to [out=left, in=-130, looseness=1.2] (-1, \h) to (1.2, \h+0.1);
\draw[fill=white]  (1.4,\h) to [out=down, in=down,looseness=2.75] (1.9, \h) to [out=up, in=right] (0,2.5) to [out=left, in=up] (-1.9, \h) to [out=down, in=down,looseness=2.75] (-1.4, \h);
\draw (-1,-1) to (-1,1.);
\draw (1,-1) to (1,1.);
\draw[blueregion,draw] (-0.2,1.2) to (-0.2,1.5) to  [out=up, in=up, looseness=4] (0.2,1.5) to  (0.2,1.2);
\draw[blueregion,draw] (-0.2,-1.2) to (-0.2,-1.5) to  [out=down, in=down, looseness=4] (0.2,-1.5) to  (-0.8,-1.2);
\draw[redregion, draw] (-0.8,\h) to [out=down, in=down, looseness=2.2] (0.8, \h) to (0.6, \h) to [out=down, in=down, looseness=2] (-0.6, \h);
\node[bigbox] at (-1,\h){};
\node[bigbox] at (1,\h){};
\node[triangleup=2] at (0,1.){};
\node[triangledown=2] at (0,-1.){};
\end{tz}
~~~~\superequals{\eqref{eq:moritaidempotentsplit}\&\eqref{eq:implied}}~~~~
\ignore{
\begin{tz}[zx]
\def\h{0.25}
\path[blueregion] (-2.5, -3.5) rectangle (2.5, 3.5);
\draw[fill=white] (1.2,\h+0.1) to  (1.,\h) to [out=-50, in=right,looseness=1.2] (0, -2.5) to [out=left, in=-130, looseness=1.2] (-1, \h) to (1.2, \h+0.1);
\draw[fill=white]  (1.4,\h) to [out=down, in=down,looseness=2.75] (1.9, \h) to [out=up, in=right] (0,2.5) to [out=left, in=up] (-1.9, \h) to [out=down, in=down,looseness=2.75] (-1.4, \h);
\draw (-1,-1) to (-1,1.);
\draw (1,-1) to (1,1.);
\draw[blueregion,draw] (-0.2,1.2) to (-0.2,1.5) to  [out=up, in=up, looseness=4] (0.2,1.5) to  (0.2,1.2);
\draw[blueregion,draw] (-0.2,-1.2) to (-0.2,-1.5) to  [out=down, in=down, looseness=4] (0.2,-1.5) to  (0.2,-1.2);
\node[bigbox] at (-1,\h){};
\node[bigbox] at (1,\h){};
\node[triangleup=2] at (0,1.){};
\node[triangledown=2] at (0,-1.){};
\end{tz}
~~=~~}
\begin{tz}[zx]
\path[blueregion] (-2., -3.5) rectangle (2., 3.5);
\draw[fill=white] (0, -1.2) circle (0.8);
\draw[fill=white] (0, 1.2) circle (0.8);
\end{tz}
~~\superequalseq{eq:specialdagger}~~
\begin{tz}[zx]
\path[blueregion] (-2., -3.5) rectangle (2., 3.5);
\end{tz}
\]
We note that it follows from~\eqref{eq:colorsplit} that the converse $\eta \epsilon= 1_{L{\circ}R}$ is equivalent to the following equation:
\begin{calign}\label{eq:conversemoritaproof}
\begin{tz}[zx]
\path[blueregion] (-2, -4.25) rectangle (2, 4.25);
\draw[fill=white, postaction={redregion,draw}] (-1,-4.25) to (-1,-1.75) to (1,-1.75) to (1,-4.25);
\draw[fill=white, postaction={redregion,draw}] (-1,4.25) to (-1,1.75) to (1,1.75) to (1,4.25);
\draw[fill=white] (-1.5, 4.25) to [out=down, in=135] (-1, 3.25 ) to [out=45, in=down] (-0.5, 4.25);
\draw[fill=white] (-1.5, -4.25) to [out=up, in=-135] (-1, -3.25 ) to [out=-45, in=up] (-0.5,- 4.25);
\draw[fill=white] (1.5, 4.25) to [out=down, in=45] (1, 3.25 ) to [out=135, in=down] (0.5, 4.25);
\draw[fill=white] (1.5, -4.25) to [out=up, in=-45] (1, -3.25 ) to [out=-135, in=up] (0.5,- 4.25);
\draw (-1, 4.25) to +(0,-1);
\draw (-1, -4.25) to +(0,1);
\draw (1, 4.25) to +(0,-1);
\draw (1, -4.25) to +(0,1);
\draw[fill=white] (-1,-2.75) to [out=45, in=135, looseness=1.5] (1,-2.75) to [out=45, in=right,looseness=1.2] (0,-0.25) to [out=left, in=135,looseness=1.2] (-1,-2.75);
\draw[fill=white] (-1,2.75) to [out=-45, in=-135, looseness=1.5] (1,2.75) to [out=-45, in=right,looseness=1.2] (0,0.25) to [out=left, in=-135,looseness=1.2] (-1,2.75);
\draw (-1,-2.75) to (-1,-1.75);
\draw (1,-2.75) to (1,-1.75);
\draw (-1,2.75) to (-1,1.75);
\draw (1,2.75) to (1,1.75);
\draw[blueregion,draw] (-0.2,-1.55) to (-0.2,-1.25) to  [out=up, in=up, looseness=4] (0.2,-1.25) to  (0.2,-1.55);
\draw[blueregion,draw] (-0.2,1.55) to (-0.2,1.25) to  [out=down, in=down, looseness=4] (0.2,1.25) to  (0.2,1.55);
\node[triangleup=2] at (0,-1.75){};
\node[zxvertex=\zxwhite, zxup] at (-1,-2.75){};
\node[zxvertex=\zxwhite, zxup] at (1,-2.75){};
\node[triangledown=2] at (0,1.75){};
\node[zxvertex=\zxwhite, zxdown] at (-1,2.75){};
\node[zxvertex=\zxwhite, zxdown] at (1,2.75){};
\node[zxvertex=\zxwhite, zxup] at (-1,3.25){};
\node[zxvertex=\zxwhite, zxup] at (1,3.25){};
\node[zxvertex=\zxwhite, zxdown] at (-1,-3.25){};
\node[zxvertex=\zxwhite, zxdown] at (1,-3.25){};
\end{tz}
~~~=~~~
\begin{tz}[zx]
\path[blueregion] (-2, -4.25) rectangle (-1, 4.25);
\path[redregion] (-1, -4.25) rectangle (1, 4.25);
\path[blueregion] (1, -4.25) rectangle (2, 4.25);
\draw (-1,-4.25) to (-1, 4.25);
\draw (1,-4.25) to (1,4.25);
\draw[fill=white] (-1.5, 4.25) to [out=down, in=135] (-1, 3.25 ) to [out=45, in=down] (-0.5, 4.25);
\draw[fill=white] (-1.5, -4.25) to [out=up, in=-135] (-1, -3.25 ) to [out=-45, in=up] (-0.5,- 4.25);
\draw[fill=white] (1.5, 4.25) to [out=down, in=45] (1, 3.25 ) to [out=135, in=down] (0.5, 4.25);
\draw[fill=white] (1.5, -4.25) to [out=up, in=-45] (1, -3.25 ) to [out=-135, in=up] (0.5,- 4.25);
\draw (-1, 4.25) to +(0,-1);
\draw (-1, -4.25) to +(0,1);
\draw (1, 4.25) to +(0,-1);
\draw (1, -4.25) to +(0,1);
\node[zxvertex=\zxwhite, zxup] at (-1,3.25){};
\node[zxvertex=\zxwhite, zxup] at (1,3.25){};
\node[zxvertex=\zxwhite, zxdown] at (-1,-3.25){};
\node[zxvertex=\zxwhite, zxdown] at (1,-3.25){};
\end{tz}
~~~\superequalseq{eq:colorsplit}~~~
\def\h{1.75}
\begin{tz}[zx]
\clip (-2.25, -4.25) rectangle (2.25, 4.25);
\draw (-1,-4.25) to (-1,4.25);
\path[redregion]  (-0.5, -4.25) to  [out=up, in=down] (-0.8,\h-1.5) to (-0.6,\h-1.5) to [out=down, in=down, looseness=2.75] (-0.25,\h-1.5) to[out=up, in=down] (-0.5,4.25) to (0.5,4.25) to [out=down, in=up] (0.25,\h-1.5) to [out=down, in=down, looseness=2.75] (0.6,\h-1.5) to (0.8, \h-1.5) to [out=down, in=up] (0.5,-4.25); 
\draw (-0.5,-4.25) to  [out=up, in=down] (-0.8,\h-1.5) to (-0.6,\h-1.5) to [out=down, in=down, looseness=2.75] (-0.25,\h-1.5) to [out=up, in=down] (-0.5,4.25);
\path[blueregion] (-1.5,-4.25) to   [out=up, in=down] (-1.2,\h-1.5) to (-1.4,\h-1.5) to [out=down, in=down, looseness=2.75] (-1.75,\h-1.5) to [out=up, in=down] (-1.5,4.25) to (-3,4.25) to (-3,-4.25);
\draw (-1.5,-4.25) to  [out=up, in=down] (-1.2,\h-1.5) to (-1.4,\h-1.5) to [out=down, in=down, looseness=2.75] (-1.75,\h-1.5) to [out=up, in=down] (-1.5,4.25);
\draw (1,-4.25) to (1,4.25);
\draw (0.5,-4.25) to  [out=up, in=down] (0.8,\h-1.5) to (0.6,\h-1.5) to [out=down, in=down, looseness=2.75] (0.25,\h-1.5) to [out=up, in=down] (0.5,4.25);
\path[blueregion] (1.5,-4.25) to   [out=up, in=down] (1.2,\h-1.5) to (1.4,\h-1.5) to [out=down, in=down, looseness=2.75] (1.75,\h-1.5) to [out=up, in=down] (1.5,4.25) to (3,4.25) to (3,-4.25);
\draw (1.5,-4.25) to  [out=up, in=down] (1.2,\h-1.5) to (1.4,\h-1.5) to [out=down, in=down, looseness=2.75] (1.75,\h-1.5) to [out=up, in=down] (1.5,4.25);
\node[bigbox] at (-1,\h-1.5){};
\node[bigbox] at (1,\h-1.5){};
\end{tz}
\end{calign}
We can rewrite the involved 2-morphism as follows:
\[\begin{tz}[zx,yscale=-1]
\clip (-2, -1.5) rectangle (2,3);
\path[blueregion] (-2, -1.5) rectangle (2, 3);
\draw[fill=white, postaction={redregion,draw}] (-1,-1.5) to (-1,1) to (1,1) to (1,-1.5);
\draw[fill=white] (-1,0) to [out=45, in=135, looseness=1.5] (1,0) to [out=45, in=right,looseness=1.2] (0,2.5) to [out=left, in=135,looseness=1.2] (-1,0);
\draw[fill=white] (-1.5,-1.5) to [out=up, in=-135] (-1,-0.5) to [out=-45, in=up] (-0.5,-1.5);
\draw[fill=white](1.5,-1.5) to [out=up, in=-45] (1,-0.5) to [out=-135, in=up] (0.5,-1.5);
\draw (-1,0) to (-1,1);
\draw (1,0) to (1,1);
\draw[blueregion,draw] (-0.2,1.2) to (-0.2,1.5) to  [out=up, in=up, looseness=4] (0.2,1.5) to  (0.2,1.2);
\draw (1,-0.5) to + (0,-1);
\draw (-1,-0.5) to + (0,-1);
\node[triangledown=2] at (0,1){};
\node[zxvertex=\zxwhite, zxdown] at (-1,0){};
\node[zxvertex=\zxwhite, zxdown] at (1,0){};
\node[zxvertex=\zxwhite, zxup] at (-1,-0.5){};
\node[zxvertex=\zxwhite, zxup] at (1,-0.5){};
\end{tz}
~~\superequalseq{eq:colorsplit}~~
\begin{tz}[zx]
\def\h{0.25}
\clip (-2.5, -3 ) rectangle (2.5,1.5);
\path[blueregion] (-2.5, -3.5) rectangle (2.5, 3.5);
\draw[fill=white] (1.2,\h+0.1) to  (1.,\h) to [out=-50, in=right,looseness=1.2] (0, -2.5) to [out=left, in=-130, looseness=1.2] (-1, \h) to (1.2, \h+0.1);
\path[fill=white]  (1.4,\h) to [out=down, in=down,looseness=2.75] (1.9, \h) to (1.9, 1.5) to (-1.9, 1.5) to  (-1.9, \h) to [out=down, in=down,looseness=2.75] (-1.4, \h);
\draw  (1.4,\h) to [out=down, in=down,looseness=2.75] (1.9, \h) to (1.9, 1.5)  (-1.9, 1.5) to  (-1.9, \h) to [out=down, in=down,looseness=2.75] (-1.4, \h);
\draw (-1,-1) to (-1,1.5);
\draw (1,-1) to (1,1.5);
\draw[blueregion,draw] (-0.2,-1.2) to (-0.2,-1.5) to  [out=down, in=down, looseness=4] (0.2,-1.5) to  (-0.8,-1.2);
\draw[redregion] (-0.8,\h) to [out=down, in=down, looseness=1.8] (0.8, \h) to (0.6, \h) to [out=down, in=down, looseness=3.5] (0.25,\h) to (0.25,1.5) to  (-0.25, 1.5) to (-0.25, \h) to [out=down, in=down, looseness=3.5](-0.6, \h);
\draw (-0.8,\h) to [out=down, in=down, looseness=1.8] (0.8, \h) to (0.6, \h) to [out=down, in=down, looseness=3.5] (0.25,\h) to (0.25,1.5)  (-0.25, 1.5) to (-0.25, \h) to [out=down, in=down, looseness=3.5](-0.6, \h);
\node[bigbox] at (-1,\h){};
\node[bigbox] at (1,\h){};
\node[triangledown=2] at (0,-1.){};
\end{tz}
~~\superequalseq{eq:implied}~~
\begin{tz}[zx]
\def\h{0.25}
\clip (-2.5, -3 ) rectangle (2.5,1.5);
\path[blueregion] (-2.5, -3) rectangle (2.5, 1.5);
\draw[fill=white, draw] (1.9, 1.5) to (1.9, -1) to [out=down, in=down, looseness=2] (0.2,-1.2) to (-0.2,-1.2) to [out=down, in=down, looseness=2] (-1.9, -1) to (-1.9, 1.5);
\draw (1,-1) to (1,1.5);
\draw (-1,-1) to (-1,1.5);
\draw[redregion] (-0.8,\h) to [out=down, in=down, looseness=1.8] (0.8, \h) to (0.6, \h) to [out=down, in=down, looseness=3.5] (0.25,\h) to (0.25,1.5) to  (-0.25, 1.5) to (-0.25, \h) to [out=down, in=down, looseness=3.5](-0.6, \h);
\draw (-0.8,\h) to [out=down, in=down, looseness=1.8] (0.8, \h) to (0.6, \h) to [out=down, in=down, looseness=3.5] (0.25,\h) to (0.25,1.5)  (-0.25, 1.5) to (-0.25, \h) to [out=down, in=down, looseness=3.5](-0.6, \h);
\draw[fill=white] (0, -2.3) circle (0.35);
\node[bigbox] at (-1,\h){};
\node[bigbox] at (1,\h){};
\node[triangledown=2] at (0,-1.){};
\end{tz}
~~\superequalseq{eq:specialdagger}~~
\begin{tz}[zx]
\def\h{0.25}
\clip (-2.5, -3 ) rectangle (2.5,1.5);
\path[blueregion] (-2.5, -3) rectangle (2.5, 1.5);
\draw[fill=white, draw] (1.9, 1.5) to (1.9, -1) to [out=down, in=down, looseness=2] (0.2,-1.2) to (-0.2,-1.2) to [out=down, in=down, looseness=2] (-1.9, -1) to (-1.9, 1.5);
\draw (1,-1) to (1,1.5);
\draw (-1,-1) to (-1,1.5);
\draw[redregion] (-0.8,\h) to [out=down, in=down, looseness=1.8] (0.8, \h) to (0.6, \h) to [out=down, in=down, looseness=3.5] (0.25,\h) to (0.25,1.5) to  (-0.25, 1.5) to (-0.25, \h) to [out=down, in=down, looseness=3.5](-0.6, \h);
\draw (-0.8,\h) to [out=down, in=down, looseness=1.8] (0.8, \h) to (0.6, \h) to [out=down, in=down, looseness=3.5] (0.25,\h) to (0.25,1.5)  (-0.25, 1.5) to (-0.25, \h) to [out=down, in=down, looseness=3.5](-0.6, \h);
\node[bigbox] at (-1,\h){};
\node[bigbox] at (1,\h){};
\node[triangledown=2] at (0,-1.){};
\end{tz}
\]
We can then prove equation~\eqref{eq:conversemoritaproof} as follows:

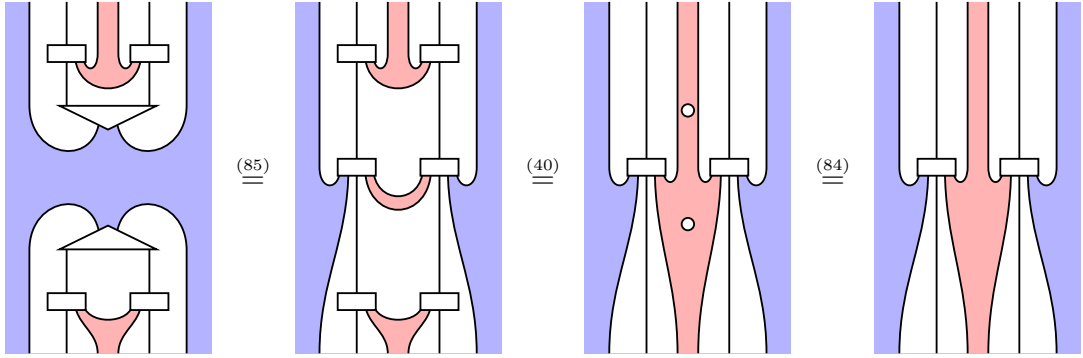
\[\begin{tz}[zx]
\def\h{0.25}
\path[blueregion] (-2.5, -4.25) rectangle (2.5, 4.25);
\draw[fill=white, draw] (1.9, 4.25) to (1.9, 1.75) to [out=down, in=down, looseness=2] (0.2,1.55) to (-0.2,1.55) to [out=down, in=down, looseness=2] (-1.9, 1.75) to (-1.9, 4.25);
\draw (1,1.75) to (1,4.25);
\draw (-1,1.75) to (-1,4.25);
\draw[redregion] (-0.8,\h+2.75) to [out=down, in=down, looseness=1.8] (0.8, \h+2.75) to (0.6, \h+2.75) to [out=down, in=down, looseness=3.5] (0.25,\h+2.75) to (0.25,4.25) to  (-0.25, 4.25) to (-0.25, \h+2.75) to [out=down, in=down, looseness=3.5](-0.6, \h+2.75);
\draw (-0.8,\h+2.75) to [out=down, in=down, looseness=1.8] (0.8, \h+2.75) to (0.6, \h+2.75) to [out=down, in=down, looseness=3.5] (0.25,\h+2.75) to (0.25,4.25)  (-0.25, 4.25) to (-0.25, \h+2.75) to [out=down, in=down, looseness=3.5](-0.6, \h+2.75);
\draw[fill=white, draw] (1.9, -4.25) to (1.9, -1.75) to [out=up, in=up, looseness=2] (0.2,-1.55) to (-0.2,-1.55) to [out=up, in=up, looseness=2] (-1.9, -1.75) to (-1.9, -4.25);
\draw (1,-1.75) to (1,-4.25);
\draw (-1,-1.75) to (-1,-4.25);
\draw[redregion,draw] (-0.25,-4.25) to [out=up, in=down] (-0.8, -\h-2.75) to (-0.6,-\h-2.75) to [out=down, in=down, looseness=1.25] (0.6,-\h-2.75) to (0.8, -\h-2.75) to [out=down, in=up] (0.25,-4.25);
\node[bigbox] at (-1,\h+2.75){};
\node[bigbox] at (1,\h+2.75){};
\node[triangledown=2] at (0,1.75){};
\node[bigbox] at (-1,-\h-2.75){};
\node[bigbox] at (1,-\h-2.75){};
\node[triangleup=2] at (0,-1.75){};
\end{tz}
~~\superequalseq{eq:implied}~~
\begin{tz}[zx]
\def\h{0.25}
\path[blueregion] (-2.5, -4.25) rectangle (2.5, 4.25);
\path[fill=white] (1.9, 4.25) to (1.9, \h) to [out=down, in=down, looseness=3]  (1.4, \h) to (1.2, \h) to [out=down, in=up] (1.9, -4.25) to (-1.9, -4.25) to [out=up, in=down] (-1.2,\h) to (-1.4,\h)to [out=down, in=down, looseness=3] (-1.9, \h) to (-1.9, 4.25);
\draw (1.9, 4.25) to (1.9, \h) to [out=down, in=down, looseness=3]  (1.4, \h) to (1.2, \h) to [out=down, in=up] (1.9, -4.25) 
(-1.9, -4.25) to [out=up, in=down] (-1.2,\h) to (-1.4,\h)to [out=down, in=down, looseness=3] (-1.9, \h) to (-1.9, 4.25);
\draw (1,-4.25) to (1,4.25);
\draw (-1,-4.25) to (-1,4.25);
\draw[redregion] (-0.8,\h+2.75) to [out=down, in=down, looseness=1.8] (0.8, \h+2.75) to (0.6, \h+2.75) to [out=down, in=down, looseness=3.5] (0.25,\h+2.75) to (0.25,4.25) to  (-0.25, 4.25) to (-0.25, \h+2.75) to [out=down, in=down, looseness=3.5](-0.6, \h+2.75);
\draw (-0.8,\h+2.75) to [out=down, in=down, looseness=1.8] (0.8, \h+2.75) to (0.6, \h+2.75) to [out=down, in=down, looseness=3.5] (0.25,\h+2.75) to (0.25,4.25)  (-0.25, 4.25) to (-0.25, \h+2.75) to [out=down, in=down, looseness=3.5](-0.6, \h+2.75);
\draw[redregion, draw] (-0.8,\h) to [out=down, in=down, looseness=2.2] (0.8, \h) to (0.6, \h) to [out=down, in=down, looseness=2] (-0.6, \h);
\draw[redregion,draw] (-0.25,-4.25) to [out=up, in=down] (-0.8, -\h-2.75) to (-0.6,-\h-2.75) to [out=down, in=down, looseness=1.25] (0.6,-\h-2.75) to (0.8, -\h-2.75) to [out=down, in=up] (0.25,-4.25);
\node[bigbox] at (-1,\h+2.75){};
\node[bigbox] at (1,\h+2.75){};
\node[bigbox] at (-1,-\h-2.75){};
\node[bigbox] at (1,-\h-2.75){};
\node[bigbox] at (-1,\h){};
\node[bigbox] at (1,\h){};
\end{tz}
~~\superequalseq{eq:bimodule}~~
\begin{tz}[zx]
\def\h{0.25}
\path[blueregion] (-2.5, -4.25) rectangle (2.5, 4.25);
\path[fill=white] (1.9, 4.25) to (1.9, \h) to [out=down, in=down, looseness=3]  (1.4, \h) to (1.2, \h) to [out=down, in=up] (1.9, -4.25) to (-1.9, -4.25) to [out=up, in=down] (-1.2,\h) to (-1.4,\h)to [out=down, in=down, looseness=3] (-1.9, \h) to (-1.9, 4.25);
\draw (1.9, 4.25) to (1.9, \h) to [out=down, in=down, looseness=3]  (1.4, \h) to (1.2, \h) to [out=down, in=up] (1.9, -4.25) 
(-1.9, -4.25) to [out=up, in=down] (-1.2,\h) to (-1.4,\h)to [out=down, in=down, looseness=3] (-1.9, \h) to (-1.9, 4.25);
\draw (1,-4.25) to (1,4.25);
\draw (-1,-4.25) to (-1,4.25);
\path[redregion] (-0.25, 4.25) to (-0.25,\h) to [out=down, in=down, looseness=3.5] (-0.6, \h) to (-0.8, \h) to [out=down, in=up] (-0.25,-4.25) to (0.25, -4.25) to [out=up, in=down] (0.8,\h) to (0.6,\h) to [out=down, in=down,looseness=3.5] (0.25,\h) to (0.25,4.25);
\draw (-0.25, 4.25) to (-0.25,\h) to [out=down, in=down, looseness=3.5] (-0.6, \h) to (-0.8, \h) to [out=down, in=up] (-0.25,-4.25) 
 (0.25, -4.25) to [out=up, in=down] (0.8,\h) to (0.6,\h) to [out=down, in=down,looseness=3.5] (0.25,\h) to (0.25,4.25);
\draw[fill=white] (0,\h-1.375) circle (0.15);
\draw[fill=white] (0,\h+1.375) circle (0.15);
\node[bigbox] at (-1,\h){};
\node[bigbox] at (1,\h){};
\end{tz}
~~\superequalseq{eq:specialdagger}~~
\begin{tz}[zx]
\def\h{0.25}
\path[blueregion] (-2.5, -4.25) rectangle (2.5, 4.25);
\path[fill=white] (1.9, 4.25) to (1.9, \h) to [out=down, in=down, looseness=3]  (1.4, \h) to (1.2, \h) to [out=down, in=up] (1.9, -4.25) to (-1.9, -4.25) to [out=up, in=down] (-1.2,\h) to (-1.4,\h)to [out=down, in=down, looseness=3] (-1.9, \h) to (-1.9, 4.25);
\draw (1.9, 4.25) to (1.9, \h) to [out=down, in=down, looseness=3]  (1.4, \h) to (1.2, \h) to [out=down, in=up] (1.9, -4.25) 
(-1.9, -4.25) to [out=up, in=down] (-1.2,\h) to (-1.4,\h)to [out=down, in=down, looseness=3] (-1.9, \h) to (-1.9, 4.25);
\draw (1,-4.25) to (1,4.25);
\draw (-1,-4.25) to (-1,4.25);
\path[redregion] (-0.25, 4.25) to (-0.25,\h) to [out=down, in=down, looseness=3.5] (-0.6, \h) to (-0.8, \h) to [out=down, in=up] (-0.25,-4.25) to (0.25, -4.25) to [out=up, in=down] (0.8,\h) to (0.6,\h) to [out=down, in=down,looseness=3.5] (0.25,\h) to (0.25,4.25);
\draw (-0.25, 4.25) to (-0.25,\h) to [out=down, in=down, looseness=3.5] (-0.6, \h) to (-0.8, \h) to [out=down, in=up] (-0.25,-4.25) 
 (0.25, -4.25) to [out=up, in=down] (0.8,\h) to (0.6,\h) to [out=down, in=down,looseness=3.5] (0.25,\h) to (0.25,4.25);
\node[bigbox] at (-1,\h){};
\node[bigbox] at (1,\h){};
\end{tz}
\]
This concludes the proof that $L{\circ}R \cong 1_X$. The converse, $R{\circ}L\cong 1_Y$ can be proven analogously.

2. We claim that if $S:X\to A\to X$ and $P:Y \to A$ are special $1$-morphisms such that $X$ and $Y$ are dagger equivalent, then $S{\circ}\conj{S}$ and $P{\circ}\conj{P}$ are Morita equivalent. Let $R:X\to Y$ and $L:Y\to X$ be an equivalence between $X$ and $Y$ and denote the corresponding 2-isomorphisms as follows:
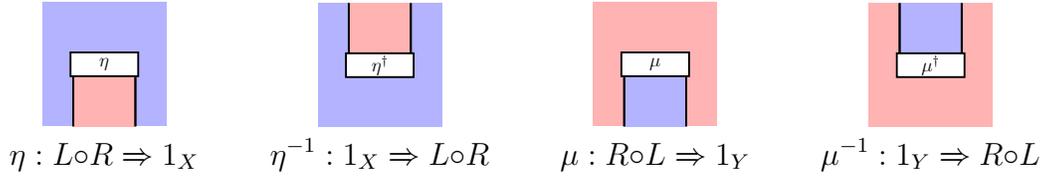
\begin{calign}\nonumber
\begin{tz}[zx]
\path[blueregion] (-1.5,0) rectangle (1.5, 3);
\draw[fill=white, postaction={redregion,draw}] (-0.75, 0) to (-0.75, 1.5) to (0.75, 1.5) to (0.75,0);
\node[widebox={1.9},minimum width=1.5cm] at (0,1.5) {$\eta$};
\end{tz}
&
\begin{tz}[zx,yscale=-1]
\path[blueregion] (-1.5,0) rectangle (1.5, 3);
\draw[fill=white, postaction={redregion,draw}] (-0.75, 0) to (-0.75, 1.5) to (0.75, 1.5) to (0.75,0);
\node[widebox={1.9},minimum width=1.5cm] at (0,1.5) {$\eta^\dagger$};
\end{tz}
&
\begin{tz}[zx]
\path[redregion] (-1.5,0) rectangle (1.5, 3);
\draw[fill=white, postaction={blueregion,draw}] (-0.75, 0) to (-0.75, 1.5) to (0.75, 1.5) to (0.75,0);
\node[widebox={1.9},minimum width=1.5cm] at (0,1.5) {$\mu$};
\end{tz}
&
\begin{tz}[zx,yscale=-1]
\path[redregion] (-1.5,0) rectangle (1.5, 3);
\draw[fill=white, postaction={blueregion,draw}] (-0.75, 0) to (-0.75, 1.5) to (0.75, 1.5) to (0.75,0);
\node[widebox={1.9},minimum width=1.5cm] at (0,1.5) {$\mu^\dagger$};
\end{tz}
\\\nonumber
\eta:L{\circ}R\To 1_X & \eta^{-1}:1_X\To L{\circ}R & \mu:R{\circ}L \To 1_Y & \mu^{-1}:1_Y \To R{\circ}L
\end{calign}
It can then be verified that the 1-morphisms $S{\circ}L{\circ}\conj{P}$ and $P{\circ}R{\circ}\conj{S}$ form bimodules
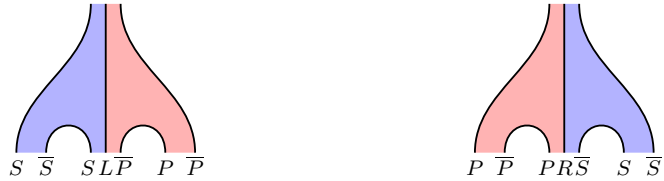
\def\vs{\vphantom{\conj{S}}}
\def\scl{1.2}
\begin{calign}\nonumber
\begin{tz}[zx,scale=\scl,xscale=1.2]
\path[blueregion] (-1.5, 0) to [out=up, in=down] (-0.25,3) to (0,3) to (0,0) to (-0.25,0) to [out=up, in=up, looseness=2.5] (-1,0);
\path[redregion] (1.5,0) to [out=up, in=down] (0.25,3) to (0,3) to (0,0) to (0.25,0) to [out=up, in=up, looseness=2.5] (1,0);
\draw(-1.5, 0) to [out=up, in=down] (-0.25,3);
\draw (0,3) to (0,0);
\draw  (-0.25,0) to [out=up, in=up, looseness=2.5] (-1,0);
\draw(1.5,0) to [out=up, in=down] (0.25,3);
\draw (0.25,0) to [out=up, in=up, looseness=2.5] (1,0);
\node[dimension, below] at (-1.5,0) {$\vs S$};
\node[dimension, below] at (-1,0) {$\conj{S}$};
\node[dimension, below] at (-0.3,0) {$\vs S$};
\node[dimension, below] at (0,0) {$\vs L $};
\node[dimension, below] at (1.5,0) {$\conj{P}$};
\node[dimension, below] at (1,0) {$\vs P$};
\node[dimension, below] at (0.3,0) {$\conj{P}$};
\end{tz}
&
\begin{tz}[zx,scale=\scl,xscale=1.2]
\path[redregion] (-1.5, 0) to [out=up, in=down] (-0.25,3) to (0,3) to (0,0) to (-0.25,0) to [out=up, in=up, looseness=2.5] (-1,0);
\path[blueregion] (1.5,0) to [out=up, in=down] (0.25,3) to (0,3) to (0,0) to (0.25,0) to [out=up, in=up, looseness=2.5] (1,0);
\draw(-1.5, 0) to [out=up, in=down] (-0.25,3);
\draw (0,3) to (0,0);
\draw  (-0.25,0) to [out=up, in=up, looseness=2.5] (-1,0);
\draw(1.5,0) to [out=up, in=down] (0.25,3);
\draw (0.25,0) to [out=up, in=up, looseness=2.5] (1,0);
\node[dimension, below] at (-1.5,0) {$\vs P$};
\node[dimension, below] at (-1,0) {$\conj{P}$};
\node[dimension, below] at (-0.3,0) {$\vs P$};
\node[dimension, below] at (0,0) {$\vs R $};
\node[dimension, below] at (1.5,0) {$\conj{S}$};
\node[dimension, below] at (1,0) {$\vs S$};
\node[dimension, below] at (0.3,0) {$\conj{S}$};
\end{tz}
\end{calign}
for which the following 2-morphisms establish Morita equivalence as in~\eqref{eq:alldataMorita}:
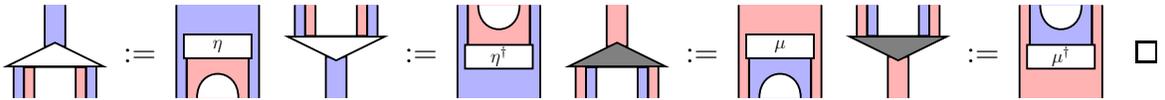
\begin{calign}\nonumber
\begin{tz}[zx]
\path[blueregion] (-1,0) rectangle (1, 0.75);
\draw[fill=white, postaction={redregion,draw}] (-0.75, 0) to (-0.75, 0.75) to (0.75, 0.75) to (0.75,0);
\draw (-1,0) to +(0,0.75);
\draw (1.,0) to +(0,0.75);
\draw[fill=white] (-0.5,0) to (-0.5,0.75) to (0.5,0.75) to (0.5,0); 
\draw[blueregion, draw] (-0.25,2.25) to (-0.25,1.25) to (0.25, 1.25) to (0.25, 2.25);
\node[triangleup=2] at (0,0.75) {};
\end{tz}
~:=~
\begin{tz}[zx]
\path[blueregion] (-1,0) rectangle (1, 2.25);
\draw[fill=white, postaction={redregion,draw}] (-0.75, 0) to (-0.75, 1.25) to (0.75, 1.25) to (0.75,0);
\draw (-1,0) to +(0,2.25);
\draw (1.,0) to +(0,2.25);
\draw[fill=white] (-0.5,0) to [out=up, in=up, looseness=2] (0.5,0);
\node[widebox={1.9},minimum width=1.5cm] at (0,1.25) {$\eta$};
\end{tz}
&
\begin{tz}[zx,yscale=-1]
\path[blueregion] (-1,0) rectangle (1, 0.75);
\draw[fill=white, postaction={redregion,draw}] (-0.75, 0) to (-0.75, 0.75) to (0.75, 0.75) to (0.75,0);
\draw (-1,0) to +(0,0.75);
\draw (1.,0) to +(0,0.75);
\draw[fill=white] (-0.5,0) to (-0.5,0.75) to (0.5,0.75) to (0.5,0); 
\draw[blueregion, draw] (-0.25,2.25) to (-0.25,1.25) to (0.25, 1.25) to (0.25, 2.25);
\node[triangledown=2] at (0,0.75) {};
\end{tz}
~:=~
\begin{tz}[zx,yscale=-1]
\path[blueregion] (-1,0) rectangle (1, 2.25);
\draw[fill=white, postaction={redregion,draw}] (-0.75, 0) to (-0.75, 1.25) to (0.75, 1.25) to (0.75,0);
\draw (-1,0) to +(0,2.25);
\draw (1.,0) to +(0,2.25);
\draw[fill=white] (-0.5,0) to [out=up, in=up, looseness=2] (0.5,0);
\node[widebox={1.9},minimum width=1.5cm] at (0,1.25) {$\eta^\dagger$};
\end{tz}
&
\begin{tz}[zx]
\path[redregion] (-1,0) rectangle (1, 0.75);
\draw[fill=white, postaction={blueregion,draw}] (-0.75, 0) to (-0.75, 0.75) to (0.75, 0.75) to (0.75,0);
\draw (-1,0) to +(0,0.75);
\draw (1.,0) to +(0,0.75);
\draw[fill=white] (-0.5,0) to (-0.5,0.75) to (0.5,0.75) to (0.5,0); 
\draw[redregion, draw] (-0.25,2.25) to (-0.25,1.25) to (0.25, 1.25) to (0.25, 2.25);
\node[triangleup=2,fill=\zxblack] at (0,0.75) {};
\end{tz}
~:=~
\begin{tz}[zx]
\path[redregion] (-1,0) rectangle (1, 2.25);
\draw[fill=white, postaction={blueregion,draw}] (-0.75, 0) to (-0.75, 1.25) to (0.75, 1.25) to (0.75,0);
\draw (-1,0) to +(0,2.25);
\draw (1.,0) to +(0,2.25);
\draw[fill=white] (-0.5,0) to [out=up, in=up, looseness=2] (0.5,0);
\node[widebox={1.9},minimum width=1.5cm] at (0,1.25) {$\mu$};
\end{tz}
&
\begin{tz}[zx,yscale=-1]
\path[redregion] (-1,0) rectangle (1, 0.75);
\draw[fill=white, postaction={blueregion,draw}] (-0.75, 0) to (-0.75, 0.75) to (0.75, 0.75) to (0.75,0);
\draw (-1,0) to +(0,0.75);
\draw (1.,0) to +(0,0.75);
\draw[fill=white] (-0.5,0) to (-0.5,0.75) to (0.5,0.75) to (0.5,0); 
\draw[redregion, draw] (-0.25,2.25) to (-0.25,1.25) to (0.25, 1.25) to (0.25, 2.25);
\node[triangledown=2,fill=\zxblack] at (0,0.75) {};
\end{tz}
~:=~
\begin{tz}[zx,yscale=-1]
\path[redregion] (-1,0) rectangle (1, 2.25);
\draw[fill=white, postaction={blueregion,draw}] (-0.75, 0) to (-0.75, 1.25) to (0.75, 1.25) to (0.75,0);
\draw (-1,0) to +(0,2.25);
\draw (1.,0) to +(0,2.25);
\draw[fill=white] (-0.5,0) to [out=up, in=up, looseness=2] (0.5,0);
\node[widebox={1.9},minimum width=1.5cm] at (0,1.25) {$\mu^\dagger$};
\end{tz}\qedhere
\end{calign}
\end{proof}
\vspace{5pt}
\begin{corollary}\label{cor:maintechnical}
Let $\mathbb{B}$ be a dagger $2$-category in which dagger idempotents split and let $A$ be an object of $\mathbb{B}$. The construction of a special dagger Frobenius monoid in $\mathbb{B}(A,A)$ from a special $1$-morphism into $A$ induces a bijection between the following sets:
\begin{itemize}
\item Dagger equivalence classes of objects $X$ such that there exists a special $1$-morphism \mbox{${S:X \to A}$.}
\item Morita equivalence classes of split special dagger Frobenius monoids $F\in \mathbb{B}(A,A)$.
\end{itemize}
\end{corollary}
\begin{proof} This function is well-defined by the \emph{only if} condition of Theorem~\ref{thm:maintechnical}, surjective by definition and injective by the \emph{if} condition.
\end{proof}

\begin{remark} Statements similar to Theorem~\ref{thm:maintechnical} and Corollary~\ref{cor:maintechnical} hold for non-dagger special Frobenius monoids in non-dagger 2-categories.
\end{remark}

\begin{remark}\label{rem:classify*iso} Corollary~\ref{cor:maintechnical} classifies objects $X$ in $\mathbb{B}$ for which there exists some special $1$-morphism $X\to A$. One may further obtain an explicit classification of these $1$-morphisms as follows:

We say that two special $1$-morphisms $S:X\to A$ and $P:Y \to A$ are \emph{equivalent} if there exists a dagger equivalence $\epsilon:X\to Y$ and a unitary isomorphism $P\circ \epsilon \cong S$. It can then be shown that equivalence classes of special $1$-morphisms into $A$ are in one-to-one correspondence with split special dagger Frobenius monoids in $\mathbb{B}(A,A)$ \emph{up to $*$-isomorphism}.

In other words, $*$-isomorphism classes of Frobenius monoids in $\mathbb{B}(A,A)$ classify $1$-morphisms into $A$, while the coarser Morita equivalence classes just classify objects to which there exists some $1$-morphism into $A$.
\end{remark}
\ignore{
\begin{theorem}\label{thm:maintechnical} Let $\mathbb{B}$ be a dagger-2-category in which all idempotents split and let $A$ be an object of $\mathbb{B}$. Then, there is a bijective correspondence between the following two sets:
\begin{itemize}
\item $\quotient{\left\{\left.\vphantom{\frac{a}{b}}\text{special $\dagger$-Frobenius monoids F in }\mathbb{B}(A,A) ~\right|~\exists \text{ special }S:A\to X, \text{ s.t. }F\cong \conj{S}{\circ}S\right\}}{\text{Morita equivalence}}$
\item $\quotient{\left\{\left.\vphantom{\frac{a}{b}}\text{objects }X\text{ in }\mathbb{B}~\right|~\exists \text{ special }S:A\to X\right\}}{\text{equivalence}}$
\end{itemize}
given by mapping a Frobenius monoid $F$ to the equivalence class of the target of a special $1$-morphism $S$ such that $F\cong \conj{S}{\circ}S$. Here $\cong$ denotes a $*$-isomorphism of Frobenius monoids.
\end{theorem}
}

\end{document}